\documentclass[12pt]{article}
\usepackage[left=1in,top=1in,right=1in,bottom=1in,head=.1in,nofoot]{geometry}

\usepackage{setspace,url,bm,amsmath} %

\usepackage[small]{titlesec}
\titlelabel{\thetitle.\quad}
\titleformat*{\section}{\bf\large\center}
\titlespacing*{\section}
{0pt}{0ex}{0ex}
\titlespacing*{\subsection}
{0pt}{0ex}{0ex}
\titlespacing*{\subsubsection}
{0pt}{0ex}{0ex}

\usepackage{graphicx} %
\usepackage{bbm}
\usepackage{latexsym}
\usepackage{caption}
\usepackage[margin=20pt]{subcaption}
\usepackage{hyperref}
\usepackage{booktabs}
\usepackage{multirow}
\usepackage{rotating}

\usepackage[table]{xcolor}

\newcommand{\GG}[1]{}

\usepackage{amsthm}
\usepackage{amssymb}
\usepackage{amsmath}
\usepackage{color}
\usepackage{mathrsfs}

\usepackage{enumitem}

\usepackage{comment}
\theoremstyle{definition}

\newtheorem{assumption}{Assumption}
\newtheorem*{theorem*}{Theorem}
\newtheorem{theorem}{Theorem}
\newtheorem*{rmk*}{Remark}

\newtheorem{proposition}{Proposition}
\newtheorem{lemma}{Lemma}

\newtheorem{condition}{Condition}

\newtheorem{remark}{Remark}
\newtheorem{corollary}{Corollary}
\newtheorem*{corollary*}{Corollary}

\def\bstrap{\mathsf{b}}

\def\bmX{\bm{X}}
\def\bmU{\bm{U}}
\def\bmO{\bm{O}}
\def\bmW{\bm{W}}
\def\bmtheta{\bm{\theta}}
\def\bmrho{\bm{\rho}}

\usepackage{natbib} %
\bibpunct{(}{)}{;}{a}{}{,} %

\usepackage{etoolbox} %
\apptocmd{\sloppy}{\hbadness 10000\relax}{}{} %

\usepackage{color}
\usepackage{listings}

\DeclareMathOperator*{\argmax}{arg\,max}

\def\ind{\begin{picture}(9,8)
         \put(0,0){\line(1,0){9}}
         \put(3,0){\line(0,1){8}}
         \put(6,0){\line(0,1){8}}
         \end{picture}
        }

\def\Pr{\mathbb{P}}

\def\Cov{\text{Cov}}
\def\sumn{\sum_{i=1}^n}

\def\converged{\rightsquigarrow}

\def\I{\mathbbm{1}}
\def\E{\mathbb{E}}
\def\logit{\text{logit}}

\def\ow{\textup{o}}
\def\low{\text{low}}
\def\up{\text{up}}
\def\m{\text{m}}

\def\OR{\textup{OR}}
\def\odd{\text{odds}}

\def\Bin{\text{Binomial}}

\def\tp{\star}
\def\com{\textup{c}}

\def\LL{\textup{L}}
\def\UU{\textup{U}}

\def\treat{\textup{t}}
\def\control{\textup{c}}

\def\limsup{\overline{\lim}}
\def\liminf{\underline{\lim}}

\RequirePackage[normalem]{ulem}

\usepackage{soul}

\usepackage{subcaption}
\usepackage[textsize=tiny, textwidth = 2cm, shadow]{todonotes}

\usepackage{appendix}

\allowdisplaybreaks

\begin{document}

\onehalfspacing

\title{\bf \Large
Robust Sensitivity Analysis via Augmented Percentile Bootstrap under Simultaneous Violations of Unconfoundedness and Overlap
}
\author{
Han Cui and Xinran Li
\footnote{
Han Cui is Doctoral Candidate, Department of Statistics, University of Illinois at Urbana-Champaign, Champaign, IL, USA (e-mail: \href{mailto:hancui5@illinois.edu}{hancui5@illinois.edu}). 
Xinran Li is Assistant Professor, Department of Statistics, University of Chicago, Chicago, IL, USA (e-mail: \href{mailto:xinranli@uchicago.edu}{xinranli@uchicago.edu}). 
Li was partly supported by the U.S. National Science Foundation (DMS-2400961).
\\
\indent Address for correspondence: Xinran Li, Department of Statistics, University of Chicago, Chicago, IL, 60637, USA. Email: \href{mailto:xinranli@uchicago.edu}{xinranli@uchicago.edu}.
}
}
\date{}
\maketitle

\begin{abstract}
The identification of causal effects in observational studies typically relies on two standard assumptions: unconfoundedness and overlap.  However, both assumptions are often questionable in practice: unconfoundedness is inherently untestable, and overlap may fail in the presence of extreme unmeasured confounding. While various approaches have been developed to address unmeasured confounding and extreme propensity scores separately, few methods accommodate simultaneous violations of both assumptions. In this paper, we propose a sensitivity analysis framework that relaxes both unconfoundedness and overlap, building upon the marginal sensitivity model. Specifically, we allow the bound on unmeasured confounding to hold for only a subset of the population, thereby accommodating heterogeneity in confounding and allowing treatment probabilities to be zero or one. Moreover, unlike prior work, our approach does not require bounded outcomes and focuses on overlap-weighted average treatment effects, which are both practically meaningful and robust to non-overlap. We develop computationally efficient methods to obtain worst-case bounds via linear programming, and introduce a novel augmented percentile bootstrap procedure for statistical inference. This bootstrap method handles parameters defined through over-identified estimating equations involving unobserved variables and may be of independent interest. Our work provides a unified and flexible framework for sensitivity analysis under violations of both unconfoundedness and overlap.    
\end{abstract}

{\bf Keywords}: 
unconfoundedness; overlap; extreme confounding; overlap weights; parameter augmentation

\newpage

\onehalfspacing

\section{Introduction}

Causal inference from observational data is essential in many scientific fields, especially when randomized experiments are infeasible, costly, or time-consuming. The seminal work of \citet{Rubin1983} establishes that, for a binary treatment (active treatment versus control), the average treatment effect can be identified under two key conditions: (i) treatment assignment is conditionally independent of the potential outcomes given observed covariates (commonly referred to as ignorability, unconfoundedness, or selection on observables), and (ii) the conditional probability of receiving the active treatment lies strictly between 0 and 1 (known as the positivity or overlap condition).
When unconfoundedness fails, even after adjusting for observed covariates, differences in outcomes between treated and control groups may reflect not only treatment effects but also the influence of unmeasured confounders. When overlap fails, some units receive treatment (or control) with probability one, making their counterfactual outcomes unidentifiable without additional assumptions. 
These two conditions are standard and, to some extent, necessary for consistent estimation of causal effects in observational studies. 

Despite their importance, both assumptions are often questionable in practice. Unconfoundedness is generally untestable using observed data alone, and violations of this assumption can lead to misleading causal conclusions, as exemplified by Simpson’s paradox. While overlap can be assessed empirically, it becomes increasingly fragile as the number of covariates increases \citep{DAmour2021}. Moreover, in the presence of unmeasured confounding, overlap with respect to the ideal propensity score (conditional on both observed and unobserved covariates) cannot be verified from observed data, and may be violated under extreme unmeasured confounding.
In such cases, we need to rely on extrapolation to estimate counterfactuals for units with extreme propensity scores, making results sensitive to model misspecification.

These challenges have motivated a rich literature on sensitivity analysis, which evaluates the robustness of causal conclusions to potential unmeasured confounding. The seminal work of \citet{cornfield1959smoking} assumes that unconfoundedness would hold if we could observe some additional latent confounders, and quantifies how strong such confounding would need to be to overturn a given conclusion. Subsequent work has proposed various frameworks for quantifying confounding strength. 
For example, \citet{Rosenbaum02a} used the odds ratio of propensity scores between two matched units to quantify confounding strength in matched observational studies, 
and \citet{tan2006} extended it to a superpopulation setting; they both focus on the impact of unmeasured confounding on the treatment assignment. 
\citet{robins2000sen} and \citet{Franks2020} used the contrast between observed outcomes and counterfactuals for units under each treatment group to measure the impact of unmeasured confounding on the potential outcomes. 
\citet{Vanderweele2014} and \citet{ding2016sensitivity} extended Cornfield's inequality by incorporating both the confounder-treatment and confounder-outcome associations, 
and \citet{rubinsen1983}, \citet{Imbens2003sen} and \citet{Zhang2022semi} considered modeling the relation between unmeasured confounding and potential outcomes as well as treatment assignment.

In parallel, a growing literature focuses on inference under violations of the overlap assumption, assuming that unconfoundedness holds.
\citet{crump2009dealing}, \citet{Yang2018trim} and \citet{LiBalancingCov2018} proposed trimming or reweighting to improve overlap. \citet{Hong2019limited} and \citet{Ma2020robust} developed robust inference for the average treatment effect when propensity scores can approach the boundaries. \citet{lee2021bounding} provided bounds for the average treatment effect when overlap fails and outcomes are bounded. However, fewer works address the simultaneous violation of both unconfoundedness and overlap. One exception is \citet{Kennedy2022sen}, who considered situations where a subset of units may be arbitrarily confounded, while the remaining satisfy unconfoundedness.

In this paper, we develop a sensitivity analysis framework that relaxes both the unconfoundedness and overlap assumptions. 
Building upon the marginal sensitivity model proposed by \citet{tan2006}, we assume that the bound on the strength of unmeasured confounding  holds only for a certain proportion of the population, thereby better accommodating the heterogeneity of unmeasured confounding. 
The remaining units may be arbitrarily confounded and may have treatment probabilities of 0 or 1 conditional on both observed and unobserved covariates, which explicitly allows for overlap violations. Our model includes both the marginal sensitivity model and the partially confounded model in \citet{Kennedy2022sen} as special cases. Importantly, we do not require bounded outcomes, an assumption often invoked in prior work on non-overlap \citep[e.g.,][]{lee2021bounding, Kennedy2022sen}. Instead, we focus on overlap-weighted average treatment effects, as recommended in \citet{LiBalancingCov2018}, which are both practically relevant and robust even when overlap fails.

We propose novel statistical methods to implement this sensitivity analysis. The methods are both computationally efficient and statistically valid. Specifically, we show that bounds on treatment effects can be computed via mixed-integer linear programming, or conservatively, via linear programming. For inference, we develop an augmented percentile bootstrap, a general framework for inferring parameters defined by over-identified estimating equations that involve unobserved variables. This approach extends the percentile bootstrap proposed by \citet{zhaosenIPW2019} and requires conceptual augmentation of the data and parameters, which is key to establishing its validity. The proposed augmented percentile bootstrap may be of independent interest for a broader class of inference problems.

Our extension to the marginal sensitivity model is conceptually similar to the recent work of \citet{WL23}, who generalize Rosenbaum’s sensitivity model in matched observational studies to accommodate extreme unmeasured confounding. 
It is worth noting that extreme unmeasured confounding can render both the marginal and Rosenbaum's sensitivity analyses unduly conservative \citep{HS2017, zhang2022}.
However, our work differs substantially from \citet{WL23} in both 
computational and inferential techniques.

\section{Framework and notation}

\subsection{Potential outcome, covariates, and treatment assignment}

We consider an observational study for evaluating the causal effect of a binary treatment, labeled as treatment and control.  
Let $Y(1)$ and $Y(0)$ denote the potential outcome under treatment and control \citep{Rubin:1974}, 
$\bmX$ denote the observed pretreatment covariate vector, 
and $Z$ denote the binary treatment assignment, 
where $Z=1$ if receiving treatment and $0$ otherwise. 
Let $\bmU$ denote some unobserved covariates or confounders, which can be either a scalar or vector. 
We use $(Y(1), Y(0), Z, \bmX, \bmU)$ to denote a representative sample from the superpopulation of interest, 
and assume that we have $n$ independent and identically distributed (i.i.d.) samples from this superpopulation, 
denoted by $\{(Y_i(1),Y_i(0), Z_i, \bmX_i, \bmU_i)\}_{i=1}^n$. 
We emphasize that the unmeasured confounders $\bmU_i$s are unobserved, 
and we can observe only one of the potential outcomes for each unit depending on its treatment assignment, which is 
$Y_i = Z_iY_i(1) + (1-Z_i) Y_i(0)$ for $1\le i \le n$.

From the seminal work of \citet{Rubin1983}, the identifiability of the average treatment effect $\E\{Y(1) - Y(0)\}$ requires the following crucial assumption.  

\begin{assumption}\label{asmp:usual}
\quad
\begin{tabular}[t]{@{}ll}
(i) & (Unconfoundedness or ignorability) $Z \ \ind \ (Y(1), Y(0)) \mid \bmX$. \\
     (ii) & (Overlap or positivity) $0 < \Pr(Z = 1 \mid \bmX) < 1$ almost surely.
\end{tabular}
\end{assumption}

Assumption \ref{asmp:usual} contains two important components: (i) unconfoundedness or ignorability, and (ii) overlap or positivity. 
These two components are both sufficient and necessary for the identification of the average treatment effect. 
However, in practice, both assumptions may fail and are often questioned. 
In such cases, a sensitivity analysis is often invoked to assess the robustness of causal conclusions to the violation of Assumption \ref{asmp:usual}(i).

Specifically, under a sensitivity analysis, we often assume that the unconfoundedness holds given additionally the unobserved confounder $\bmU$ \citep{rubinsen1983}. 
\begin{assumption}\label{asmp:ignorable}
    $Z \ \ind \ (Y(1), Y(0)) \mid \bmX, \bmU$. 
\end{assumption}
We emphasize that
Assumption \ref{asmp:ignorable} always holds when the unmeasured confounder $\bmU$ is set to be the potential outcomes $(Y(1), Y(0))$. 
In other words, it is not an assumption once we set $\bmU = (Y(1), Y(0))$. 
As discussed later, we will impose certain constraints on the strength of the unmeasured confounding $\bmU$, and investigate how sensitive the inference for treatment effect can be. 
Importantly, our sensitivity analysis will also allow for the violation of the overlap assumption, i.e., Assumption \ref{asmp:usual}(ii). 
Specifically, 
let
\begin{align*}
    e_{\tp}(\bm{x}, \bm{u}) \equiv \Pr(Z=1\mid \bmX = \bm{x}, \bmU = \bm{u}) 
\end{align*}
denote the oracle and generally unidentifiable propensity score. 
We will allow $e_{\tp}(\bmX, \bmU)$ to take
$0$ and $1$ values for a non-negligible proportion of units.

\subsection{Robust marginal sensitivity model}\label{sec:rsm}

We consider sensitivity analysis for violations of both unconfoundedness and overlap assumptions. 
Our approach builds upon the marginal sensitivity model proposed by \citet{tan2006}, which has recently received increasing attention \citep[see, e.g.,][]{zhaosenIPW2019, Dorn23, Dorn24, Tan24}; see also the related work by \citet{dan2023interpretable} and \citet{huang2024variance}. 
Specifically, 
the marginal sensitivity model assumes that, for some constant $\Lambda \ge 1$, 
\begin{align}\label{eq:msm}
    \Lambda^{-1}\leq \dfrac{e_{\tp}(\bmX,\bmU)/\{1-e_{\tp}(\bmX,\bmU)\}}{e_{\tp}(\bmX)/\{1-e_{\tp}(\bmX)\}}\leq\Lambda 
    \ \text{ almost surely}, 
\end{align}
where $e_{\tp}(\bm{x}) \equiv \Pr(Z = 1\mid \bmX=\bm{x})$ denotes the observable propensity score. 
That is, the odds ratio between the true and observable propensity scores is bounded between $\Lambda^{-1}$ and $\Lambda$ for all units. 
When $\Lambda = 1$, $e_{\tp}(\bmX,\bmU)$ equals $e_{\tp}(\bmX)$ almost surely, under which there is essentially no unmeasured confounding. 
When $\Lambda$ increases, $e_{\tp}(\bmX,\bmU)$ can deviate further from $e_{\tp}(\bmX)$, under which the treatment effect estimation based on the observed data can be more biased for the corresponding truth. 
In other words, 
the strength of unmeasured confounding is quantified by the odds ratio between the true and observable propensity scores, and the marginal sensitivity model in \eqref{eq:msm} imposes a uniform bound $\Lambda$ on the strength of the unmeasured confounding across the whole population.

In this paper we want to generalize the marginal sensitivity model, to overcome its potential limitations as detailed below. 
First, the marginal sensitivity model in \eqref{eq:msm} assumes that the odds ratio between the true and observable propensity scores are bounded by the same constant $\Lambda$ for all units in the population. 
This may not be realistic unless $\Lambda$ takes a very large value, since the strength of unmeasured confounding is likely to vary across units. 
Specifically, if an investigator suspects that some units may have extreme unmeasured confounding, then the odds ratio between true and observable propensity scores may take very large values for some units. 
This will render the marginal sensitivity model in \eqref{eq:msm} implausible for moderate value of $\Lambda$, and thus deteriorates the power of the sensitivity analysis.
Such an issue has also been noted in the literature. 
For example, 
in the related Rosenbaum's sensitivity analysis for matched observational studies \citep{Rosenbaum02a}, 
\citet{HS2017} commented that some units may be almost certain to smoke due to unmeasured confounding such as high peer pressure, and they propose sensitivity analysis based on the average strength of unmeasured confounding across matched sets; see also \citet{FH2017}. 
Recently, \citet{zhang2022} also extend the marginal sensitivity model in \eqref{eq:msm} to bound only the averages of the squared ratios between the observable and true propensity scores among treated and control units, respectively, i.e., for some $\Lambda \ge 1$, 
\begin{align}\label{eq:sen_l2}
    \E \Big[ \Big\{ \frac{e_{\tp}(\bmX)}{e_{\tp}(\bmX,\bmU)} \Big\}^2 \mid Z=1 \Big] \le \Lambda
    \ \text{ and } \ 
    \E \Big[ \Big\{ \frac{1-e_{\tp}(\bmX)}{1-e_{\tp}(\bmX,\bmU)} \Big\}^2 \mid Z=0 \Big] \le \Lambda. 
\end{align}

Second, 
the marginal sensitivity model in \eqref{eq:msm} does not allow violations of the overlap assumption. 
Specifically, the analyses under the marginal sensitivity model typically impose positivity assumption on the observable propensity scores, under which the true propensity score $e_{\tp}(\bmX, \bmU)$ is strictly bounded between $0$ and $1$ almost surely. 
This is also true for the sensitivity model in \eqref{eq:sen_l2}. 
However, 
as discussed before, the overlap assumption is often questioned in practice. 

To address the above two limitations on existing marginal sensitivity models, 
we propose the following robust marginal sensitivity model, which takes into account the potential heterogeneity of unmeasured confounding strength and allows the true propensity scores to take $0$ and $1$ values for a non-negligible proportion of units.

\begin{assumption}[Robust marginal sensitivity model]\label{asmp:robust_marg_sen} 
Recall the true propensity score $e_{\tp}(\bm{x}, \bm{u})$ and the observable propensity score $e_{\tp}(\bm{x})$. 
For given $\Lambda_1, \Lambda_0 \geq 1$ and $\delta_1, \delta_0 \in [0,1]$, we have 
\begin{align}\label{eq:robust_sen}
    \Pr \Big( \Lambda_z^{-1}\leq \dfrac{e_{\tp}(\bmX,\bmU)/\{1-e_{\tp}(\bmX,\bmU)\}}{e_{\tp}(\bmX)/\{1-e_{\tp}(\bmX)\}}\leq\Lambda_z \mid Z=z \Big) & \geq 1-\delta_z, \quad (z=0,1). 
\end{align}
\end{assumption}

In Assumption \ref{asmp:robust_marg_sen}, we allow different sensitivity parameters $(\Lambda_1, \delta_1)$ and  $(\Lambda_0, \delta_0)$ for treated and control subpopulations; in practice, we can set them to be the same. 
In the supplementary material, we also consider robust marginal sensitivity model that imposes the constraint as in \eqref{eq:robust_sen} but for the whole population including both treated and control units, under which we essentially need to consider a series of models in Assumption \ref{asmp:robust_marg_sen} and the corresponding sensitivity analysis will involve higher computational costs.

Assumption \ref{asmp:robust_marg_sen} is equivalent to the usual marginal sensitivity model in \eqref{eq:msm} when $\delta_1 = \delta_0 = 0$ and $\Lambda_1 = \Lambda_0$ equal a common $\Lambda$, 
but is in general weaker than \eqref{eq:msm}. 
Specifically, in Assumption \ref{asmp:robust_marg_sen}, 
we only assume that the strength of unmeasured confounding is bounded by a certain $\Lambda$ for a certain proportion, say $1-\delta$, of units within the treated or control subpopulations, whereas the remaining $\delta$ proportion of units may have arbitrary unmeasured confounding. 
In particular, for these $\delta$ proportion of units, their true propensity scores $e_{\tp}(\bmX, \bmU)$ are allowed to be either $0$ or $1$, under which the overlap is violated. 
In other words, the robust marginal sensitivity model in Assumption \ref{asmp:robust_marg_sen} allows for the potential violation of both unconfoundedness and overlap assumptions. 
Throughout the paper, we will focus on sensitivity analysis for treatment effects under the robust marginal sensitivity model in Assumption \ref{asmp:robust_marg_sen}.

Moreover, in practice, we suggest try multiple or all values of $\Lambda_1, \Lambda_0, \delta_1$ and $\delta_0$ for the sensitivity model in Assumption \ref{asmp:robust_marg_sen}. 
Similar to \citet{WL23}, we can show that these sensitivity analyses are simultaneously valid across all values of $(\Lambda_1, \Lambda_0, \delta_1, \delta_0)$ and can thus greatly enhance causal evidence. 
Due to length constraints, we briefly discuss this in Section \ref{sec:diss} and relegate the details, including the numerical illustration, to the supplementary material.

\section{Causal identification in the potential absence of overlap}\label{sec:estimand_no_overlap}

In this section, we first consider the scenario where the overlap assumption is violated, but the unconfoundedness assumption holds. 
Nevertheless, in the following discussion throughout this section, 
we will still involve the unmeasured confounder $\bmU$, since it will be useful for our later sensitivity analysis. 
For this section only, we will pretend that the unmeasured confounder $\bmU$ is observed, and discuss the identifiability of general weighted average treatment effects in the potential absence of overlap. This will motivate causal estimands in our later sensitivity analysis under the robust marginal sensitivity model.

Let $\tau(\bm{x}, \bm{u})=\E\{Y(1)-Y(0)|\bmX=\bm{x}, \bmU=\bm{u}\}$ denote the conditional average treatment effect given $(\bmX, \bmU) = (\bm{x}, \bm{u})$,  
and $f(\bm{x}, \bm{u})$ denote the marginal density of the observed and unobserved covariates $(\bmX, \bmU)$ with respect to a certain measure $\mu$. 
For a general weight function $h(\bm{x}, \bm{u}) \ge 0$, 
the corresponding weighted average treatment effect is 
\begin{align}\label{eq:tau_h}
    \tau_h\equiv \dfrac{\int \tau(\bm{x}, \bm{u}) f(\bm{x}, \bm{u}) h(\bm{x}, \bm{u}) d\mu}{\int f(\bm{x}, \bm{u}) h(\bm{x}, \bm{u}) d\mu}
    = \frac{\E\{ h(\bmX, \bmU) \tau(\bmX, \bmU) \}}{\E\{h(\bmX, \bmU)\}}. 
\end{align}
When $h$ is a constant function that takes a positive value, 
$\tau_h$ reduces to the common average treatment effect $\tau_{\com} \equiv \E\{Y(1) - Y(0)\}$. 
Generally, different choices of $h$ correspond to  different weighted populations, 
for which the average treatment effects are evaluated.

Pretending the unmeasured confounder $\bmU$ were observed, below we investigate the identifiability of the general weighted average treatment effect in \eqref{eq:tau_h} when overlap can be violated, i.e.,  the propensity score $e_{\tp}(\bmX,\bmU)$ can be $0$ or $1$ for a non-negligible proportion of units.

\begin{proposition}\label{prop:tau_h_identifiable}
    Suppose that Assumption \ref{asmp:ignorable} holds, $\bmU$ is observed, and the potential outcomes have finite second moments.  
    For any nonnegative square integrable function $h(\cdot)$, the weighted average treatment effect $\tau_h$ is identifiable if and only if $\E[h(\bmX, \bmU)\I\{e_{\tp}(\bmX,\bmU)=z\}]=0$ for $z=0,1$, 
    which is further equivalent to $h(\bmX, \bmU)=\tilde{h}(\bmX, \bmU)e_{\tp}(\bmX,\bmU)\{1-e_{\tp}(\bmX,\bmU)\}$ almost surely for some function $\tilde{h}(\cdot)$.     
\end{proposition}

From Proposition \ref{prop:tau_h_identifiable},
many weights proposed in the literature, including
the overlap weight $h(\bm{x}, \bm{u})=e_{\tp}(\bm{x}, \bm{u})\{1-e_{\tp}(\bm{x}, \bm{u})\}$ \citep{LiBalancingCov2018}, trimming weight $h(\bm{x}, \bm{u})=\I\{\alpha<e_{\tp}(\bm{x}, \bm{u})<1-\alpha\}$ \citep{crump2009dealing} and matching weight $h(\bm{x}, \bm{u})=\min\{e_{\tp}(\bm{x}, \bm{u}), 1-e_{\tp}(\bm{x}, \bm{u}) \}$ \citep{li2013weighting}, 
can lead to identifiable weighted average treatment effects even when the overlap assumption 
fails. 
Note that, when the overlap assumption fails and the weight function $h(\cdot)$ is nonzero for units with $0$ or $1$ propensity scores, 
the corresponding weighted average treatment effect is not only unidentifiable, but also can take any value on the whole real line when the potential outcomes are unbounded. 
This means that 
inference on an unidentifiable weighted average treatment effect 
may be 
uninformative. 
Therefore, 
when the overlap assumption may be violated, as in our robust sensitivity model in Assumption \ref{asmp:robust_marg_sen},  
we suggest to focus on identifiable weighted average treatment effects, such as the overlap-weighted average treatment effect, 
so that we can have more informative inference.

\section{Sensitivity analysis for overlap-weighted treatment  effect}\label{sec:point_sen_ana}

We now study sensitivity analysis under the robust marginal sensitivity model in Assumption \ref{asmp:robust_marg_sen}. 
Throughout the paper, we will focus on the overlap-weighted average treatment effect $\tau_\ow$ defined as in \eqref{eq:tau_h} with $h(\bm{x}, \bm{u}) = e_{\tp}(\bm{x}, \bm{u})\{1-e_{\tp}(\bm{x}, \bm{u})\}$:  
\begin{align}\label{eq:tau_ow}
    \tau_{\ow} & 
    =
    \frac{\E\{ h(\bmX, \bmU) \tau(\bmX, \bmU) \}}{\E\{h(\bmX, \bmU)\}}
    =\dfrac{\E[Z\{1-e_{\tp}(\bmX, \bmU)\}Y]}{\E[Z\{1-e_{\tp}(\bmX, \bmU)\}]}-\dfrac{\E\{(1-Z) e_{\tp}(\bmX, \bmU)Y\}}{\E\{(1-Z)e_{\tp}(\bmX, \bmU)\}}, 
\end{align}
where the equivalent form in \eqref{eq:tau_ow} indicates a natural way to estimate $\tau_{\ow}$ based on the observed data when the true propensity score is known; see, e.g., \citet{LiBalancingCov2018} for its derivation.
As discussed in Section \ref{sec:estimand_no_overlap}, we consider 
the overlap-weighted average treatment effect in \eqref{eq:tau_ow} because it is robust to the violation of overlap, allowing for more informative sensitivity analysis.  
In particular, we do not consider the usual average treatment effect  $\tau_{\com}=\E\{Y(1)-Y(0)\}$, because, with unbounded potential outcomes, it may not have informative bounds under our robust marginal sensitivity model with a positive $\delta_1$ or $\delta_0$.  
In addition, as discussed in \citet{LiBalancingCov2018}, the overlap-weighted average treatment effect $\tau_{\ow}$ can often be of practical interest, since it emphasizes units with substantial probabilities of receiving both treatments. 
For example, in medical studies, such units may represent patients for whom clinical consensus is ambiguous, and for whom research on treatment effects is most needed.

\subsection{Sensitivity bounds of the overlap-weighted average treatment effect}\label{subsec:sen_bound_ow}

In this section, we study how to estimate bounds on the overlap-weighted average treatment effect $\tau_{\ow}$ under our proposed robust marginal sensitivity model in \eqref{eq:robust_sen}.
From 
the equivalent forms of $\tau_{\ow}$ in 
\eqref{eq:tau_ow}, we consider the following H\'ajek-type ``estimator'' for $\tau_{\ow}$:
\begin{align}\label{eq:tau_ow_hat}
    \hat{\tau}_{\ow}=\dfrac{\sum_{i=1}^n Z_i \{1-e_{\tp}(\bmX_i, \bmU_i)\}Y_i}{\sum_{i=1}^n Z_i \{1-e_{\tp}(\bmX_i, \bmU_i)\}}-\dfrac{\sum_{i=1}^n (1-Z_i) e_{\tp}(\bmX_i, \bmU_i)Y_i}{\sum_{i=1}^n (1-Z_i) e_{\tp}(\bmX_i, \bmU_i)}.
\end{align}
Note that 
$\hat{\tau}_{\ow}$
in \eqref{eq:tau_ow_hat} is not computable based on the observed data, since it requires the knowledge of the true propensity scores, which in particular involve 
the unmeasured confounder $\bmU_i$s. 
In the following, we will search the lower and upper bounds of \eqref{eq:tau_ow_hat} under the proposed robust sensitivity model in \eqref{eq:robust_sen} for any prespecified $(\Lambda_1, \Lambda_0, \delta_1, \delta_0)$. 

To facilitate the discussion, we introduce $\psi_{\tp}(\bm{x}, \bm{u}) \equiv \logit\{e_{\tp}(\bm{x}, \bm{u})\}-\logit\{e_{\tp}(\bm{x})\}$ to denote the difference on the logit scale between the true and observable propensity scores. For any function $e(\bm{x})\in [0,1]$ and $\psi(\bm{x}, \bm{u})\in \mathbb{R}$, we define 
\begin{align}\label{eq:e_psi}
    e^{(\psi)}(\bm{x}, \bm{u}) = \logit^{-1}[ \psi(\bm{x}, \bm{u}) + \logit\{e(\bm{x})\} ]
=( 1+\exp[-\psi(\bm{x}, \bm{u}) - \logit\{e(\bm{x})\}] )^{-1}.
\end{align}
By definition, 
$e_{\tp}^{(\psi_{\tp})}(\bm{x}, \bm{u})$ equals the true propensity score $e_{\tp}(\bm{x}, \bm{u})$. 
Let $\hat{e}(\bm{x})$ be an estimated propensity score based on the observed data, 
and $\psi(\bm{x}, \bm{u})$ be a function that needs to be specified.
Similar to \eqref{eq:tau_ow_hat}, we consider the following H\'ajek-type estimator for $\tau_{\ow}$: 
\begin{align}\label{eq:tau_hat_psi}
    \hat{\tau}_{\ow}^{(\psi)}=\dfrac{\sum_{i=1}^n Z_i(1-\hat{e}^{(\psi_i)}_i)Y_i}{\sum_{i=1}^n Z_i (1-\hat{e}^{(\psi_i)}_i)}-\dfrac{\sum_{i=1}^n (1-Z_i) \hat{e}^{(\psi_i)}_i Y_i}{\sum_{i=1}^n (1-Z_i) \hat{e}^{(\psi_i)}_i},
\end{align}
where 
$\hat{e}_i = \hat{e}(\bmX_i)$, 
$\psi_i = \psi(\bmX_i, \bmU_i)$, 
and 
$\hat{e}^{(\psi_i)}_i \equiv \logit^{-1}\{ \psi_i + \logit(\hat{e}_i) \}$ equals $\hat{e}^{(\psi)}(\bmX_i, \bmU_i)$ defined as in \eqref{eq:e_psi}. 
Ideally, we want $\hat{e}(\bm{x})$ to be close to $e_{\tp}(\bm{x})$ and choose $\psi(\cdot)$ to be $\psi_{\tp}(\cdot)$, so that $\hat{e}^{(\psi)}(\bm{x}, \bm{u})$ can be close to the true propensity score $e_\tp(\bm{x}, \bm{u})$. 
However, $\psi_{\tp}(\cdot)$ is generally unknown and not estimable from the observed data. 
Hence, we will instead seek bounds of $\hat{\tau}_{\ow}^{(\psi)}$ over all possible $\psi(\cdot)$ that satisfies certain constraints informed by the true $\psi_{\tp}(\cdot)$. 
Specifically, we will consider two types of constraints: one is from the robust marginal sensitivity model, and the other is from the covariate balancing property of the true propensity score.

We consider first constraints from the robust marginal sensitivity model in Assumption \ref{asmp:robust_marg_sen}, 
which equivalently assumes $\Pr\{ |\psi_{\tp}(\bmX, \bmU)| \le \log(\Lambda_z) \mid Z=z \} \ge 1-\delta_z$ for $z=0,1$. 
This motivates the following constraints for 
the values of $\psi(\cdot)$ evaluated at the $n$ 
samples:
\begin{align}\label{eq:Lambda_delta_psi}
    \frac{1}{n_z}\sum_{i:Z_i=z} \I \{ |\psi_i | \leq \log(\Lambda_z) \} \geq 1-\delta_z,  
    \quad (z=0,1)
\end{align}
where $n_1$ and $n_0$ denote the numbers of treated and control units in the observed samples. 
In \eqref{eq:Lambda_delta_psi}, 
we essentially require that, 
for at least $1-\delta_z$ proportion of the observed treated or control units, the difference on the logit scale between the true and estimated propensity scores is bounded between $-\log(\Lambda_z)$ and $\log(\Lambda_z)$, for $z=0,1$.

We consider then constraints from the covariate balancing property of the true propensity score. For any prespecified scalar or vector function $g(\bm{x})$ of the covariates,  
we have 
\begin{align}\label{eq:cov_bal_property}
    \dfrac{\E[Z\{1-e_{\tp}(\bmX, \bmU)\}g(\bmX)]}{\E[Z\{1-e_{\tp}(\bmX, \bmU)\}]} = \dfrac{\E\{(1-Z) e_{\tp}(\bmX, \bmU)g(\bmX)\}}{\E\{(1-Z) e_{\tp}(\bmX, \bmU)\}}.
\end{align}
This is commonly known as the covariate balancing property, and can be derived by applying the law of total expectation through conditioning on $(\bmX, \bmU)$.
This then motivates 
the following constraint for the values of $\psi(\cdot)$ 
at the 
$n$ 
observed samples:
\begin{align}\label{eq:cov_bal_psi}
    \dfrac{\sum_{i=1}^n Z_i (1-\hat{e}^{(\psi_i)}_i) g(\bmX_i)}{\sum_{i=1}^n Z_i (1-\hat{e}^{(\psi_i)}_i)} = \dfrac{\sum_{i=1}^n (1-Z_i) \hat{e}^{(\psi_i)}_i g(\bmX_i)}{\sum_{i=1}^n (1-Z_i) \hat{e}^{(\psi_i)}_i }.
\end{align}

An estimator for bounds on the overlap-weighted average treatment effect $\tau_{\ow}$ under the proposed robust marginal sensitivity model in Assumption \ref{asmp:robust_marg_sen} is then achieved by optimizing the estimator $\hat{\tau}_{\ow}^{(\psi)}$ in \eqref{eq:tau_hat_psi} 
over all possible $\psi(\cdot)$
satisfying 
the constraints in \eqref{eq:Lambda_delta_psi} and \eqref{eq:cov_bal_psi}. 
If the model for the observable propensity score is correctly specified, then intuitively, as the sample size $n$ increases, the estimated propensity score $\hat{e}(\bm{x})$ will be close to the observable $e_{\tp}(\bm{x})$, and the true $\psi_{\tp}(\cdot)$ will approximately satisfy both constraints in \eqref{eq:Lambda_delta_psi} and \eqref{eq:cov_bal_psi}. Consequently, under correct model specification and as the sample size grows, 
the estimated bounds will approximately cover the true average treatment effect $\tau_{\ow}$. 
In the later Section \ref{sec:ci_for_overlap}, we will take into account the uncertainty in the estimated $\hat{e}(\cdot)$ and in the constraints \eqref{eq:Lambda_delta_psi} and \eqref{eq:cov_bal_psi}, and construct large-sample 
confidence bounds for the overlap-weighted average treatment effect.

\subsection{
Optimization 
for the sensitivity bounds}\label{subsec:double_frac_program}

We now study how to optimize $\hat{\tau}_{\ow}^{(\psi)}$ in 
\eqref{eq:tau_hat_psi} 
subject to the constraints
in \eqref{eq:Lambda_delta_psi} and \eqref{eq:cov_bal_psi}: 
\begin{align}\label{eq:opt_tau_hat_ow}
    \min 
    \text{ or } \max \quad 
    & \hat{\tau}_{\ow}^{(\psi)}=\dfrac{\sum_{i=1}^n Z_i (1-\hat{e}^{(\psi_i)}_i)Y_i}{\sum_{i=1}^n Z_i (1-\hat{e}^{(\psi_i)}_i ) }-\dfrac{\sum_{i=1}^n (1-Z_i) \hat{e}^{(\psi_i)}_i Y_i}{\sum_{i=1}^n (1-Z_i) \hat{e}^{(\psi_i)}_i},
    \\
    \text{subject to} \quad 	&
    \sum_{i:Z_i=z} \I \{ |\psi_i | \leq \log(\Lambda_z) \} \ge n_z
    (1-\delta_z) \  \text{ for } z=0,1, 
    \nonumber
    \\
    & \dfrac{\sum_{i=1}^n Z_i (1-\hat{e}^{(\psi_i)}_i) g(\bmX_i)}{\sum_{i=1}^n Z_i (1-\hat{e}^{(\psi_i)}_i)} = \dfrac{\sum_{i=1}^n (1-Z_i) \hat{e}^{(\psi_i)}_ig(\bmX_i)}{\sum_{i=1}^n (1-Z_i) \hat{e}^{(\psi_i)}_i}. 
    \nonumber
\end{align}
It is worth noting that, for the purpose of optimization, we do not need to treat $\psi(\cdot)$ as a function. 
Instead, we can simply view $\{\psi_i\}_{i=1}^n$, the function values at the observed samples, as $n$ variables that can take any real or infinite values, except that they need to satisfy the constraints in \eqref{eq:opt_tau_hat_ow}. 
This is due to the arbitrariness of the unmeasured confounder $\bmU_i$s. 

\begin{remark}
When $\delta_1 = \delta_0 = 0$ and $\Lambda_1 = \Lambda_0 =1$, there is no unmeasured confounding,  
the $\psi_s$s in \eqref{eq:opt_tau_hat_ow} all become zeros, and the adjusted propensity score $\hat{e}_i^{(\psi_i)}$ becomes the same as the estimated $\hat{e}_i$ for all $i$. 
In this case, however, the optimization in \eqref{eq:opt_tau_hat_ow} may not have a feasible solution, since the covariate balance constraint may not 
hold exactly 
for the estimated propensity scores $\hat{e}_i$s. 
To ensure the existence of a solution when $\delta_1 = \delta_0 = 0$ and $\Lambda_1 = \Lambda_0 =1$, 
we suggest to estimate the propensity score using logistic regression with predictors including $g(\bmX_i)$s. 
Under such a propensity score estimation, 
the covariate balance constraint in \eqref{eq:opt_tau_hat_ow} must hold when $\hat{e}_i^{(\psi_i)}=\hat{e}_i$ for all $i$, 
as shown in \citet[][Theorem 3]{LiBalancingCov2018}.
\end{remark}

Below we further simplify the optimization in \eqref{eq:opt_tau_hat_ow} by introducing additional notation. 
For $1\le i \le n$, 
let $\Delta_i = \I \{ |\psi_i | \leq \log(\Lambda_{Z_i}) \}$ be a dummy variable representing whether the strength of unmeasured confounding for unit $i$ can be bounded by $\log (\Lambda_{Z_i})$. 
We can verify that, 
if $\Delta_i = 1$, then $\psi_i \in [-\log(\Lambda_{Z_i}), \log(\Lambda_{Z_i})]$, 
and $\hat{e}^{(\psi_i)}_i$ can take values between 
\begin{align*}%
    \underline{e}_i = \logit^{-1}\{-\log(\Lambda_{Z_i}) + \logit(\hat{e}_i)\}
    \ \ \text{and} \ \ 
    \overline{e}_i = \logit^{-1}\{\log(\Lambda_{Z_i}) + \logit(\hat{e}_i)\};
\end{align*}
otherwise, 
$\hat{e}^{(\psi_i)}_i$ can take all values in $[0,1]$. 
Without loss of generality, we assume $Z_i=1$ for $1\le i \le n_1$ and $Z_i = 0$ for $n_1+1 \le i \le n$, 
and 
define further 
\begin{align}\label{eq:a_i_low_up}
    \begin{pmatrix}
        a_i^{\low} \\
        a_i^{\up} 
    \end{pmatrix}
    = 
    \begin{pmatrix}
        1 - \overline{e}_i\\
        1 - \underline{e}_i
    \end{pmatrix}
    \text{ for } 1\le i\le n_1, 
    \text{ and }
    \begin{pmatrix}
        a_i^{\low}\\
        a_i^{\up}
    \end{pmatrix}
    = 
    \begin{pmatrix}
        \underline{e}_i\\
        \overline{e}_i 
    \end{pmatrix}
    \text{ for } n_1 < i \le n.
\end{align}
We can then transform the optimization 
in \eqref{eq:opt_tau_hat_ow} into the following one over $\omega_i$s and $\Delta_i$s:  
\begin{align}\label{eq:frac_double}
    \min 
    \text{ or } \max \quad & \dfrac{\sum_{i=1}^{n_1}\omega_iY_i}{\sum_{i=1}^{n_1}\omega_i}-\dfrac{\sum_{i=n_1+1}^n\omega_iY_i}{\sum_{i=n_1+1}^n\omega_i} \\
    \text{subject to} \quad 	&\omega_i \in [\omega_i^{\low}, \omega_i^{\up}]\equiv [a_i^{\low}\Delta_i, 1-(1-a_i^{\up})\Delta_i] \ \text{ for all } 1\le i\le n,  
    \nonumber
    \\
	&\sum_{i=1}^{n_1} \Delta_i\geq  
    \lceil n_1 (1-\delta_1)\rceil, \quad 
    \sum_{i=n_1+1}^n \Delta_i\geq  
    \lceil n_0 (1-\delta_0)\rceil,
    \quad 
    \Delta_i = 0 \text{ or } 1 \text{ for all } i, 
    \nonumber
    \\
    & \dfrac{\sum_{i=1}^{n_1} \omega_i g(\bmX_i)}{\sum_{i=1}^{n_1}\omega_i} = \dfrac{\sum_{i=n_1+1}^{n} \omega_i g(\bmX_i)}{\sum_{i=n_1+1}^n\omega_i}, 
    \nonumber
\end{align}
where $\lceil x \rceil$ denotes the minimum integer that is no less than $x$, and $a_i^{\low}$s and $a_i^{\up}$s are defined in \eqref{eq:a_i_low_up}. 
In \eqref{eq:frac_double}, $\omega_i$ essentially represents $1-\hat{e}^{(\psi_i)}_i$ for a treated unit and $\hat{e}^{(\psi_i)}_i$ for a control unit. 
It is worth noting that \eqref{eq:frac_double} is not a linear fractional programming, even if we ignore the integer constraints on $\Delta_i$s. 
This is because the objective function involves the difference between two fractions, and the optimization for them cannot be separated due to the constraint on covariate balance.  
Fortunately, we can still transform 
\eqref{eq:frac_double} into a 
linear programming problem with integer-type constraints, 
utilizing the classical Charnes-Cooper transformation for linear fractional programming, as detailed in the next subsection.

\subsection{Linear programming for the sensitivity bounds}\label{sec:milp}

Motivated by the Charnes-Cooper transformation, 
we define
\begin{alignat*}{5}
   \bar{\omega}_{i}&=\dfrac{\omega_i}{\sum_{j:Z_j=Z_i}\omega_j},
    &\quad
    \bar{\Delta}_{i}=\frac{\Delta_i}{\sum_{j:Z_j=Z_i}\omega_j},
    \ \ \ 
    \text{ and }
    \ \ \ 
    t_z&=\dfrac{1}{\sum_{j:Z_j=z}\omega_j} \ \text{ for }
    z=0,1.
\end{alignat*}
Then 
\eqref{eq:frac_double} becomes equivalent to the following 
optimization over  $\bar{\omega}_{i}$s, $\bar{\Delta}_{i}$s, $t_1$ and $t_0$:
\begin{align}\label{eq:MILP_sep}
\min 
 \text{ or } \max \quad & \sum_{i=1}^{n_1} \bar{\omega}_{i} Y_i-\sum_{i=n_1+1}^n \bar{\omega}_{i} Y_i \\
\text{subject to} \quad 	&
a_i^{\low} \bar{\Delta}_{i} \leq \bar{\omega}_{i} \leq t_{Z_i}-(1-a_i^{\up})\bar{\Delta}_{i}   \text{ for } 1\le i\le n,
\nonumber
\\
	&\sum_{i=1}^{n_1} \bar{\Delta}_{i}\geq  
    \lceil n_1 (1-\delta_1)\rceil t_1, \quad 
    \sum_{i=n_1+1}^n \bar{\Delta}_{i}\geq  
    \lceil n_0 (1-\delta_0)\rceil t_0,
 \nonumber
 \\
 &
 \bar{\Delta}_{i}\in\{0, t_{Z_i}\} \text{ for } 1\le i \le n, 
 \nonumber
 \\
 & \sum_{i=1}^{n_1} \bar{\omega}_{i} g(\bmX_i) = \sum_{i=n_1+1}^n \bar{\omega}_{i} g(\bmX_i), 
 \nonumber
 \\
& \sum_{i=1}^{n_1} \bar{\omega}_{i}=1, \quad \sum_{i=n_1+1}^n \bar{\omega}_{i}=1,
\nonumber
\quad t_1\ge 0, \quad t_0\ge 0,
\nonumber
\end{align}
where $a_i^{\low}$s and $a_i^{\up}$s are defined as in \eqref{eq:a_i_low_up}. 

We can linearize the constraints that $\bar{\Delta}_{i} \in \{0, t_{Z_i}\}$ for all $i$
using \citet{GloverImproved:1975}’s linearization, so that the optimization in \eqref{eq:MILP_sep} becomes a standard mixed-integer linear programming problem; for conciseness, we relegate the details into the supplementary material. 
Alternatively, 
we can relax the integer-type constraints on 
$\Delta_{i}$s and instead require 
$0\le \Delta_{i}\le t_{Z_i}$ for all $1\le i\le n$. 
Obviously, after the relaxation, the optimization in \eqref{eq:MILP_sep} reduces to a standard linear programming problem, which can be solved in polynomial time. 
Moreover, 
the resulting range of objective values
will cover that from the original problem in \eqref{eq:MILP_sep} with integer-type constraints. 
This is because the variables in the optimization will have a larger feasible region after the relaxation of the integer-type constraints. 
Consequently, the relaxed version of \eqref{eq:MILP_sep} can still lead to valid sensitivity analysis, although being more conservative than the original one. 
Our numerical experiments in the supplementary material suggest that relaxing the integer-type constraints often introduces little conservativeness in the bounds.

\section{Augmented percentile bootstrap for inference under over-identified estimating equations}\label{sec:boot_over_iden}

In Section \ref{sec:point_sen_ana} we focused on estimation of bounds for the overlap-weighted average treatment effect under the robust marginal sensitivity model.
In the remaining of the paper, we will consider its uncertainty quantification. Specifically, we will construct large-sample confidence bounds for the true overlap-weighted average treatment effect. 
In this section, we will first discuss a general strategy for uncertainty quantification in Z-estimation with partial observations and over-identified estimating equations, extending \citet{zhaosenIPW2019}'s percentile bootstrap for the just-identified case. 
We will then apply this general strategy to our robust sensitivity analysis in Section \ref{sec:ci_for_overlap}. 

\subsection{Z-estimation with over-identification and partial observations}\label{eq:z_est_over_unmea}

Assume that we have i.i.d.~samples $\{\bmO_i, \bmW_i^\tp\}_{i=1}^n$ from a superpopulation, 
where, for each $i$, $\bmO_i$ is a vector of observed variables and $\bmW_i^\tp \in \mathbb{R}^{d_1}$ is a vector of unobserved variables. 
We use the $\tp$ in the supscript to emphasize the true values of the unobserved variables; later we will use $\bmW_i$s to denote general values of the unobserved variables that may be different from the truth.  
We are interested in the parameter ${\bmtheta}_0\in \mathbb{R}^{d_0}$ (or its certain transformations), which is defined to be the unique solution of the following set of estimating equations: 
\begin{align}\label{eq:gen_Z}
    \E\{p_{\bmtheta}(\bm\bmO,\bm\bmW^\tp)\}=\bm{0},
\end{align}
where $(\bm\bmO,\bm\bmW^\tp)$ denotes an identically distributed copy of the $(\bmO_i, \bmW_i^\tp)$s, 
and 
$p_{{\bmtheta}}(\cdot) \in \mathbb{R}^{d}$ with $d> d_0$.
Thus, the estimating equations in \eqref{eq:gen_Z} are over-identified, with more equations than the number of parameters. 
In addition, 
we have some constraints on the unobserved $\bmW_i^\tp$s. 
For clarity, we first assume that $\bmW^\tp \in \mathcal{W}_0$ almost surely for some set $\mathcal{W}_0 \subset \mathbb{R}^{d_1}$. 
In the later Section \ref{sec:pred_set_W}, we will extend this framework to accommodate general forms of constraints
from, say, prediction sets for the unobserved variables within the sample.

Pretending that $\bmW_i^\tp$s are known, it is natural to estimate the parameter ${\bmtheta}_0$ by solving the following empirical version of the estimating equations in \eqref{eq:gen_Z}: 
\begin{align}\label{eq:gen_Z_emp}
    \sum_{i=1}^n p_{\bmtheta}(\bmO_i,\bmW_i^\tp)=\bm{0}. 
\end{align}
Since there are more equations than the number of unknown parameters, \eqref{eq:gen_Z_emp} may not have a solution. 
A general way to solving this issue is to reduce the number of estimating equations in \eqref{eq:gen_Z_emp}. For example, we can instead minimize a certain norm of the estimating equations in \eqref{eq:gen_Z_emp}, which is equivalent to finding the root of the corresponding first-order derivative over ${\bmtheta}$, which then gives $d_0$ estimating equations. 
However, in our case of unknown $\bmW_i^\tp$s, we do not want to reduce the number of estimating equations,  
since these additional estimating equations can provide useful information about (or equivalently impose useful constraints on) the unknown $\bmW_i^\tp$s, 
which is helpful for sharpening our inference about ${\bmtheta}_0$. 
For example, \citet{Dorn23} show that, with additional constraints imposed by the balancing property of the true propensity score with respect to some conditional quantiles of the potential outcomes given observed covariates, the analysis under the usual marginal sensitivity model in \eqref{eq:msm}
can become sharp, 
thereby substantially improving the inference for treatment effects.

We will be particularly interested in a scalar transformation of the parameter ${\bmtheta}_0$, denoted by $H({\bmtheta}_0)$. 
An intuitive estimator for, say, an upper bound of $H({\bmtheta}_0)$ is 
\begin{align}\label{eq:z_est_point_upper_bound}
    \sup \quad & H({\bmtheta}) 
    \\
    \text{subject to} \quad & 
    \sum_{i=1}^n p_{\bmtheta}(\bmO_i, \bmW_i)=\bm{0},
    \nonumber
    \\
    & \bmW_i \in \mathcal{W}_0, \ i=1,2,\ldots, n,
    \nonumber
\end{align}
where we optimize over ${\bmtheta}$ and $\bmW_i$s. 
The above form of estimation is also what we considered in Section \ref{sec:point_sen_ana}; see the later Section \ref{sec:ci_for_overlap} for details. 
Below we focus on constructing large-sample valid confidence bounds for $H({\bmtheta}_0)$ with guaranteed coverage probabilities.

\subsection{Parameter and data augmentation, and regularity conditions}
In the just-identified case where the number of estimating equations equals the number of unknown parameters, \citet{zhaosenIPW2019} proposed percentile bootstrap to construct confidence bounds for $H({\bmtheta}_0)$. 
Below we extend this strategy to address the over-identified case. As previously discussed, having more estimating equations than necessary can be beneficial for further constraining the unobserved variables. The key to this extension is to augmenting the parameter ${\bmtheta}$ with an additional  
auxiliary parameter denoted by, say, ${\bmrho}$.
This auxiliary parameter ${\bmrho}$ plays a crucial role in justifying the validity of the inference; however, it will not appear in the actual implementation of our procedure, as detailed shortly. This is why we call it an auxiliary parameter.

To be more specific, we consider the following set of estimating equations for the augmented parameter vector $({\bmtheta}, {\bmrho})$: 
\begin{align}\label{eq:gen_Z_aug_para}
    \E\{p_{\bmtheta}(\bmO,\bmW^\tp + \bm{R}^\tp {\bmrho} )\}=\bm{0}
\end{align}
where $\bm{R}^\tp \in \mathbb{R}^{d_1 \times (d-d_0)}$ can be any random matrix that can depend on $\bmO$ and $\bmW^\tp$, and ${\bmrho} \in \mathbb{R}^{d - d_0}$ is the auxiliary parameter.
We can also understand $\bm{R}^\tp$ as augmented data, which again will be auxiliary since our actual procedure discussed shortly will not depend on its explicit construction. 
Let ${\bmrho}_0 = \bm{0}$.
Then, obviously, $({\bmtheta}_0, {\bmrho}_0)$ is a solution to the estimation equation in \eqref{eq:gen_Z_aug_para} for $({\bmtheta}, {\bmrho})$. 
Importantly, with the augmented auxiliary parameter, the estimating equations in \eqref{eq:gen_Z_aug_para} become just-identified, since the number of equations equals the number of parameters in $({\bmtheta}, {\bmrho})$. 

We then impose assumptions on the asymptotic normality and bootstrap validity of the Z-estimation for $({\bmtheta}_0, {\bmrho}_0)$ based on \eqref{eq:gen_Z_aug_para}, pretending that the values of $\bmW^\tp$ and $\bm{R}^\tp$ are known for each observed sample.  
Specifically, 
let $(\hat{{\bmtheta}}, \hat{{\bmrho}})$ be the root to (or equivalently the Z-estimator based on) the following empirical estimating equations: 
\begin{align}\label{eq:gen_Z_aug_para_emp}
\sum_{i=1}^n p_{\bmtheta}(\bmO_i, \bmW_i^\tp + \bm{R}^\tp_i {\bmrho} ) =\bm{0},
\end{align}
and $(\hat{{\bmtheta}}_{\bstrap}, \hat{{\bmrho}}_{\bstrap})$ follow the bootstrap distribution of the above Z-estimator, i.e., 
$(\hat{{\bmtheta}}_{\bstrap}, \hat{{\bmrho}}_{\bstrap})$ is the root to the following estimating equations:
\begin{align}\label{eq:gen_Z_aug_para_emp_boot}
    \sum_{i=1}^n k_{\bstrap i} p_{\bmtheta}(\bmO_{i}, \bmW_i^\tp + \bm{R}^\tp_i {\bmrho} ) =\bm{0},
\end{align}
where $(k_{\bstrap 1}, k_{\bstrap 2}, \ldots, k_{\bstrap n}) \sim \text{Multinomial} (n; n^{-1}, n^{-1}, \ldots, n^{-1})$. 

\begin{assumption}[Validity of the standard nonparametric bootstrap]\label{asmp:valid_bootstrap_z_est}
 For some random matrix $\bm{R}^\tp$ with finite $2+\nu$ moment for some $\nu>0$, in the sense that 
 $\E\{ \| \bm{R}^\tp \|^{2+\nu} \} < \infty$,  
 the Z-estimator $(\hat{{\bmtheta}}, \hat{{\bmrho}})$ from \eqref{eq:gen_Z_aug_para_emp} and its bootstrap version  $(\hat{{\bmtheta}}_{\bstrap}, \hat{{\bmrho}}_{\bstrap})$ from \eqref{eq:gen_Z_aug_para_emp_boot} have the following asymptotic properties as $n\rightarrow \infty$: 
    \begin{align}\label{eq:norm_z_est}
        \sqrt{n}
        \begin{pmatrix}
            \hat{{\bmtheta}} - {\bmtheta}_0\\
            \hat{{\bmrho}} - {\bmrho}_0
        \end{pmatrix}
        \converged 
        \mathcal{N}\left(\bm{0}, \bm{\Sigma}\right),
        \quad 
        \sqrt{n}
        \begin{pmatrix}
        \hat{{\bmtheta}}_{\mathsf{b}} - \hat{{\bmtheta}}\\
            \hat{{\bmrho}}_{\mathsf{b}} - \hat{{\bmrho}}
        \end{pmatrix}
        \overset{\Pr}{\converged} \mathcal{N}\left(\bm{0}, \bm{\Sigma}\right), 
    \end{align}
    where $\converged$ means convergence in distribution,  $\mathcal{N}(\bm{0}, \bm{\Sigma})$ denotes a Gaussian distribution with mean zero and covariance matrix $\bm{\Sigma}$, 
    and $\overset{\Pr}{\converged}$ means that the weak convergence of the bootstrap distribution holds in probability. 
    In addition, recall that ${\bmrho}_0 = \bm{0}$. 
\end{assumption}

Assumption \ref{asmp:valid_bootstrap_z_est} is a standard result from the theory of Z-estimation; see, e.g., \citet{Wellner1996bootstrapping} and \citet{KosorokEmp:2006}. 
In Assumption \ref{asmp:valid_bootstrap_z_est}, we mainly require the existence of an $\bm{R}^\tp$ with a certain bounded moment such that the theory of Z-estimation applies. We do not require the explicit construction of $\bm{R}^\tp$ in our bootstrap procedure, as described shortly. 
In Section \ref{sec:ci_for_overlap}, we provide a construction of $\bm{R}^\tp$ along with explicit conditions under which Assumption \ref{asmp:valid_bootstrap_z_est} holds in the context of our robust sensitivity analysis.

\subsection{Augmented percentile bootstrap confidence bounds}

Below we introduce our augmented percentile bootstrap in the case of over-identification. To facilitate understanding, we proceed step by step. We first discuss the standard bootstrap, pretending that $\bmW_i^\tp$s and $\bm{R}_i^\tp$s are known. 
We then discuss bootstrap inference pretending that $\bm{R}_i^\tp$s and the estimated auxiliary parameter (based on both original and bootstrap samples) are known. 
Finally, we introduce our augmented percentile bootstrap. 
For conciseness and without loss of generality, 
we focus on constructing a $1-\alpha$ upper confidence bound for the scalar parameter of interest $H({\bmtheta}_0)$, where $\alpha\in (0,1)$ specifies the desired coverage level. 
The lower confidence bounds and consequently two-sided confidence bounds can be constructed analogously. 
Furthermore, 
we assume that $H(\cdot)$ is differentiable at ${\bmtheta}_0$ so that, by the Delta method, $H(\hat{{\bmtheta}})$ is asymptotically Gaussian under Assumption \ref{asmp:valid_bootstrap_z_est}, and  that the asymptotic variance of $\sqrt{n}\{H(\hat{{\bmtheta}}) - H({\bmtheta}_0)\}$ is positive. 

First, we pretend that we know both $\bmW_i^\tp$s and $\bm{R}_i^\tp$s. In this case, we can apply the standard bootstrap and use $Q_{1-\alpha} \{ H(\hat{{\bmtheta}}_{\bstrap}) \}$, the $(1-\alpha)$th quantile of the distribution of $H(\hat{{\bmtheta}}_{\bstrap})$, as an $1-\alpha$ upper confidence bound for $H({\bmtheta}_0)$. 
Recall that $(\hat{{\bmtheta}}_{\bstrap}, \hat{{\bmrho}}_{\bstrap})$ is the solution for \eqref{eq:gen_Z_aug_para_emp_boot} and follows the bootstrapped distribution of the Z-estimator based on \eqref{eq:gen_Z_aug_para_emp}. 
Under Assumption \ref{asmp:valid_bootstrap_z_est}, we can show that 
$\limsup_{n\rightarrow \infty} \Pr[ H({\bmtheta}_0) > Q_{1-\alpha} \{ H(\hat{{\bmtheta}}_{\bstrap}) \} ] = \alpha$, i.e., $Q_{1-\alpha} \{ H(\hat{{\bmtheta}}_{\bstrap}) \}$ is an asymptotically valid $1-\alpha$ upper confidence bound for $H({\bmtheta}_0)$. 

Second, we pretend that $\bm{R}_i^\tp$s, the estimator for the auxiliary parameter $\hat{{\bmrho}}$, and its bootstrapped version $\hat{{\bmrho}}_{\bstrap}$ are all known.  
We can verify that $H(\hat{{\bmtheta}}_{\bstrap})$ must be upper bounded by 
\begin{align}\label{eq:bootstrap_R_rho}
    \check{\theta}_{H, \bstrap} = \ & \sup H({\bmtheta}) 
    \\
    \text{subject to} \quad & 
    \sum_{i=1}^n k_{\bstrap i} p_{\bmtheta}(\bmO_{i}, \bmW_i + \bm{R}^\tp_i \hat{{\bmrho}}_{\bstrap} ) =\bm{0},
    \nonumber
    \\
    & \bmW_i \in \mathcal{W}_0, \ i=1,2,\ldots, n.
    \nonumber
\end{align}
where we optimize over $({\bmtheta}, \bmW_1, \ldots, \bmW_n)$. 
This is because $(\hat{{\bmtheta}}_{\bstrap}, \bmW_1^\tp, \ldots, \bmW_n^\tp)$ is in the feasible region of \eqref{eq:bootstrap_R_rho}. 
We can then use  $Q_{1-\alpha} ( \check{\theta}_{H, \bstrap} )$ as a $1-\alpha$ upper confidence bound for $H({\bmtheta}_0)$, and its asymptotic validity follows immediately from the validity of the standard bootstrap and the  fact that $\check{\theta}_{H, \bstrap} \ge H(\hat{{\bmtheta}}_\bstrap)$. 
In the just-identified case, we do not need to introduce $\bm{R}^\tp$ and ${\bmrho}$, and we can thus remove $\bm{R}^\tp_i \hat{{\bmrho}}_{\bstrap}$s in \eqref{eq:bootstrap_R_rho}, under which $Q_{1-\alpha} ( \check{\theta}_{H, \bstrap} )$ reduces to the percentile bootstrap confidence bound in \citet{zhaosenIPW2019}. 
Note that, in general, both $\bm{R}^\tp_i$s and $\hat{{\bmrho}}_{\bstrap}$ are unknown and can depend on the unobserved variables $\bmW_i^\tp$s. 

Finally, we propose our bootstrap confidence bounds that do not require the knowledge of unobserved information. Note that the two bootstrap methods discussed above are generally infeasible to implement in practice, due to their dependence on unobserved information. 
The construction of our confidence bounds relies on the fact that, although the exact value of $\hat{{\bmrho}}_{\bstrap}$ is unknown, it is close to zero and more precisely of order $O_{\Pr}(n^{-1/2})$. 
This follows from the true value of ${\bmrho}$ being ${\bmrho}_0 = \bm{0}$, along with the asymptotic normality of the standard Z-estimation in Assumption \ref{asmp:valid_bootstrap_z_est}. 
Moreover, $\bm{R}^\tp_i \hat{{\bmrho}}_{\bstrap}$s will also be close to zero since $\bm{R}^\tp_i$s have bounded moments by construction. 
These motivate us to consider the following optimization based on the bootstrapped samples: 
\begin{align}\label{eq:bootstrap_R_epsilon}
    \tilde{\theta}_{H, \bstrap} = \ & \sup H({\bmtheta}) 
    \\
    \text{subject to} \quad & 
    \sum_{i=1}^n k_{\bstrap i} p_{\bmtheta}(\bmO_{i}, \bmW_i) =\bm{0},
    \nonumber
    \\
    & \bmW_i \in \mathcal{W}_0^\varepsilon, \ i=1,2,\ldots, n,
    \nonumber
\end{align}
where $\varepsilon>0$ is a positive constant and  $\mathcal{W}_0^\varepsilon = \{\bmW: \|\bmW - \bmW_0\| \le \varepsilon \text{ for some } \bmW_0 \in \mathcal{W}_0 \}$ is the $\varepsilon$-neighborhood of $\mathcal{W}_0$. 
We then propose to use $Q_{1-\alpha} ( \tilde{\theta}_{H, \bstrap} )$ as a $1-\alpha$ upper confidence bound for $H({\bmtheta}_0)$. 
We now give some intuition for the validity of the proposed confidence bound. 
We essentially consider $\bmW_i + \bm{R}^\tp_i \hat{{\bmrho}}_{\bstrap}$ in \eqref{eq:bootstrap_R_rho} as a single variable in the optimization, and it will be in $\mathcal{W}_0^\varepsilon$ when $\bmW_i \in \mathcal{W}_0$ and $\bm{R}^\tp_i  \hat{{\bmrho}}_{\bstrap}$ is sufficiently small, where the latter holds with high probability. 
Our proposed bootstrap confidence bound is asymptotically valid for any positive $\varepsilon$. 
We summarize the results in the following theorem.

\begin{theorem}\label{thm:bootstrap_over_identify}
    Assume that Assumption \ref{asmp:valid_bootstrap_z_est} holds,
    $H(\cdot)$ is differentiable at ${\bmtheta}_0$, 
    and $\dot{H}({\bmtheta}_0) \bm{\Sigma}_{{\bmtheta}{\bmtheta}} \dot{H}({\bmtheta}_0)^\top$ $>0$, where $\dot{H}(\cdot)$ denotes the derivative of $H(\cdot)$ and $\bm{\Sigma}_{{\bmtheta}{\bmtheta}}$ is the submatrix of $\bm{\Sigma}$ in \eqref{eq:norm_z_est} corresponding to the parameter ${\bmtheta}$. 
    For any $\varepsilon > 0$, 
    $Q_{1-\alpha} ( \tilde{\theta}_{H, \bstrap} )$ is an asymptotic $1-\alpha$ upper confidence bound for $H({\bmtheta}_0)$, in the sense that 
    $\limsup_{n\rightarrow \infty} \Pr\{ H({\bmtheta}_0) > Q_{1-\alpha} ( \tilde{\theta}_{H, \bstrap} ) \} \le \alpha $.
\end{theorem}

Below we give two remarks about Theorem \ref{thm:bootstrap_over_identify}. 
First, 
we can replace $\varepsilon$ in \eqref{eq:bootstrap_R_epsilon} by $\varepsilon_n$ that can vanish with the sample size $n$, with rate depending on the order of $\max_{1\le i\le n} \|\bm{R}_i^\tp\|$;  
see the supplementary material for details.
For example, when $\bm{R}^\tp$ is bounded, we can choose $\varepsilon_n$ to be any sequence such that $\varepsilon_n \gg n^{-1/2}$.  
Second, the bootstrapped $\tilde{\theta}_{H, \bstrap}$ in \eqref{eq:bootstrap_R_epsilon} has almost the same form as the intuitive estimator in \eqref{eq:z_est_point_upper_bound} for an upper bound of $H({\bmtheta}_0)$, except that $\tilde{\theta}_{H, \bstrap}$ uses the bootstrapped samples and involves an $\varepsilon$-neighborhood of $\mathcal{W}_0$. 
Therefore, we are essentially using the bootstrapped distribution of the intuitive estimator in \eqref{eq:z_est_point_upper_bound} to construct confidence bounds for $H({\bmtheta}_0)$. 
Different from the standard bootstrap procedure, 
we slightly relax the constraints with an $\varepsilon>0$, to accommodate the over-identified estimating equations. 
In the next subsection, we extend our bootstrap procedure to accommodate more general forms of constraints on the unobserved variables. 

\begin{remark}
    The augmentation of the feasible set $\mathcal{W}_0$ is in some sense necessary. For example, in the case where $\mathcal{W}_0$ contains only one element, $\tilde{\theta}_{H, \bstrap}$ in \eqref{eq:bootstrap_R_epsilon} may be ill-defined when $\varepsilon=0$, as the feasible region for the optimization in \eqref{eq:bootstrap_R_epsilon} could be an empty set. 
    This is because, when $\mathcal{W}_0$ is a singleton and $\varepsilon = 0$, the only variable being optimized is ${\bmtheta}$, whose dimension is smaller than the number of equality constraints imposed by the estimating equations.
\end{remark}

\subsection{Extension to more general forms of constraints on the unobserved variables}\label{sec:pred_set_W}

We now consider a general form of constraints on the unobserved variables. 
Specifically, we will work with a general prediction set for the unobserved variables within the samples, which has a guaranteed coverage probability at some prespecified level asymptotically.
The use of a prediction set rather than $\mathcal{W}_0$ as in the previous subsections can be preferable when the unobserved variables can take any value in $\mathbb{R}^{d_1}$
(i.e., $\mathcal{W}_0 = \mathbb{R}^{d_1}$) 
but is likely to be in some smaller set. 
In addition, the prediction set can incorporate ``dependence'' among these unobserved variables within the samples. For example, under our robust marginal sensitivity model in Assumption \ref{asmp:robust_marg_sen}, the number of units in the treated or control group with unmeasured confounding strength greater than a certain $\Lambda$ is likely to be bounded. 
The construction of such prediction sets is context-specific. 
In the following, we work with a general prediction set $\mathcal{W}_n \subset (\mathbb{R}^{d_1})^{\otimes n}$ for the $n$ unobserved variables within the samples, 
where $\Pr\{ (\bmW_1^\tp, \ldots, \bmW_n^\tp) \notin \mathcal{W}_n \}$ is asymptotically bounded from the above by some prespecified level.

Let $\mathcal{W}_n^\varepsilon$ be an $\varepsilon$-neighborhood of the set $\mathcal{W}_n$ defined as:
\begin{align*}
    \mathcal{W}_n^\varepsilon \equiv 
    \{
    (\bmW_1, \ldots, \bmW_n) \in (\mathbb{R}^{d_1})^{\otimes n}: 
    \max_{1\le i \le n} \|\bmW_i - \bmW_i'\| \le \varepsilon \text{ for some } (\bmW_1', \ldots, \bmW_n') \in \mathcal{W}_n 
    \}. 
\end{align*}
Similar to \eqref{eq:bootstrap_R_epsilon}, we consider the following optimization based on the bootstrapped samples: 
\begin{align}\label{eq:bootstrap_epsilon_pred_set}
    \hat{\theta}_{H, \bstrap} = \ & \sup H({\bmtheta}) 
    \\
    \text{subject to} \quad & 
    \sum_{i=1}^n k_{\bstrap i} p_{\bmtheta}(\bmO_{i}, \bmW_i) =\bm{0},
    \nonumber
    \\
    & (\bmW_1, \bmW_2, \ldots, \bmW_n) \in \mathcal{W}_n^\varepsilon.
    \nonumber
\end{align}
It is worth emphasizing that the first constraint in \eqref{eq:bootstrap_epsilon_pred_set} is based on estimating equations evaluated at the bootstrapped samples, whereas the second constraint is based on the original samples (rather than their bootstrapped versions). 
In other words, when we perform the bootstrap, we use only bootstrapped samples for the first constraint from the estimating equations, and maintain the original samples for the second constraint, which is often from our imposed assumptions on the strength of unmeasured confounding in the context of sensitivity analyses.
This also explains why we use the $k_{\bstrap i}$s to represent the bootstrapped samples, which highlights such differences in these two constraints and also facilitate the implementation of this optimization. 
Such a distinction between the two constraints in \eqref{eq:bootstrap_epsilon_pred_set} is particularly important for our robust marginal sensitivity model, 
as opposed to the usual marginal sensitivity model; see Section \ref{sec:ci_for_overlap} for details.

\begin{theorem}\label{thm:bootstrap_over_identify_pred_set}
    Under the same setting as in Theorem \ref{thm:bootstrap_over_identify},  
    for any $\varepsilon > 0$, 
    \begin{align*}
        \limsup_{n\rightarrow \infty} \Pr\{ H({\bmtheta}_0) > Q_{1-\alpha} ( \hat{\theta}_{H, \bstrap} ), \  (\bmW_1^\tp, \bmW_2^\tp, \ldots, \bmW_n^\tp) \in \mathcal{W}_n \} \le \alpha, 
    \end{align*}
    where $\hat{\theta}_{H, \bstrap}$ is defined in \eqref{eq:bootstrap_epsilon_pred_set}.
    Consequently, if further $\mathcal{W}_n$ is an asymptotic $1-\zeta$ prediction set for $(\bmW_1^\tp, \bmW_2^\tp, \ldots, \bmW_n^\tp)$, then $Q_{1-\alpha} ( \hat{\theta}_{H, \bstrap} )$ is an asymptotic $1-\alpha-\zeta$ upper confidence bound for $H({\bmtheta}_0)$, i.e., 
    \begin{align*}
        \limsup_{n\rightarrow \infty} \Pr\{ H({\bmtheta}_0) > Q_{1-\alpha} ( \hat{\theta}_{H, \bstrap} ) \}
        \le 
        \alpha + 
        \limsup_{n\rightarrow \infty} \Pr\{ (\bmW_1^\tp, \bmW_2^\tp, \ldots, \bmW_n^\tp) \notin \mathcal{W}_n \}
        \le \alpha + \zeta. 
    \end{align*}
\end{theorem}

\section{Confidence intervals for the overlap-weighted treatment effect}\label{sec:ci_for_overlap}

In this section, we extend our estimated bounds in Section \ref{sec:point_sen_ana} to confidence bounds for the range of possible treatment effects under the robust marginal sensitivity model.
We consider parametric models, such as the logistic regression model, for fitting the propensity score. 
We first extend the robust marginal sensitivity model in Assumption \ref{asmp:robust_marg_sen} to accommodate model misspecification. 
We then use the augmented percentile bootstrap developed in Section \ref{sec:boot_over_iden} to construct confidence bounds for the overlap-weighted average treatment effect $\tau_\ow$. 

\subsection{A modified sensitivity analysis model allowing model misspecification}

Let $e(\bm{x};{\bm{\beta}})$ denote the model for the observable propensity score $\Pr(Z=1\mid \bmX=\bm{x})$, 
and $e_\m(\bm{x}) = e(\bm{x}; {\bm{\beta}}_0)$ denote the estimable propensity score under our model, where ${\bm{\beta}}_0 = \argmax_{\bm{\beta}} \E[ Z \log\{ e(\bmX;{\bm{\beta}}) \} + (1-Z) \log \{ 1 - e(\bmX;{\bm{\beta}}) \} ]$ is, under certain regularity conditions, the probability limit of the maximum likelihood estimator for ${\bm{\beta}}$ as $n\rightarrow \infty$.  
We can view $e_\m(\bm{x})$ as the best 
approximation of the observable propensity score $e_{\tp}(\bm{x})$ under our specification of the propensity score model. When the propensity score model is correctly specified,  $e_\m(\bm{x})$ becomes the same as $e_{\tp}(\bm{x})$. 
To allow model misspecification for the observable propensity score, we introduce the following modified robust marginal sensitivity model, which is stated in terms of the probability limit of our propensity score estimator. 

\begin{assumption}[Robust marginal sensitivity model under propensity score specification]\label{asmp:robust_sen_spec}
Recall the true propensity score $e_{\tp}(\bm{x}, \bm{u})$, and  define $e_{\m}(\bm{x})$ as an estimable propensity score based on certain model assumptions. 
For given parameters $\Lambda_1, \Lambda_0 \geq 1$ and $0\leq\delta_1, \delta_0\leq 1$, 
\begin{align}\label{eq:para_robust_sen}
	\Pr\left( \Lambda_z^{-1}\leq \dfrac{e_{\tp}(\bmX, \bmU)/\{1-e_{\tp}(\bmX, \bmU)\}}{e_\m(\bmX)/\{1-e_\m(\bmX)\}}\leq\Lambda_z \mid Z=z\right) \geq 1-\delta_z, \quad (z=0,1). 
\end{align}	
\end{assumption}

Obviously, when the propensity score model is correctly specified, in the sense that $e_\m(\bm{x}) = e_\tp(\bm{x})$, the sensitivity model in the above Assumption \ref{asmp:robust_sen_spec} is equivalent to that in Assumption \ref{asmp:robust_marg_sen}. 
In general, the sensitivity model in Assumption \ref{asmp:robust_sen_spec} involves not only the strength of unmeasured confounding, but also the degree of model misspecification for the observable propensity score $e_{\tp}(\cdot)$. 
In the following, we will focus on the sensitivity model in Assumption \ref{asmp:robust_sen_spec} instead of  Assumption \ref{asmp:robust_marg_sen}; this is the same as  
\citet{zhaosenIPW2019}.

\subsection{Formulation and regularity conditions}\label{sec:formu_z_est}

We now formulate the inference for the overlap-weighted average treatment effect $\tau_\ow$ using the framework in Section \ref{sec:boot_over_iden}. 
To facilitate the discussion, we first introduce some notation. 
Let the observed variables $\bmO = (\bmX, Z, Y)$ be the collection of observed covariates, treatment assignment and outcome, 
and the unobserved variable 
$W^\tp = \psi_\m(\bmX, \bmU)\in \mathbb{R}$ be the difference between the true propensity score $e_{\tp}(\bmX,\bmU)$ and the estimable propensity score $e_\m(\bmX)$ in the logit scale, 
i.e., $\psi_{\m}(\bm{x}, \bm{u})\equiv\logit\{e_{\tp}(\bm{x}, \bm{u})\}-\logit\{e_\m(\bm{x})\}$. 
Let $e_{{\bm{\beta}}}^{(W)}(\bmX) \equiv  \logit^{-1}[ W + \logit\{e(\bmX; {\bm{\beta}})\} ]$. Then by definition, $e_{{\bm{\beta}}_0}^{(W^\tp)}(\bmX) = e_{\tp}(\bmX, \bmU)$ is the true propensity score, recalling that ${\bm{\beta}}_0$ is the oracle parameter for the propensity score model. 
Define further $\kappa_{\treat 0} \equiv \E[Z\{1-e_{\tp}(\bmX, \bmU)\}]$ and $\kappa_{\control 0} \equiv \E\{(1-Z)e_{\tp}(\bmX, \bmU)\}$, which are the demonstrators in the expression of $\tau_\ow$ in \eqref{eq:tau_ow}. 

We consider then the following set of estimating equations for ${\bmtheta}_0 = (\tau_\ow, {\bm{\beta}}_0, \kappa_{\treat 0}, \kappa_{\control 0})$: 
\begin{align}\label{eq:est_equ_ow}
    \E [p_{{\bmtheta}}(\bmO, W^\tp)]
    \equiv 
    \E 
    \begin{bmatrix}
        Z \cdot \frac{\partial}{\partial {\bm{\beta}}}\log\{ e(\bmX;{\bm{\beta}}) \} + (1-Z) \cdot \frac{\partial}{\partial {\bm{\beta}}} \log \{ 1 - e(\bmX;{\bm{\beta}}) \}
        \\
        \{1-e_{\bm{\beta}}^{(W^\tp)}(\bmX)\}Z - \kappa_{\treat}\\
        e_{\bm{\beta}}^{(W^\tp)}(\bmX)(1 - Z) - \kappa_{\control}\\
        \{1-e_{\bm{\beta}}^{(W^\tp)}(\bmX)\}ZY/\kappa_{\treat} - 
        e_{\bm{\beta}}^{(W^\tp)}(\bmX)(1 - Z)Y/\kappa_{\control} - \tau\\
        \{1-e_{\bm{\beta}}^{(W^\tp)}(\bmX)\}Z g(\bmX)/\kappa_{\treat} - 
        e_{\bm{\beta}}^{(W^\tp)}(\bmX)(1 - Z) g(\bmX) /\kappa_{\control} 
    \end{bmatrix}
    = \bm{0},
\end{align}
where ${\bmtheta} = (\tau, {\bm{\beta}}, \kappa_{\treat}, \kappa_{\control})$, $p_{\bmtheta}(\cdot)$ is defined as above, and $g(\cdot)$ denotes some prespecified covariate transformation. Below we explain the validity of the estimating equations in \eqref{eq:est_equ_ow}, or equivalently why ${\bmtheta}_0 = (\tau_\ow, {\bm{\beta}}_0, \kappa_{\treat 0}, \kappa_{\control 0})$ is a solution for it. 
The first equation in \eqref{eq:est_equ_ow} corresponds to the score equation for the likelihood function of ${\bm{\beta}}$ from the propensity score model. 
The second and third equations correspond to the definition of $\kappa_{\treat 0}$ and $\kappa_{\control 0}$. 
The fourth equation corresponds to the definition of $\tau_\ow$ in \eqref{eq:tau_ow}. 
The last equation corresponds to the covariate balancing property of the true propensity score in \eqref{eq:cov_bal_property}. 

We now discuss the over-identification in \eqref{eq:est_equ_ow}. 
Let $\dim({\bm{\beta}})$ denote the dimension of the parameter for the propensity score model, 
and $\dim(g)$ denote the dimension of transformed covariates used for covariate balancing. 
The number of parameters in ${\bmtheta}$ is then $\dim({\bm{\beta}}) + 3$, and the number of estimating equations is $\dim({\bm{\beta}}) + 3 + \dim(g)$.  
Thus, the estimating equations in \eqref{eq:est_equ_ow} are over-identified, and the over-identification comes from the last equation in \eqref{eq:est_equ_ow} about the covariate balancing property.  
As discussed in Section \ref{eq:z_est_over_unmea}, when the true $W^\tp$ is unobserved, these additional estimating equations can be useful for restricting the unobserved $W^\tp$ and consequently sharpening our inference for the true treatment effect.

Following Section \ref{sec:boot_over_iden}, we will introduce an auxiliary parameter ${\bmrho}\in \mathbb{R}^{\dim(g)}$ with its true value being ${\bmrho}_0 = \bm{0}$, 
along with augmented data $\bm{R}^\tp \in \mathbb{R}^{1\times \dim(g)}$. 
We consider the estimating equation 
$\E [p_{{\bmtheta}}(\bmO, W^\tp + \bm{R}^\tp {\bmrho})] = \bm{0}$ for $({\bmtheta}_0, {\bmrho}_0)$, whose validity follows immediately from the validity of the estimating equations in \eqref{eq:est_equ_ow} and the fact that ${\bmrho}_0 = \bm{0}$. 
To apply the method developed in Section \ref{sec:boot_over_iden}, it suffices to verify that asymptotic normality and bootstrap validity as in Assumption \ref{asmp:valid_bootstrap_z_est} for the Z-estimation of $({\bmtheta}_0, {\bmrho}_0)$ based on the estimating equations $\E [p_{{\bmtheta}}(\bmO, W^\tp + \bm{R}^\tp {\bmrho})] = \bm{0}$, 
as well as the existence of a matrix $\bm{R}^\tp$ with a certain finite moment such that these asymptotic properties hold. 
Below we give a sufficient condition under the logistic regression model for the propensity score, 
where we choose $e(\bm{x}; {\bm{\beta}}) = \logit^{-1}\{ {\bm{\beta}}^\top s(\bm{x}) \}$ with $s(\bm{x})$ being some prespecified transformation of observed covariates. 
Similar logic applies to other parametric models for the propensity score  as well.

\begin{condition}\label{cond:logistic_Z_est}
    $({\bmtheta}, {\bmrho})\equiv(\tau, {\bm{\beta}}, \kappa_{\treat}, \kappa_{\control}, {\bmrho})$ is in a compact parameter space $\Theta\subset \mathbb{R}\times\mathbb{R}^{\dim({\bm{\beta}})}\times[\kappa_{\min}, \infty)\times[\kappa_{\min}, \infty)\times\mathbb{R}^{\dim({\bmrho})}$, where $\kappa_{\min}$ is a positive constant and the true parameter $({\bmtheta}_0, {\bmrho}_0)$ is in the interior of $\Theta$.
    Moreover, the joint distribution of $(Z, \bmX, \bmU, Y)$ satisfies:
    \begin{enumerate}
        \item[(i)] $\E(Y^4)<\infty$, $\E\{\|s(\bmX)\|^4\}<\infty$, and 
        $\E\{ \|g(\bmX)\|^{4} \}$ is finite;
        \item[(ii)] $\text{det}[\E\{s(\bmX)s(\bmX)^{\top}e^{{\bm{\beta}}_0^{\top}s(\bmX)}/(1+e^{{\bm{\beta}}_0^{\top}s(\bmX)})^2\}] >0$;
        \item[(iii)] $\Pr\{0<e_\tp(\bmX, \bmU)<1\}>0$, and $\Cov\{g(\bmX) \mid 0<e_\tp(\bmX, \bmU)<1\}$ is positive definite. 
    \end{enumerate}
\end{condition}

\begin{theorem}\label{thm:valid_bootstrap_R_g}
    Under a logistic regression model for the propensity score, if Condition \ref{cond:logistic_Z_est} holds,  
    then Assumption \ref{asmp:valid_bootstrap_z_est} holds with $\bm{R}^\tp = g^\top(\bmX)$. 
\end{theorem}

From Condition \ref{cond:logistic_Z_est} and Theorem \ref{thm:valid_bootstrap_R_g}, Assumption \ref{asmp:valid_bootstrap_z_est} and consequently our augmented percentile bootstrap inference require rather weak regularity conditions. In particular, we give an explicit construction of $\bm{R}^\tp$ here, which can simply be the transformed covariates in the covariate balancing constraint. 

\subsection{Bootstrap inference for the overlap-weighted average treatment effect}

With the formulation in Section \ref{sec:formu_z_est}, we now apply the method developed in Section \ref{sec:boot_over_iden} to construct bootstrap confidence bounds for the overlap-weighted average treatment effect $\tau_\ow$. 
Without loss of generality, in the remaining of this section, we assume that Assumption \ref{asmp:robust_sen_spec} holds for some given $\Lambda_1, \Lambda_0 \ge 1$ and $\delta_1, \delta_0 \in [0,1]$, 
and we impose the logistic regression model for the observable propensity score as described in Section \ref{sec:formu_z_est}.

We first construct prediction sets for $\bmW_{1:n}^\tp \equiv (W_1^\tp, \ldots, W_n^\tp) = (\psi_\m(\bmX_1, \bmU_1), \ldots, \psi_\m(\bmX_n, \bmU_n))$. 
From Assumption \ref{asmp:robust_sen_spec} and by the definition of $\psi_\m(\cdot)$, we know that 
$\Pr\{ |W^\tp| \le \log(\Lambda_z) \mid Z=z\} 
\ge 1 - \delta_z$ for $z=0,1$. 
This motivates us to consider the following form of prediction sets for $W_i^\tp$s.  
For any $\Lambda_1', \Lambda_0' \ge 1$ and $\delta_1', \delta_0' \in [0,1]$, define $\vec{\bm{\Lambda}}' = (\Lambda_1', \Lambda_0')$, 
$\vec{\bm{\delta}}' = (\delta_1', \delta_0')$, and 
\begin{equation}\label{eq:Lambda_delta_w_pred_set}
    \mathcal{W}_n(\vec{\bm{\Lambda}}', \vec{\bm{\delta}}')  \equiv \big\{ 
    (W_1, \ldots, W_n) \in \mathbb{R}^{n}: 
    \frac{1}{n_z} \sum_{i: Z_i=z} \I \{ |W_i| \leq \log(\Lambda_z') \} \geq 1-\delta_z', z=0,1
    \big\}. 
\end{equation}
Under Assumption \ref{asmp:robust_sen_spec} and 
setting  $n_z\delta_z'$ to be the $\sqrt{1-\zeta}$th quantile of $\Bin(n_z,\delta_z)$ for $z=0,1$, 
we can verify that $\mathcal{W}_n(\vec{\bm{\Lambda}}, \vec{\bm{\delta}}')$ is a prediction set for $\bmW_{1:n}^\tp$ with coverage probability at least $1-\zeta$, where $\vec{\bm{\Lambda}} \equiv (\Lambda_1, \Lambda_0)$. 
In addition, for any $\varepsilon > 0$, the $\varepsilon$-neighborhood of $\mathcal{W}_n(\vec{\bm{\Lambda}}, \vec{\bm{\delta}}')$ is equivalently $\mathcal{W}_n(\vec{\bm{\Lambda}}', \vec{\bm{\delta}}')$ with $\log(\Lambda_z') \equiv \log(\Lambda_z) + \varepsilon$.

We then consider the bootstrap confidence bounds from \eqref{eq:bootstrap_epsilon_pred_set}. 
We consider the following optimization for bootstrapped samples:
for any $\Lambda_1', \Lambda_0' \ge 1$ and $\delta_1', \delta_0' \in [0,1]$, 
\begin{align}\label{eq:tau_boot_epsilon_pred_set}
    \hat{\tau}_{\bstrap}^{\max} (\vec{\bm{\Lambda}}', \vec{\bm{\delta}}') = \ & \sup \tau 
    \quad 
    \textsc{or}
    \quad 
    \hat{\tau}_{\bstrap}^{\min} (\vec{\bm{\Lambda}}', \vec{\bm{\delta}}') = \inf \tau
    \\
    \text{subject to} \quad & 
    \sum_{i=1}^n k_{\bstrap i} p_{\bmtheta}(\bmO_{i}, W_i) =\bm{0}, 
    \nonumber
    \\
    & (W_1, W_2, \ldots, W_n) \in \mathcal{W}_n(\vec{\bm{\Lambda}}', \vec{\bm{\delta}}'),
    \nonumber
\end{align}
where $p_{{\bmtheta}}(\cdot)$ takes the form in \eqref{eq:est_equ_ow} with  
${\bmtheta} = (\tau, {\bm{\beta}}, \kappa_{\treat}, \kappa_{\control})$,  $\mathcal{W}_n(\vec{\bm{\Lambda}}', \vec{\bm{\delta}}')$ is defined as in \eqref{eq:Lambda_delta_w_pred_set}, 
and 
$(k_{\bstrap 1}, k_{\bstrap 2}, \ldots, k_{\bstrap n}) \sim \text{Multinomial} (n; n^{-1}, n^{-1}, \ldots, n^{-1})$. 
Define further 
\begin{align}\label{eq:tau_boot_quantile_gen}
    \LL_{\alpha}(\vec{\bm{\Lambda}}', \vec{\bm{\delta}}') \equiv Q_{\alpha} \big\{ \hat{\tau}_{\bstrap}^{\min} (\vec{\bm{\Lambda}}', \vec{\bm{\delta}}')  \big\}
    \quad 
    \text{and}
    \quad 
    \UU_{\alpha} (\vec{\bm{\Lambda}}', \vec{\bm{\delta}}') \equiv Q_{1-\alpha} \big\{ \hat{\tau}_{\bstrap}^{\max} (\vec{\bm{\Lambda}}', \vec{\bm{\delta}}') \big\}
\end{align}
to denote, respectively, the $\alpha$th lower and upper quantiles of the bootstrap distributions of the minimum and maximum objective values in \eqref{eq:tau_boot_epsilon_pred_set}.

The theorem and corollary below then follow immediately from Theorem \ref{thm:bootstrap_over_identify_pred_set}, which provide confidence bounds for the overlap-weighted average treatment effect under the robust sensitivity model in Assumption \ref{asmp:robust_sen_spec}.  
Recall $\bmW_{1:n}^\tp \equiv (W_1^\tp, W_2^\tp, \ldots, W_n^\tp)$ with $W_i^\tp$ denoting the difference between the true propensity score and estimable propensity score in the logit scale.

\begin{theorem}\label{thm:ci_tau_pred}
    Under Assumption \ref{asmp:robust_sen_spec}  for the robust marginal sensitivity model with some given  $(\Lambda_1, \Lambda_0, \delta_1, \delta_0)$ and Condition \ref{cond:logistic_Z_est}, 
    for any $\alpha \in (0,1)$, $\Lambda_1'>\Lambda_1$ and $\Lambda_0'>\Lambda_0$, and any $\delta_{n, 1}, \delta_{n,0} \in [0,1]$ that can vary with the sizes of treated and control groups, 
    \begin{align*}
        \limsup_{n\rightarrow \infty} \Pr\{ \tau_\ow > \UU_{\alpha}(\vec{\bm{\Lambda}}', \vec{\bm{\delta}}_{n}), \  \bmW_{1:n}^\tp \in \mathcal{W}_n(\vec{\bm{\Lambda}}, \vec{\bm{\delta}}_{n}) \} & \le \alpha, \\
        \limsup_{n\rightarrow \infty} \Pr\{ \tau_\ow < \LL_{\alpha}(\vec{\bm{\Lambda}}', \vec{\bm{\delta}}_{n}), \  \bmW_{1:n}^\tp \in \mathcal{W}_n(\vec{\bm{\Lambda}}, \vec{\bm{\delta}}_{n}) \} & \le \alpha,
    \end{align*}
    and consequently 
    $
        \limsup_{n\rightarrow \infty} \Pr\{ \tau_\ow \notin [\LL_{\alpha}(\vec{\bm{\Lambda}}', \vec{\bm{\delta}}_{n}), \UU_{\alpha}(\vec{\bm{\Lambda}}', \vec{\bm{\delta}}_{n})], \  W_{1:n}^\tp \in \mathcal{W}_n(\vec{\bm{\Lambda}}, \vec{\bm{\delta}}_{n}) \}  \le 2 \alpha,
    $
    where
    $\vec{\bm{\Lambda}}' = (\Lambda_{1}', \Lambda_{0}')$, 
    $\vec{\bm{\delta}}_{n} = (\delta_{n, 1}, \delta_{n, 0})$, 
    $\vec{\bm{\Lambda}} = (\Lambda_{1}, \Lambda_{0})$, 
    $\LL_{\alpha}(\vec{\bm{\Lambda}}', \vec{\bm{\delta}}_{n})$ and $\UU_{\alpha}(\vec{\bm{\Lambda}}', \vec{\bm{\delta}}_{n})$ are defined as in \eqref{eq:tau_boot_quantile_gen}, and $\mathcal{W}_n(\vec{\bm{\Lambda}}, \vec{\bm{\delta}}_{n})$ is defined as in \eqref{eq:Lambda_delta_w_pred_set}. 
\end{theorem}

\begin{corollary}\label{cor:ci_tau}
    Consider the setting in Theorem \ref{thm:ci_tau_pred}, and let $n_z \delta_{n, z}$ be the $\sqrt{1-\zeta}$th quantile of $\Bin(n_z,\delta_z)$ for $z=0,1$. 
    For any given $\Lambda_1'>\Lambda_1$ and $\Lambda_0'>\Lambda_0$,
    $\UU_{\alpha}(\vec{\bm{\Lambda}}', \vec{\bm{\delta}}_{n})$ and $\LL_{\alpha}(\vec{\bm{\Lambda}}', \vec{\bm{\delta}}_{n})$ are, respectively, asymptotic $1-(\alpha+\zeta)$ upper and lower confidence bounds for $\tau_\ow$: 
    \begin{align*}
        \limsup_{n\rightarrow \infty} \Pr\{ \tau_\ow > \UU_{\alpha}(\vec{\bm{\Lambda}}', \vec{\bm{\delta}}_{n}) \} & \le \alpha + \zeta, 
        \qquad 
        \limsup_{n\rightarrow \infty} \Pr\{ \tau_\ow < \LL_{\alpha}(\vec{\bm{\Lambda}}', \vec{\bm{\delta}}_{n}) \}  \le \alpha + \zeta,
    \end{align*}
    and $[\LL_{\alpha}(\vec{\bm{\Lambda}}', \vec{\bm{\delta}}_{n}), \UU_{\alpha}(\vec{\bm{\Lambda}}', \vec{\bm{\delta}}_{n})]$ is an asymptotic $1-(2\alpha+\zeta)$ confidence interval for $\tau_\ow$: 
    \begin{align*}
        \limsup_{n\rightarrow \infty} \Pr\{ \tau_\ow \notin [\LL_{\alpha}(\vec{\bm{\Lambda}}', \vec{\bm{\delta}}_{n}), \UU_{\alpha}(\vec{\bm{\Lambda}}', \vec{\bm{\delta}}_{n})] \}  \le 2 \alpha + \zeta.
    \end{align*}
\end{corollary}

In practice, we can choose $\alpha$ and $\zeta$ based on the desired confidence level. For example, we can set $\alpha$ (or $2\alpha$) and $\zeta$
to be the same as half of 1 minus the confidence level. 
That is, we can set $\alpha = \zeta = 0.025$ when constructing a one-sided $95\%$ confidence interval, and $2\alpha = \zeta = 0.025$ when constructing a two-sided $95\%$ confidence interval.   

The choice of $\Lambda_z'$ is more nuanced. 
Theorem \ref{thm:ci_tau_pred} and Corollary \ref{cor:ci_tau} establish asymptotic validity for any $\Lambda_z' > \Lambda_z$, regardless of how small the gap $\Lambda_z' - \Lambda_z$ is. Moreover, this gap can shrink to zero at an appropriate rate as $n \to \infty$; see the discussion after Theorem \ref{thm:bootstrap_over_identify}. Numerical experiments indicate that the confidence bounds are robust to small variations in $\Lambda_z'$. Thus, we generally recommend setting $\Lambda_z' = \Lambda_z$ in practice. However, we emphasize that for the theoretical validity of our augmented percentile bootstrap procedure, a strictly positive gap between $\Lambda_z'$ and $\Lambda_z$, no matter how small, is required.

\subsection{
Optimization for the percentile bootstrap confidence intervals
}\label{sec:opt_per_boot}

In this subsection, we will focus on how to compute the confidence bounds in Theorem \ref{thm:ci_tau_pred} and Corollary \ref{cor:ci_tau}, 
which all rely on the optimization in \eqref{eq:tau_boot_epsilon_pred_set} based on the bootstrapped samples. 
In the following, we focus on the computation of  
$\hat{\tau}_{\bstrap}^{\max} (\vec{\bm{\Lambda}}, \vec{\bm{\delta}})$ and $\hat{\tau}_{\bstrap}^{\min} (\vec{\bm{\Lambda}}, \vec{\bm{\delta}})$ defined as in \eqref{eq:tau_boot_epsilon_pred_set} for any $\Lambda_1, \Lambda_0\ge 1$ and $\delta_1, \delta_0 \in [0,1]$. 
To facilitate the discussion, 
we first introduce some notation similar to that in Section \ref{sec:point_sen_ana}. 
For each bootstrap sample denoted by $\bstrap$ and the associated $(k_{\bstrap 1}, \ldots, k_{\bstrap n})$, 
we use $\hat{{\bm{\beta}}}_\bstrap$ to denote the fitted parameter value for the propensity score model based on the bootstrap sample. 
Define further, for $1\le i \le n$ and $\psi_i \in \mathbb{R}$, 
$\hat{e}_{\bstrap i}^{(\psi_i)} \equiv e_{\hat{{\bm{\beta}}}_{\bstrap}}^{(\psi_i)}(\bmX_i) = \logit^{-1}[ \psi_{i} + \logit\{ e(\bmX_i; \hat{{\bm{\beta}}}_{\bstrap}) \} ]$ analogously as in Section \ref{sec:formu_z_est}. 

The optimization for $\hat{\tau}_{\bstrap}^{\min} (\vec{\bm{\Lambda}}, \vec{\bm{\delta}})$ and $\hat{\tau}_{\bstrap}^{\max} (\vec{\bm{\Lambda}}, \vec{\bm{\delta}})$ as in \eqref{eq:tau_boot_epsilon_pred_set} can be equivalently written as 
\begin{align}\label{eq:double_frac_program_bootstrap_2}
    \min 
    \text{ or } \max \quad 
    & 
    \dfrac{\sum_{i=1}^n k_{\bstrap i} Z_{i} (1-
    \hat{e}_{\bstrap i}^{(\psi_i)}
    )Y_{i}}{\sum_{i=1}^n k_{bi} Z_{i} (1-
    \hat{e}_{\bstrap i}^{(\psi_i)}
    ) }-\dfrac{\sum_{i=1}^n k_{\bstrap i}(1-Z_{i}) \hat{e}_{\bstrap i}^{(\psi_i)} Y_{i}}{\sum_{i=1}^n k_{\bstrap i} (1-Z_{bi}) \hat{e}_{\bstrap i}^{(\psi_i)}},
    \\
    \text{subject to} \quad 	&
    \sum_{i: Z_i=z} \I \{ |\psi_i | \leq \log(\Lambda_z) \} \ge \lceil n_z
    (1-\delta_z)\rceil \text{ for } z=0,1, 
    \nonumber
    \\
    & \dfrac{\sum_{i=1}^n k_{\bstrap i} Z_{i} (1-
    \hat{e}_{\bstrap i}^{(\psi_i)}
    ) g(\bmX_{i})}{\sum_{i=1}^n k_{\bstrap i} Z_{i} (1-
    \hat{e}_{\bstrap i}^{(\psi_i)}
    ) }
    = 
    \dfrac{\sum_{i=1}^n k_{\bstrap i}(1-Z_{i}) \hat{e}_{\bstrap i}^{(\psi_i)} g(\bmX_{i})}{\sum_{i=1}^n 
    k_{\bstrap i}(1-Z_{i}) \hat{e}_{\bstrap i}^{(\psi_i)}
    },
    \nonumber
\end{align}
where the minimum and maximum objective values are the same as  $\hat{\tau}_{\bstrap}^{\min} (\vec{\bm{\Lambda}}, \vec{\bm{\delta}})$ and $\hat{\tau}_{\bstrap}^{\max} (\vec{\bm{\Lambda}}, \vec{\bm{\delta}})$, respectively. The equivalent form in \eqref{eq:double_frac_program_bootstrap_2} follows from \eqref{eq:tau_boot_epsilon_pred_set} by solving the empirical estimating equations $\sum_{i=1}^n k_{\bstrap i} p_{\bmtheta}(\bmO_{i}, W_i) =\bm{0}$ and denoting $W_i$ by $\psi_i$, where the latter is to make it more coherent with the bound estimation in \eqref{eq:opt_tau_hat_ow}.

It is worth pointing out that, although the optimization in \eqref{eq:double_frac_program_bootstrap_2} appears similar to that in \eqref{eq:opt_tau_hat_ow}, there are subtle yet important differences between them.
Specifically, \eqref{eq:double_frac_program_bootstrap_2} is generally different from the optimization in \eqref{eq:opt_tau_hat_ow} applied directly to the bootstrap sample. 
This is because both the objective and the second constraint on covariate balance depend only on the bootstrap sample, whereas the first constraint from the robust sensitivity model refers to $\{\psi_i\}_{i=1}^n$ for the original sample. 
This subtle difference is not that critical under the usual marginal sensitivity model, under which all the $\psi_i$s as well as their bootstrapped versions will be bounded by the same value; see, e.g., \citet{zhaosenIPW2019}.

Nevertheless, we can still view the optimization in \eqref{eq:double_frac_program_bootstrap_2} as a weighted version generalizing that in \eqref{eq:opt_tau_hat_ow}.  
In particular, by the same logic as that in Sections \ref{subsec:double_frac_program} and \ref{sec:milp}, we can write \eqref{eq:double_frac_program_bootstrap_2} in a form analogous to \eqref{eq:frac_double}, and then transform it into a (mixed-integer) linear programming problem. 
Importantly, linear programming still yields valid confidence bounds, often with little increase in conservativeness.
We relegate the details to the supplementary material.

\section{Illustration}\label{sec:illu}

We now illustrate our methods by investigating the causal effect of fish consumption on the blood mercury level. 
We use the data from  \citet{Zhao18cross} and \citet{zhaosenIPW2019}, which are obtained from the National Health and Nutrition Examination Survey 2013–2014. 
After their data prepossessing, the final dataset includes $234$ treated units with high fish consumption and $873$ control units with low fish consumption. 
The outcome of interest is the blood mercury level, measured by $\log_2$(total blood
mercury in micrograms per liter).
Besides the treatment and outcome variables, we have eight covariates for these units, including  gender, age, income, whether income is missing, race, education, ever smoked and number of cigarettes smoked last month. 
We apply the proposed sensitivity analysis to study the causal effect of fish consumption, where we impose a logistic regression model for the observable propensity score and enforce balance for all the eight covariates as in \eqref{eq:cov_bal_psi}. 
In addition, the bootstrap distributions are approximated via Monte Carlo method with 1000 random draws.

\begin{figure}[htbp]
    \centering\includegraphics[width=0.66\textwidth]{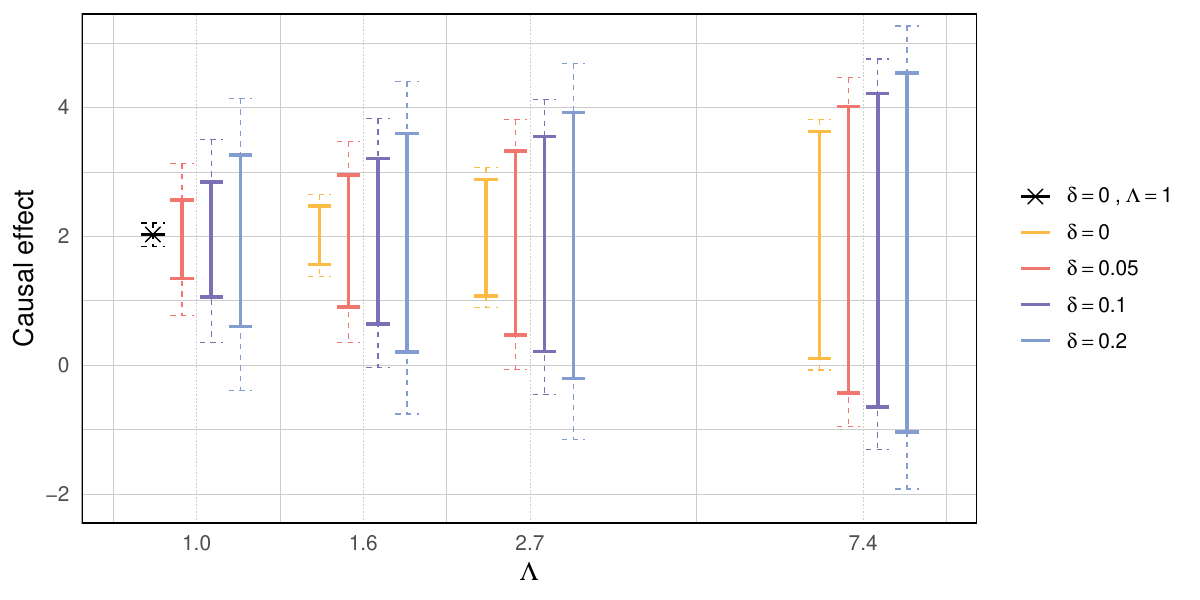}
    \caption{Estimated bounds and 95\% confidence intervals for the overlap-weighted average causal effect $\tau_{\ow}$ of fish consumption on the blood mercury level.
    The solid intervals show the estimated bounds of $\tau_{\ow}$ as described in Section \ref{sec:point_sen_ana}, and the dashed intervals show the 95\% confidence intervals for $\tau_{\ow}$, 
    under various sensitivity models indexed by $(\Lambda, \delta)$. 
    Specifically, $(\Lambda, \delta)$ corresponds to the robust marginal sensitivity model assuming that at least $1-\delta$ proportions of units within both treated and control populations have unmeasured confounding strength bounded by $\Lambda$. Note that $(\Lambda, \delta)=(1,0)$ corresponds to the case of no unmeasured confounding.
    }
    \label{fig:ci}
\end{figure}

Specifically, we consider the robust marginal sensitivity model in Assumption \ref{asmp:robust_sen_spec} with the same constraints on unmeasured confounding for the treated and control populations; that is, $\Lambda_1=\Lambda_0=\Lambda$ and $\delta_1=\delta_0=\delta$ for some $\Lambda$ and $\delta$. 
We consider four values of $\Lambda$: $\exp(0)$, $\exp(0.5)$, $\exp(1)$, and $\exp(2)$, and four values of $\delta$: $0$, $0.05$, $0.1$, and $0.2$. This results in a total of 16 combinations of $(\Lambda, \delta)$. 
Figure \ref{fig:ci} shows the confidence intervals under each combination of $(\Lambda, \delta)$. 
Note that $(\Lambda, \delta) = (1,0)$ corresponds to the case with  no unmeasured confounding.  

From Figure \ref{fig:ci}, in the absence of unmeasured confounding, the point estimate is about $2$, and the corresponding confidence interval is quite narrow and does not include zero, which indicates a significant nonzero causal effect. 
However, we may worry that the inferred significant effect may be due to unmeasured confounding rather than true causal effects.

We then consider the conventional marginal sensitivity model, which assumes that the strength of unmeasured confounding is bounded by a certain value $\Lambda$ uniformly across all units; this corresponds to cases with $\delta = 0$ in our analyses. As shown in Figure~\ref{fig:ci}, the causal effect remains significantly different from zero for $\Lambda$ values up to about $\exp(2) \approx 7.4$. In other words, 
fish consumption can significantly increase the blood mercury level 
even when the odds of treatment probability could increase or decrease by a factor of 7.4 once further conditioning on the unmeasured confounder.
While such a strong level of unmeasured confounding may be implausible for all units, it is possible that a small subset of units experience extreme unmeasured confounding. In this case, the conventional marginal sensitivity analysis may not be able to provide robust evidence for a significant causal effect.

We finally consider the proposed robust marginal sensitivity model. When $\delta = 0.05$, meaning that up to 5\% of units in both the treated and control groups may be subject to arbitrarily strong unmeasured confounding, the causal effect remains significant for $\Lambda$ values up to approximately $\exp(1) \approx 2.7$. Analogously, for $\delta = 0.1$, the effect is significant for $\Lambda$ values up to approximately $\exp(0.5) \approx 1.6$. However, when $\delta = 0.2$, the causal effect is no longer significant, even when $80\%$ of units in both treated and control populations do not suffer from unmeasured confounding at all. These results suggest that the observational data support a nonzero causal effect that is moderately robust to unmeasured confounding, even when such confounding may be severe for some units.

\section{Extensions and discussion}\label{sec:diss}

First, we extend to robust marginal sensitivity models assuming that 
\begin{align}\label{eq:robust_sen_all}
    \Pr \Big( \Lambda^{-1}\leq \dfrac{e_{\tp}(\bmX,\bmU)/\{1-e_{\tp}(\bmX,\bmU)\}}{e_{\tp}(\bmX)/\{1-e_{\tp}(\bmX)\}}\leq\Lambda  \Big) & \geq 1-\delta, \quad \text{for some $\Lambda\ge 1$ and $\delta \in [0,1]$}. 
\end{align}
The inference can be conducted analogously using the augmented percentile bootstrap. However, it generally requires higher computational cost. Intuitively, we can understand the model in \eqref{eq:robust_sen_all} as a collection of models in \eqref{eq:robust_sen} such that $\Lambda_1 = \Lambda_0 = \Lambda$ and $\Pr(Z=1) \cdot \delta_1 + \Pr(Z=0) \cdot \delta_0 = \delta$. 
Therefore, we can solve the optimization under \eqref{eq:robust_sen_all} via multiple linear programming problems in Section \ref{sec:milp}. 
Alternatively, we can use quadratically constrained linear programming. 
We relegate the details to Appendix A3 of the supplementary material. 

Second, we develop simultaneous sensitivity analyses across multiple (or all) models in \eqref{eq:robust_sen_all} or \eqref{eq:robust_sen}. 
This utilizes an equivalent understanding of sensitivity analysis as discussed in \citet{WL23}, where we instead infer the minimum amount of unmeasured confounding under the assumption of no causal effects. 
For ease of interpretation and implementation, we propose to work directly with the confounding strength within the sample. 
Let $\OR_{\m i} \ge 1$ be the odds ratio between the true and estimable propensity scores or its reciprocal for unit $i$, 
and $\OR_{\m [q]}^\treat$ (or $\OR_{\m [q]}^\control$) be the $q$th sample quantile of $\OR_{\m i}$s among treated (or control) units. 
Assuming the treatment has no effect, 
we can construct simultaneous prediction sets for $(\OR_{\m [q_1]}^\treat, \OR_{\m [q_0]}^\control)$ across all $q_1, q_0 \in [0,1]$,
which also implies 
simultaneous prediction sets for the sample quantiles of $\OR_{\m i}$s from both treated and control units. 
Considering the computational cost, we recommend reporting the resulting prediction sets for $\max\{ \OR_{\m [q]}^\treat, \OR_{\m [q]}^\control\}$ over $q \in [0,1]$. We relegate the details to Appendix A1 of the supplementary material.

Our paper provides a unified framework for sensitivity analysis under the violations of both unconfoundedness and overlap, which involves both linear programming for efficient computation and augmented percentile bootstrap for statistical inference. 
Moreover, the proposed bootstrap approach for parameters defined by over-identified estimating equations may be of independent interest and applicable to other settings. 
It will also be interesting to extend \citet{Dorn23}'s sharp sensitivity analysis for the usual marginal sensitivity model to our robust marginal sensitivity model, though this may be challenging due to the combinatorial nature of the robust marginal sensitivity model.
We leave these for future investigation.

\bibliographystyle{plainnat}
\bibliography{reference}

\newpage

\begin{center}
	\bf \LARGE 
	Supplementary Material 
\end{center}

\setcounter{equation}{0}
\setcounter{section}{0}
\setcounter{figure}{0}
\setcounter{example}{0}
\setcounter{proposition}{0}
\setcounter{corollary}{0}
\setcounter{theorem}{0}
\setcounter{lemma}{0}
\setcounter{table}{0}
\setcounter{condition}{0}
\setcounter{assumption}{0}
\setcounter{remark}{0}

\renewcommand {\theproposition} {A\arabic{proposition}}
\renewcommand {\theexample} {A\arabic{example}}
\renewcommand {\thefigure} {A\arabic{figure}}
\renewcommand {\thetable} {A\arabic{table}}
\renewcommand {\theequation} {A\arabic{equation}}
\renewcommand {\thelemma} {A\arabic{lemma}}
\renewcommand {\thesection} {A\arabic{section}}
\renewcommand {\thetheorem} {A\arabic{theorem}}
\renewcommand {\thecorollary} {A\arabic{corollary}}
\renewcommand {\thecondition} {A\arabic{condition}}
\renewcommand {\thermk} {A\arabic{rmk}}
\renewcommand {\theremark} {A\arabic{remark}}
\renewcommand {\theassumption} {A\arabic{assumption}}

\renewcommand {\thepage} {A\arabic{page}}
\setcounter{page}{1}

\onehalfspacing

\section{Sensitivity analysis based on multiple quantiles of unmeasured confounding strength}\label{eq:multiple_quantile}

In Section \ref{sec:ci_for_overlap}, we construct confidence intervals for the overlap-weighted average treatment effect $\tau_{\ow}$ under the robust marginal sensitivity model with any prespecified $(\vec{\bm{\Lambda}}, \vec{\bm{\delta}})$. 
In practice, it is often desirable to try multiple sensitivity models with various values of $(\vec{\bm{\Lambda}}, \vec{\bm{\delta}})$. 
This naturally leads to the following questions: how should the results from multiple sensitivity analyses be interpreted, and is any adjustment needed to account for conducting multiple analyses?
We will address these questions throughout this section. 

We first note that the robust marginal sensitivity model in Assumption \ref{asmp:robust_sen_spec} is equivalent to
imposing bounds on the quantiles of unmeasured confounding strength measured by
\begin{equation}\label{eq:OR_m}
    \OR_{\m}(\bmX,\bmU)
    \equiv \max\left\{ \dfrac{\odd\{ e_{\tp}(\bmX,\bmU) \}}{\odd\{ e_{\m}(\bmX)\}}, 
    \dfrac{\odd\{ e_{\m}(\bmX) \}}{\odd \{ e_{\tp}(\bmX,\bmU) \}} \right\}
    = \exp\{ |\psi_{\m}(\bmX,\bmU)| \},
\end{equation}
where the last equality holds by the definition of $\psi_{\m}(\cdot)$ in Section \ref{sec:formu_z_est}. 
Specifically, Assumption \ref{asmp:robust_sen_spec} for any given $(\vec{\bm{\Lambda}}, \vec{\bm{\delta}})$ is equivalent to that the $(1-\delta_z)$th quantile of the conditional distribution of $\OR_{\m}(\bmX,\bmU)$ given $Z=z$ is bounded by $\Lambda_z$, for $z=0,1$. 
Consequently, Theorem \ref{thm:ci_tau_pred} and Corollary \ref{cor:ci_tau} equivalently provide confidence intervals for the overlap-weighted average treatment effect $\tau_{\ow}$ given bounds on the quantiles of unmeasured confounding strength.

\subsection{An alternative view of sensitivity analysis}\label{sec:equiv_view_sa}

In the main paper, we assume certain bounds on the strength of unmeasured confounding, as in the usual marginal sensitivity model, 
and then investigate how sensitive the inferred treatment effects can be. In particular, we can construct confidence bounds for the true overlap-weighted average treatment effect $\tau_{\ow}$ under any given sensitivity model. 
If the confidence intervals do not cover zero, then the treatment effect $\tau_{\ow}$ is said to be significantly different from zero, or significantly positive or negative if we use one-sided confidence intervals, under the given sensitivity model. 
The more the amount of unmeasured confounding allowed under this sensitivity model, the more robust the inferred causal conclusions will be.

To better accommodate sensitivity analyses with various constraints on the strength of unmeasured confounding, we consider an alternative and equivalent view of sensitivity analysis, as discussed in, for example, \citet{WL23}. 
Instead of inferring bounds of treatment effects under constraints on strength of unmeasured confounding, 
we instead infer the strength of unmeasured confounding under assumptions on treatment effects. 
For example, we can assume that the overlap-weighted average treatment effect is zero (or nonpositive, or nonnegative), and construct confidence bounds on the strength of unmeasured confounding under this assumption, which typically give lower bounds for the unmeasured confounding strength. 
Intuitively, this indicates the minimum amount of unmeasured confounding required to make a zero average treatment effect possible. 
When the required strength of unmeasured confounding is large, it suggests that the average treatment effect is likely to be nonzero. On the contrary, if the required strength of unmeasured confounding is small, then the average treatment effect is likely to be zero, or at least the conclusion of a nonzero effect is sensitive to unmeasured confounding. 

In the remaining of this section, we will take this alternative perspective of sensitivity analysis. 
In particular, we will assume that the true overlap-weighted average treatment effect $\tau_{\ow}$ is zero (or nonpositive, or nonnegative), and then infer the strength of unmeasured confounding measured by the quantiles of the distribution of $\OR_{\m}(\bmX,\bmU)$ within treated and control groups. 
Moreover, 
we will consider simultaneously valid confidence intervals for multiple quantiles of unmeasured confounding strength. 
As detailed below, we will essentially consider various values of $\vec{\bm{\Lambda}}$ and $\vec{\bm{\delta}}$ for the robust sensitivity model in Assumption \ref{asmp:robust_sen_spec}, check whether the corresponding confidence intervals cover zero, and then use the results to construct simultaneous confidence intervals for quantiles of unmeasured confounding strength.

\subsection{Simultaneous sensitivity analysis for quantiles of confounding strength}

In the following, we assume that the overlap-weighted average treatment effect is zero, i.e., $\tau_{\ow} = 0$, and use bootstrap two-sided confidence intervals of form $[\LL_{\alpha}(\vec{\bm{\Lambda}}, \vec{\bm{\delta}}), \UU_{\alpha}(\vec{\bm{\Lambda}}, \vec{\bm{\delta}})]$ in Theorem \ref{thm:ci_tau_pred} to conduct inference for the strength of unmeasured confounding. 
Similar logic applies when we assume $\tau_{\ow} \ge 0$ or $\tau_{\ow} \le 0$, and use upper or lower bootstrap confidence bounds.  
For the ease of interpretation and implementation, 
we will focus on inferring the strength of unmeasured confounding within the sample; see Remark \ref{remark:gen_sam_pop} for generalization to the treated and control populations. 
Specifically, 
for $0\le q\le 1$, 
let $\OR_{\m[q]}^\treat$ denote the $q$th sample quantile of the strength of unmeasured confounding among the treated units, i.e., $\{\OR_{\m}(\bmX_i, \bmU_i): Z_i = 1, 1\le i \le n\}$. 
Analogously, let $\OR_{\m [q]}^\control$ denote the $q$th sample quantile of the strength of unmeasured confounding among the control units. 
Note that, by definition, $W_{1:n}^\tp \in \mathcal{W}_n(\vec{\bm{\Lambda}}, \vec{\bm{\delta}})$ 
if and only if 
$\OR_{\m[1-\delta_{1}]}^\treat \le \Lambda_1$ 
and 
$\OR_{\m[1-\delta_{0}]}^\control \le \Lambda_0$. 
From the equivalent understanding of sensitivity analysis in Section \ref{sec:equiv_view_sa} and Theorem \ref{thm:ci_tau_pred}, 
we can immediately derive asymptotically valid tests for the following null hypothesis about the sample quantiles of unmeasured confounding strength: 
\begin{equation}\label{eq:H_samp_U}
    H_{\vec{\bm{q}},\vec{\bm{\Lambda}}}^{\text{U}}: \OR^\treat_{\m[q_1]} \le \Lambda_1 \text{ and } \OR^\control_{\m[q_0]} \le \Lambda_0, 
\end{equation}
where $\vec{\bm{q}} = (q_1, q_0)\in [0,1]^2$ and $\vec{\bm{\Lambda}} = (\Lambda_1, \Lambda_0) \in [1, \infty]^2$ are any prespecified quantiles and bounds, and we use the subscript U to emphasize that this is a hypothesis about the strength of unmeasured confounding. 
The hypothesis in \eqref{eq:H_samp_U} is random due to the randomness in the sample, and a test is said to be (asymptotically) valid if and only if (the limit superior of) the probability that the null hypothesis is true and the test rejects it is bounded from the above by the desired significance level \citep[see, e.g.,][]{chen2024enhanced}.

\begin{corollary}\label{cor:test_U}
    Assume that Condition \ref{cond:logistic_Z_est} holds and the true overlap-weighted average treatment effect is zero, i.e., $\tau_{\ow} = 0$. 
    Consider the null hypothesis $H_{\vec{\bm{q}},\vec{\bm{\Lambda}}}^{\text{U}}$ in \eqref{eq:H_samp_U} for any given $\vec{\bm{q}} \in [0,1]^2$ and $\vec{\bm{\Lambda}} \in [1, \infty]^2$. Let $\xi > 1$ be any fixed constant. 
    For any $\alpha \in (0,1)$ and $\xi > 1$, 
    the test that rejects $H_{\vec{\bm{q}},\vec{\bm{\Lambda}}}^{\text{U}}$ if and only if $0\notin [\LL_{\alpha/2}(\xi \vec{\bm{\Lambda}}, \vec{\bm{1}}-\vec{\bm{q}}), \UU_{\alpha/2}(\xi \vec{\bm{\Lambda}}, \vec{\bm{1}}-\vec{\bm{q}})]$ is an asymptotically valid level-$\alpha$ test for the null hypothesis $H_{\vec{\bm{q}},\vec{\bm{\Lambda}}}$,    
    in the sense that 
    \begin{align*}
        \limsup_{n\rightarrow \infty} \Pr\big\{
        H_{\vec{\bm{q}},\vec{\bm{\Lambda}}}^{\text{U}} \text{ holds and } 0\notin [\LL_{\alpha/2}(\xi \vec{\bm{\Lambda}}, \vec{\bm{1}}-\vec{\bm{q}}), \UU_{\alpha/2}(\xi \vec{\bm{\Lambda}}, \vec{\bm{1}}-\vec{\bm{q}})]
        \big\} \le \alpha, 
    \end{align*}
    where $\xi \vec{\bm{\Lambda}} = (\xi \Lambda_1, \xi \Lambda_0)$ and 
    $\vec{\bm{1}}-\vec{\bm{q}} = (1-q_1, 1-q_0)$. 
\end{corollary}

Moreover, by test inversion, we can then construct prediction sets for quantiles of unmeasured confounding strength within the sample. Importantly, such prediction sets will be simultaneously valid across all quantiles, which can then strengthen our causal evidence from observational studies; see Section \ref{sec:ill_simultaneous} for an illustration.
We summarize the results in the following theorem. 
Define 
\begin{align}\label{eq:I_q_vec_alpha}
        \mathcal{I}_{\vec{\bm{q}}}^\alpha \equiv \big\{ \vec{\bm{\Lambda}} = (\Lambda_1, \Lambda_0) \in [1, \infty]^2: 
        0 \in [\LL_{\alpha/2}(\vec{\bm{\Lambda}}, \vec{\bm{1}}-\vec{\bm{q}}), \UU_{\alpha/2}(\vec{\bm{\Lambda}}, \vec{\bm{1}}-\vec{\bm{q}})] \big\}.
\end{align}
Define further $\mathcal{I}_{\vec{\bm{q}}, \xi}^\alpha = \xi^{-1} \mathcal{I}_{\vec{\bm{q}}}^\alpha = \{\xi^{-1}\vec{\bm{\Lambda}}: \vec{\bm{\Lambda}} \in \mathcal{I}_{\vec{\bm{q}}}^\alpha\}$ for $\xi \ge 1$.

\begin{theorem}\label{thm:U_quant_pred_int}
    Assume that Condition \ref{cond:logistic_Z_est} holds and the true overlap-weighted average treatment effect is zero, i.e., $\tau_{\ow} = 0$. 
    Let $\alpha \in (0,1)$ and $\xi >1$ be any fixed constants. 
    Then, for any $\vec{\bm{q}} = (q_1, q_0)\in [0,1]^2$,  $\mathcal{I}_{\vec{\bm{q}}, \xi}^\alpha$ is an asymptotic $1-\alpha$ prediction set for $(\OR^\treat_{\m[q_1]}, \OR^\control_{\m[q_0]})$. Moreover, it is simultaneously valid across all quantiles, in the sense that 
    \begin{align*}    
        \liminf_{n \rightarrow \infty} \Pr\big\{ (\OR^\treat_{\m[q_1]}, \OR^\control_{\m[q_0]}) \in \mathcal{I}_{\vec{\bm{q}}, \xi}^\alpha \text{ for all } \vec{\bm{q}} \in [0,1]^2 \big\} \ge 1-\alpha. 
    \end{align*}
\end{theorem}

From Theorem \ref{thm:U_quant_pred_int}, under the assumption of zero average treatment effect, 
we can construct lower prediction bounds for quantiles of unmeasured confounding strength within the sample. Importantly, these prediction bounds are simultaneously valid, which greatly enhances the causal evidence by requiring more constraints on the strength of unmeasured confounding. 
If any of these prediction bounds are not likely to hold, we can then conclude that the average treatment effect is likely nonzero, and such a conclusion is robust to measured confounding. 
However, computing all the prediction bounds in Theorem \ref{thm:U_quant_pred_int} can be computationally demanding, as it requires conducting the sensitivity analysis in Theorem \ref{thm:ci_tau_pred} over all $(\vec{\bm{\Lambda}}, \vec{\bm{\delta}})$ specified by the robust marginal sensitivity model in Assumption \ref{asmp:robust_sen_spec}. 
In addition, the resulting prediction sets are not easy to visualize. 
In the following, we propose to focus on constructing prediction sets for 
$\overline{\OR}_{\m [q]} \equiv \max\{\OR^\treat_{\m[q]}, \OR^\control_{\m[q]}\}$,  for $q\in [0,1]$. 
The resulting prediction sets essentially form a superset of the simultaneous prediction sets in Theorem \ref{thm:U_quant_pred_int}, but require less computation cost and are easier to visualize. 

For any $\alpha\in (0,1)$ and $q\in [0,1]$, 
define 
\begin{align*}
        \mathcal{J}_{q}^\alpha \equiv \big\{ \Lambda \in [1, \infty]: 
        0 \in [\LL_{\alpha/2}( \Lambda \vec{\bm{1}}, (1-q)\vec{\bm{1}}), \UU_{\alpha/2}(\Lambda \vec{\bm{1}}, (1-q)\vec{\bm{1}})] \big\}, 
\end{align*}
which must be an interval of form $(c, \infty]$ or $[c, \infty]$. 
Define further $\mathcal{J}_{q, \xi}^\alpha = \xi^{-1} \mathcal{J}_{q}^\alpha$. 

\begin{corollary}\label{cor:U_quant_pred_int_max_tr_co}
    Assume that Condition \ref{cond:logistic_Z_est} holds and the true overlap-weighted average treatment effect is zero, i.e., $\tau_{\ow} = 0$. 
    Let $\alpha \in (0,1)$ and $\xi >1$ be any fixed constants. 
    Then, for any $q \in [0,1]$,  $\mathcal{J}_{q, \xi}^\alpha$ is an asymptotic $1-\alpha$ prediction interval for $\overline{\OR}_{\m [q]} \equiv \max\{\OR^\treat_{\m[q]}, \OR^\control_{\m[q]}\}$. Moreover, it is simultaneously valid across all quantiles, in the sense that 
    $
        \liminf_{n \rightarrow \infty} \Pr\{ \overline{\OR}_{\m [q]}\in \mathcal{J}_{q, \xi}^\alpha \text{ for all } q \in [0,1]\} \ge 1-\alpha. 
    $
\end{corollary}

\begin{remark}\label{remark:gen_sam_pop}
    In the above Theorem \ref{thm:U_quant_pred_int} and Corollary \ref{cor:U_quant_pred_int_max_tr_co}, we focus on inference for the strength of unmeasured confounding within the sample. We can generalize the inference to population quantiles of unmeasured confounding strength, using a similar strategy from \citet{chen2024enhanced}; see also \citet{chen2024role}. 
    Specifically, we can first construct (simultaneous) ``confidence'' sets for quantiles of unmeasured confounding strength within the treated (or control) population using the sample quantiles among treated (or control) units; these intervals are not calculable since the sample quantiles of unmeasured confounding strength are unobserved.
We then use Theorem \ref{thm:U_quant_pred_int} and Corollary \ref{cor:U_quant_pred_int_max_tr_co} to construct simultaneous prediction sets for the sample quantiles of unmeasured confounding strength, which can then be used to bound the population quantiles. 
In practice, we suggest reporting directly the simultaneous prediction bounds in Corollary \ref{cor:U_quant_pred_int_max_tr_co} for the sample quantiles of unmeasured confounding strength, which is more convenient and involves less choice of tuning parameters. 
Also, the subtle difference between sample and population quantiles will become negligible as the sample size $n \to \infty$.

\end{remark}

\subsection{Numerical illustration}\label{sec:ill_simultaneous}

We now illustrate the simultaneous prediction intervals from Corollary \ref{cor:U_quant_pred_int_max_tr_co} and their advantage using the application in Section \ref{sec:illu}. 
Note that the theoretical validity in Corollary \ref{cor:U_quant_pred_int_max_tr_co} requires a $\xi$ value greater than 1, no matter  how small the gap between $\xi$ and $1$ is. 
We set the $\xi$ parameter to be $1$ for convenience, since the prediction bounds are robust to small variations in $\xi$. 
Figure \ref{fig:strength_unmeaured_confounding} shows the 
lower prediction bounds for $\overline{\OR}_{\m [q]}$s; note that the plot is truncated at $q=80\%$, since the lower prediction bounds for smaller quantiles are the uninformative $1$. 
From Figure \ref{fig:strength_unmeaured_confounding}, if fish consumption has no causal effect on the blood mercury level, then, for example, 
some treated or control units must have unmeasured confounding strength of at least 6.70 (i.e., the unmeasured confounding changes the odds of treatment probability by a factor of at least about 6.70), 
$5\%$ of treated or control units must have unmeasured confounding strength of at least 3.27, 
and $10\%$ of treated or control units must have unmeasured confounding strength of at least 2.09.
Importantly, these statements will hold simultaneously with probability at least $95\%$, if fish consumption indeed has no causal effect. Therefore, we can detect significant causal effect if any of them is unlikely to be true. 
Figure \ref{fig:strength_unmeaured_confounding} suggests that the causal effect of fish consumption is robust to unmeasured confounding.

\begin{figure}[htbp]
    \centering\includegraphics[width=0.7\textwidth]{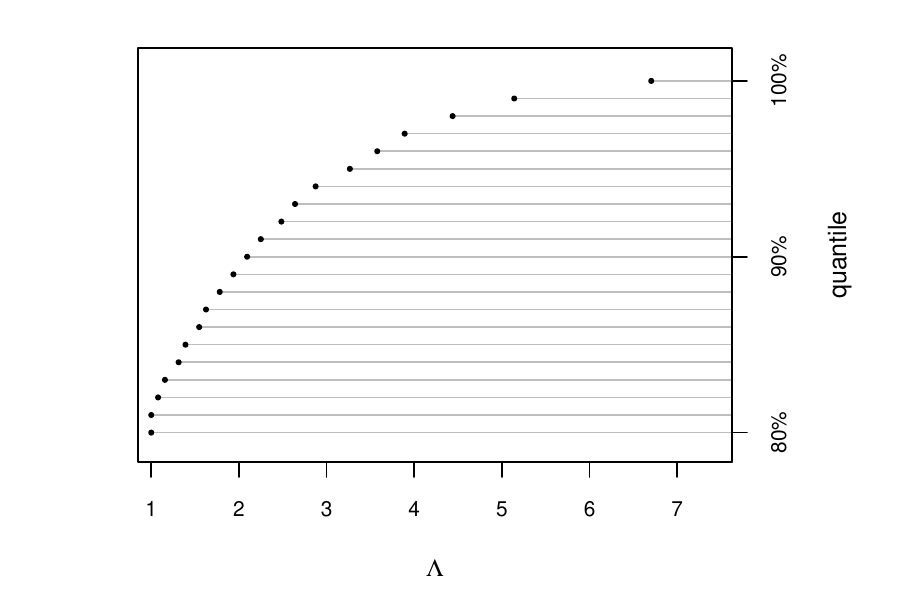}
    \caption{Lower prediction bounds for $\overline{\OR}_{\m [q]}$s across various quantiles of unmeasured confounding strength, under the assumption of a zero overlap-weighted average treatment effect of the fish consumption on the blood mercury level.}
    \label{fig:strength_unmeaured_confounding}
\end{figure}

\section{Additional details for the numerical optimization}

\subsection{Glover's linearization for Mixed-integer linear programming}\label{sec:glover_milp_appdix}
In this subsection, we linearize the constraints of $\bar{\Delta}_i\in \{0, t_{Z_i}\}$ using the \citet{GloverImproved:1975}’s linearization, so that the optimization in \eqref{eq:MILP_sep} becomes a standard mixed-integer linear programming problem. 
Specifically, for a sufficiently large $M$, the optimized objective from \eqref{eq:MILP_sep} is equivalent to that from the following one, optimized over $\bar{\omega}_{i}$s, $\bar{\Delta}_{i}$s, $t_1$, $t_0$ and $\tilde{\Delta}_i$s:
\begin{align}\label{eq:MILP_glover}
\min 
 \text{ or } \max \quad & \sum_{i=1}^{n_1} \bar{\omega}_{i} Y_i-\sum_{i=n_1+1}^n \bar{\omega}_{i} Y_i \\
\text{subject to} \quad 	&
a_i^{\low} \bar{\Delta}_{i} \leq \bar{\omega}_{i} \leq t_{Z_i}-(1-a_i^{\up})\bar{\Delta}_{i}   \text{ for } 1\le i\le n,
\nonumber
\\
	&\sum_{i=1}^{n_1} \bar{\Delta}_{i}\geq  
    \lceil n_1 (1-\delta_1)\rceil t_1, \quad 
    \sum_{i=n_1+1}^n \bar{\Delta}_{i}\geq  
    \lceil n_0 (1-\delta_0)\rceil t_0,
 \nonumber
 \\
 & 0 \le \bar{\Delta}_i\le M \tilde{\Delta}_i, \quad t_{Z_i}-M(1-\tilde{\Delta}_i)\le\bar{\Delta}_i\le t_{Z_i} \text{ for } 1\le i \le n,\nonumber
 \\
 &\tilde{\Delta}_i \in\{0, 1\} \text{ for } 1\le i \le n, 
 \nonumber
 \\
 & \sum_{i=1}^{n_1} \bar{\omega}_{i} g(\bmX_i) = \sum_{i=n_1+1}^n \bar{\omega}_{i} g(\bmX_i), 
 \nonumber
 \\
& \sum_{i=1}^{n_1} \bar{\omega}_{i}=1, \quad \sum_{i=n_1+1}^n \bar{\omega}_{i}=1,
\nonumber
\quad t_1\ge 0, \quad t_0\ge 0,
\nonumber
\end{align}
where $a_i^{\low}$s and $a_i^{\up}$s are defined as in \eqref{eq:a_i_low_up}.
Note that $\tilde{\Delta}_i$ is restricted to be either 0 or 1 for $1\le i \le n$, and consequently $\bar{\Delta}_i$ is restricted to be either 0 or $t_{Z_i}$. This explains the equivalence between \eqref{eq:MILP_sep} and \eqref{eq:MILP_glover}. 

\subsection{Optimization for the bootstrap confidence intervals}\label{sec:boot_milp_appdix}

In this subsection, we provide the details of the optimization used to construct the bootstrap confidence intervals in Section \ref{sec:opt_per_boot}.

We first write \eqref{eq:double_frac_program_bootstrap_2} in a form analogous to \eqref{eq:frac_double}. 
Specifically, analogous to Section \ref{subsec:double_frac_program}, we can assume, without loss of generality, that $Z_1=\cdots = Z_{n_1} = 1$ and $Z_{n_1+1} = \cdots = Z_{n} = 0$, and 
equivalently write the optimization in \eqref{eq:double_frac_program_bootstrap_2} as 
the following one
over $w_{\bstrap i}$s and $\Delta_i$s: 
\begin{align}\label{eq:boot_frac_double}
    \min \text{ or } \max \quad &\dfrac{\sum_{i=1}^{n_1}k_{\bstrap i}w_{\bstrap i}Y_i}{\sum_{i=1}^{n_1}k_{\bstrap i}w_{\bstrap i}}-\dfrac{\sum_{i=n_1+1}^nk_{\bstrap i}w_{\bstrap i}Y_i}{\sum_{i=n_1+1}^nk_{\bstrap i}w_{\bstrap i}}\\
    \text{subject to} \quad & w_{\bstrap i} \in [w_{\bstrap i}^{\low}, w_{\bstrap i}^{\up}]\equiv[a_{\bstrap i}^{\low}\Delta_i, 1-(1-a_{\bstrap i}^{\up})\Delta_i]
    \nonumber
    \\&\sum_{i=1}^{n_1} \Delta_i\geq \lceil n_1(1-\delta_1)\rceil, \quad \sum_{i=n_1+1}^{n} \Delta_i\geq \lceil n_0(1-\delta_0)\rceil,
    \nonumber
    \\
    &\Delta_i = 0 \text{ or } 1 \text{ for all } i, 
    \nonumber
    \\
    & \dfrac{\sum_{i=1}^{n_1}k_{\bstrap i}w_{\bstrap i}g(\bmX_i)}{\sum_{i=1}^{n_1}k_{\bstrap i}w_{\bstrap i}}=\dfrac{\sum_{i=n_1+1}^nk_{\bstrap i}w_{\bstrap i}g(\bmX_i)}{\sum_{i=n_1+1}^nk_{\bstrap i}w_{\bstrap i}}.\nonumber
\end{align}
In \eqref{eq:boot_frac_double}, $\Delta_i$ indicates whether $|\psi_i|$ is bounded by $\log(\Lambda_{Z_i})$, $w_{\bstrap i}$ represents $1-\hat{e}_{\bstrap i}^{(\psi_i)}$ for treated unit $1\le i \le n_1$ and $\hat{e}_{\bstrap i}^{(\psi_i)}$ for control unit $n_1<i\le n$, 
and $a_{\bstrap i}^{\low}$ and $a_{\bstrap i}^{\up}$ are the lower and upper bounds of $w_{\bstrap i}$ when $\Delta_i=1$. 
Let 
$\underline{e}_{\bstrap i} = \logit^{-1}[-\log(\Lambda_{Z_i}) + \logit\{e(\bmX_i; \hat{\bm{\beta}}_{\bstrap})\}]$ and $\overline{e}_{\bstrap i} = \logit^{-1}[ \log(\Lambda_{Z_i}) + \logit\{e_{b}(\bmX_i; \hat{\bm{\beta}}_{\bstrap})\} ]$. 
We then have 
$
[a_{\bstrap i}^{\low}, a_{\bstrap i}^{\up}] = [ 1 - \overline{e}_{\bstrap i}, 1 - \underline{e}_{\bstrap i}]
$
for $1\le i \le n_1$
and 
$
[a_{\bstrap i}^{\low}, a_{\bstrap i}^{\up}] = [ \underline{e}_{\bstrap i}, \overline{e}_{\bstrap i}]
$
for $n_1<i\le n$. 
The optimization in \eqref{eq:boot_frac_double} can be viewed as a weighted version of that in \eqref{eq:frac_double}.

We then transform \eqref{eq:boot_frac_double} into a (mixed-integer) linear programming problem. 
Similar to Section \ref{sec:milp}, we apply the Charnes-Cooper transformation, and define
\begin{alignat*}{5}
   \bar{\omega}_{\bstrap i}&=\dfrac{k_{\bstrap i}\omega_{\bstrap i}}{\sum_{j:Z_j=Z_i}k_{\bstrap j}\omega_j},
    &\quad
    \bar{\Delta}_{\bstrap i}=\frac{\Delta_i}{\sum_{j:Z_j=Z_i}k_{\bstrap j}\omega_j},
    \ \ \ 
    \text{ and }
    \ \ \ 
    t_{\bstrap z}&=\dfrac{1}{\sum_{j:Z_j=z}k_{\bstrap j}\omega_j} \ \text{ for }
    z=0,1.
\end{alignat*}
Then 
\eqref{eq:boot_frac_double} becomes equivalent to the following 
optimization over  $\bar{\omega}_{\bstrap i}$s, $\bar{\Delta}_{\bstrap i}$s, $t_{\bstrap 1}$ and $t_{\bstrap 0}$:
\begin{align}\label{eq:boot_MILP_sep}
\min 
 \text{ or } \max \quad & \sum_{i=1}^{n_1} \bar{\omega}_{\bstrap i} Y_i-\sum_{i=n_1+1}^n \bar{\omega}_{\bstrap i} Y_i \\
\text{subject to} \quad 	&
k_{\bstrap i}a_{\bstrap i}^{\low} \bar{\Delta}_{i} \leq \bar{\omega}_{\bstrap i} \leq k_{\bstrap i}t_{\bstrap Z_i}-k_{\bstrap i}(1-a_{\bstrap i}^{\up})\bar{\Delta}_{\bstrap i}   \text{ for } 1\le i\le n,
\nonumber
\\
	&\sum_{i=1}^{n_1} \bar{\Delta}_{\bstrap i}\geq  
    \lceil n_1 (1-\delta_1)\rceil t_{\bstrap 1}, \quad 
    \sum_{i=n_1+1}^n \bar{\Delta}_{\bstrap i}\geq  
    \lceil n_0 (1-\delta_0)\rceil t_{\bstrap 0},
 \nonumber
 \\
 &\bar{\Delta}_{\bstrap i}\in\{0, t_{\bstrap Z_i}\} \text{ for } 1\le i \le n, 
 \nonumber
 \\
 & \sum_{i=1}^{n_1} \bar{\omega}_{\bstrap i} g(\bmX_i) = \sum_{i=n_1+1}^n \bar{\omega}_{\bstrap i} g(\bmX_i), 
 \nonumber
 \\
& \sum_{i=1}^{n_1} \bar{\omega}_{\bstrap i}=1, \quad \sum_{i=n_1+1}^n \bar{\omega}_{\bstrap i}=1,
\nonumber
\quad t_{\bstrap 1}\ge 0, \quad t_{\bstrap 0}\ge 0.
\nonumber
\end{align}
If we relax the constraints on $\bar{\Delta}_{\bstrap i}$s and instead require that $0\le \bar{\Delta}_{\bstrap i}\le t_{\bstrap Z_i}$ for all $1\le i \le n$, 
then the optimization in \eqref{eq:boot_MILP_sep} reduces to a linear programming problem. 
Importantly, the relaxation will still lead to valid bootstrap confidence intervals, often with little loss in conservativeness. 
Alternatively, we can linearize the constraints on $\bar{\Delta}_{\bstrap i}$s
using the \citet{GloverImproved:1975}’s linearization, as discussed in Section \ref{sec:glover_milp_appdix}, 
and then transform the optimization in \eqref{eq:boot_MILP_sep} to a mixed-integer linear programming problem.

\section{Robust marginal sensitivity analysis with bounds on the strength of unmeasured confounding across the whole population}\label{sec:whole}

\subsection{Robust marginal sensitivity model for the whole population}

In Assumptions \ref{asmp:robust_marg_sen} and \ref{asmp:robust_sen_spec}, we impose bounds on the strength of unmeasured confounding within the treated and control populations separately. 
In this section, we consider the robust marginal sensitivity model that imposes bounds on the strength of unmeasured confounding across the overall population, as discussed after Assumption \ref{asmp:robust_marg_sen} and in Section \ref{sec:diss}.

\begin{assumption}[Robust marginal sensitivity model for the whole population]\label{asmp:robust_marg_sen_whole} 
Recall the true propensity score $e_{\tp}(\bm{x},\bm{u})$ and the observable propensity score $e_{\tp}(\bm{x})$. 
For given $\Lambda \geq 1$ and $\delta \in [0,1]$, we have 
\begin{align}\label{eq:robust_sen_whole}
    \Pr \Big( \Lambda^{-1}\leq \dfrac{e_{\tp}(\bmX,\bmU)/\{1-e_{\tp}(\bmX,\bmU)\}}{e_{\tp}(\bmX)/\{1-e_{\tp}(\bmX)\}}\leq\Lambda \Big) & \geq 1-\delta.
\end{align}
\end{assumption}

Similar to Section \ref{sec:ci_for_overlap}, we also consider extension of Assumption \ref{asmp:robust_marg_sen_whole} to accommodate model misspecification.

\begin{assumption}[Robust marginal sensitivity model under specification for the whole population]\label{asmp:robust_sen_spec_whole}
Recall the true propensity score $e_{\tp}(\bm{x},\bm{u})$, and  define $e_{\m}(\bm{x})$ as an estimable propensity score based on certain model assumptions. 
For given parameters $\Lambda \geq 1$ and $0\leq\delta\leq 1$, we have 
\begin{align}\label{eq:para_robust_sen_whole}
	\Pr\left( \Lambda^{-1}\leq \dfrac{e_{\tp}(\bmX, \bmU)/\{1-e_{\tp}(\bmX, \bmU)\}}{e_\m(\bmX)/\{1-e_\m(\bmX)\}}\leq\Lambda \right) \geq 1-\delta. 
\end{align}	
\end{assumption}

In the following, we focus on constructing bootstrap confidence bounds for the average treatment effect under any given robust marginal sensitivity model in Assumption \ref{asmp:robust_sen_spec_whole}. 
The estimation of the bounds as in Section \ref{sec:point_sen_ana} under Assumption \ref{asmp:robust_marg_sen_whole} can be viewed as a special case of the optimization for the bootstrap sample with $(k_{\bstrap 1}, k_{\bstrap 2}, \ldots, k_{\bstrap n})=(1,1, \ldots, 1)$, and is thus omitted for conciseness. 
Similar to Section \ref{eq:multiple_quantile}, we also study simultaneous sensitivity analyses for the strength of unmeasured confounding across the whole population.

\subsection{Bootstrap inference for the overlap-weighted average treatment effect}

We use the same formulation as in Section \ref{sec:formu_z_est} and apply the method developed in Section \ref{sec:boot_over_iden} to construct bootstrap confidence bounds for the overlap-weighted average treatment effect $\tau_\ow$. 
Throughout this subsection, we assume Assumption \ref{asmp:robust_sen_spec_whole} for some given $\Lambda \ge 1$ and $\delta \in [0,1]$. 

We first construct prediction sets for $\bmW_{1:n}^\tp \equiv (W_1^\tp, \ldots, W_n^\tp) = (\psi_\m(\bmX_1, \bmU_1), \ldots, \psi_\m(\bmX_n, \bmU_n))$. 
From Assumption \ref{asmp:robust_sen_spec_whole} and by the definition of $\psi_\m(\cdot)$, we know that 
$\Pr\{ |W^\tp| \le \log(\Lambda) \} = \Pr\{ |\psi_\m(\bmX, \bmU)| \le \log(\Lambda) \}
\ge 1 - \delta$. 
This then motivates us to consider the following form of prediction sets for $W_i^\tp$s:   
for $\Lambda' \ge 1$ and $\delta' \in [0,1]$, 
\begin{align}\label{eq:Lambda_delta_w_pred_set_whole}
    \mathcal{W}_n(\Lambda', \delta') & = \Big\{ (W_1, \ldots, W_n) \in \mathbb{R}^{n}: \      
    \frac{1}{n} \sum_{i=1}^n \I \{ |W_i| \leq \log(\Lambda') \} \geq 1-\delta' 
    \Big\}. 
\end{align}

Under Assumption \ref{asmp:robust_sen_spec_whole} and letting $n\delta'=(1-\zeta)$th quantile of $\Bin(n,\delta)$,
we can verify that $\mathcal{W}_n(\Lambda, \delta')$ is a prediction set for $\bmW_{1:n}^\tp$ with coverage probability at least $1-\zeta$. 
In addition, for any $\varepsilon > 0$, the $\varepsilon$-neighborhood of $\mathcal{W}_n(\Lambda, \delta')$ is equivalently $\mathcal{W}_n(\Lambda', \delta')$ with $\log(\Lambda') \equiv \log(\Lambda) + \varepsilon$.

We then consider the bootstrap confidence bounds from \eqref{eq:bootstrap_epsilon_pred_set}. 
We consider the following optimization for bootstrapped samples:
for any $\Lambda'\ge1$ and $\delta'\in[0,1]$,
\begin{align}\label{eq:tau_boot_epsilon_pred_set_whole}
    \hat{\tau}_{\bstrap}^{\max} (\Lambda', \delta') = \ & \sup \tau 
    \quad 
    \textsc{or}
    \quad 
    \hat{\tau}_{\bstrap}^{\min} (\Lambda', \delta') = \inf \tau
    \\
    \text{subject to} \quad & 
    \sum_{i=1}^n k_{\bstrap i} p_{\bmtheta}(\bmO_{i}, W_i) =\bm{0}, 
    \nonumber
    \\
    & (W_1, W_2, \ldots, W_n) \in \mathcal{W}_n(\Lambda', \delta'),
    \nonumber
\end{align}
where $p_{{\bmtheta}}(\cdot)$ takes the form in \eqref{eq:est_equ_ow} with  
${\bmtheta} = (\tau, \bm{\beta}, \kappa_{\treat}, \kappa_{\control})$,  $\mathcal{W}_n(\Lambda', \delta')$ is defined as in \eqref{eq:Lambda_delta_w_pred_set_whole}, 
and 
$(k_{\bstrap 1}, k_{\bstrap 2}, \ldots, k_{\bstrap n}) \sim \text{Multinomial} (n; n^{-1}, n^{-1}, \ldots, n^{-1})$.
Define further 
\begin{align}\label{eq:tau_boot_quantile_gen_whole}
    \LL_{\alpha}(\Lambda', \delta') \equiv Q_{\alpha} \big\{ \hat{\tau}_{\bstrap}^{\min} (\Lambda', \delta')  \big\}
    \quad 
    \text{and}
    \quad 
    \UU_{\alpha} (\Lambda', \delta') \equiv Q_{1-\alpha} \big\{ \hat{\tau}_{\bstrap}^{\max} (\Lambda', \delta') \big\}
\end{align}
to denote the $\alpha$th lower and upper quantiles of the bootstrap distributions of the minimum and maximum objective values in \eqref{eq:tau_boot_epsilon_pred_set_whole}, respectively. 

The theorem and corollary below then follow immediately from Theorem \ref{thm:bootstrap_over_identify_pred_set}, which provide confidence bounds for the overlap-weighted average treatment effect under the robust sensitivity model in Assumption \ref{asmp:robust_sen_spec_whole}.  
Recall that $\bmW_{1:n}^\tp = (W_1^\tp, W_2^\tp, \ldots, W_n^\tp)$ with $W_i^\tp$ denoting the difference between the true and estimable propensity scores in the logit scale. 

\begin{theorem}\label{thm:ci_tau_pred_whole}
    Under Assumption \ref{asmp:robust_sen_spec_whole}  for the robust marginal sensitivity model with some given  $(\Lambda, \delta)$ and Condition \ref{cond:logistic_Z_est}, 
    for any $\alpha \in (0,1)$ and $\Lambda'>\Lambda$, and any $\delta_n \in [0,1]$ that can vary with the sample size $n$, 
    \begin{align*}
        \limsup_{n\rightarrow \infty} \Pr\{ \tau_\ow > \UU_{\alpha}(\Lambda', \delta_n), \  \bmW_{1:n}^\tp \in \mathcal{W}_n(\Lambda, \delta_n) \} & \le \alpha, \\
        \limsup_{n\rightarrow \infty} \Pr\{ \tau_\ow < \LL_{\alpha}(\Lambda', \delta_n), \  \bmW_{1:n}^\tp \in \mathcal{W}_n(\Lambda, \delta_n) \} & \le \alpha,
    \end{align*}
    and consequently 
    $
        \limsup_{n\rightarrow \infty} \Pr\{ \tau_\ow \notin [\LL_{\alpha}(\Lambda', \delta_n), \UU_{\alpha}(\Lambda', \delta_n)], \  \bmW_{1:n}^\tp \in \mathcal{W}_n(\Lambda, \delta_n) \}  \le 2 \alpha,
    $
    where $\LL_{\alpha}(\Lambda', \delta_n)$ and $\UU_{\alpha}(\Lambda', \delta_n)$ are defined as in \eqref{eq:tau_boot_quantile_gen_whole}, and $\mathcal{W}_n(\Lambda, \delta_n)$ is defined as in \eqref{eq:Lambda_delta_w_pred_set_whole}. 
\end{theorem}

\begin{corollary}\label{cor:ci_tau_whole}
    Consider the same setting as Theorem \ref{thm:ci_tau_pred_whole}. 
    Let $n \delta_n$ be the $1-\zeta$ quantile of $\Bin(n,\delta)$. 
    Then, for any given $\Lambda'>\Lambda$, 
    $\UU_{\alpha}(\Lambda', \delta_n)$ and $\LL_{\alpha}(\Lambda', \delta_n)$ are, respectively, asymptotic $1-(\alpha+\zeta)$ upper and lower confidence bounds for $\tau_\ow$: 
    \begin{align*}
        \limsup_{n\rightarrow \infty} \Pr\{ \tau_\ow > \UU_{\alpha}(\Lambda', \delta_n) \} & \le \alpha + \zeta, 
        \qquad 
        \limsup_{n\rightarrow \infty} \Pr\{ \tau_\ow < \LL_{\alpha}(\Lambda', \delta_n) \}  \le \alpha + \zeta,
    \end{align*}
    and $[\UU_{\alpha}(\Lambda', \delta_n), \LL_{\alpha}(\Lambda', \delta_n)]$ is asymptotic $1-(2\alpha+\zeta)$ confidence intervals for $\tau_\ow$: 
    \begin{align*}
        \limsup_{n\rightarrow \infty} \Pr\{ \tau_\ow \notin [\LL_{\alpha}(\Lambda', \delta_n), \UU_{\alpha}(\Lambda', \delta_n)] \}  \le 2 \alpha + \zeta.
    \end{align*}
\end{corollary}

In practice, we can choose $\alpha$ and $\zeta$ based on the desired confidence level, as discussed after Corollary \ref{cor:ci_tau}.

Below we construct narrower confidence bounds for the overlap-weighted average treatment effect under the robust sensitivity model in Assumption \ref{asmp:robust_sen_spec_whole}, 
which build upon the confidence bounds in Section \ref{sec:ci_for_overlap} under sensitivity models that impose separate constraints for treated and control populations. 
Note that the set $\mathcal{W}_n(\Lambda, \delta_n)$ has the following equivalent form: 
\begin{align*}
    \mathcal{W}_n(\Lambda, \delta_n)=\bigcup_{\vec{\bm{\delta}}_n:\lfloor n_1\delta_{n,1}\rfloor+\lfloor n_0\delta_{n,0}\rfloor=\lfloor n\delta_n\rfloor}\mathcal{W}_n(\Lambda\vec{\bm{1}}, \vec{\bm{\delta}}_n)
    =\bigcup_{i = \max\{0, \lfloor n\delta_n\rfloor-n_0\}}^{\min\{\lfloor n\delta_n\rfloor, n_1\}}\mathcal{W}_n\left(\Lambda\vec{\bm{1}}, \left(\frac{i}{n_1}, \frac{\lfloor n\delta_n\rfloor-i}{n_0}\right)\right), 
\end{align*}
where $\vec{\bm{1}} = (1,1), \vec{\bm{\delta}}_n=(\delta_{n,1}, \delta_{n,0})$,  and $\mathcal{W}_n(\Lambda\vec{\bm{1}}, \vec{\bm{\delta}}_n)$ and $\mathcal{W}_n(\Lambda\vec{\bm{1}}, ({i}/{n_1}, ({\lfloor n\delta_n\rfloor-i})/{n_0}))$ are defined as in \eqref{eq:Lambda_delta_w_pred_set}.

\begin{theorem}\label{thm:ci_tau_pred_whole_sharper}
    Under Assumption \ref{asmp:robust_sen_spec_whole}  for the robust marginal sensitivity model with some given  $(\Lambda, \delta)$ and Condition \ref{cond:logistic_Z_est}, 
    for any $\alpha \in (0,1)$ and $\Lambda'>\Lambda$, and any $\delta_n \in [0,1]$ that can vary with the sample size $n$, 
    \begin{align*}
        \limsup_{n\rightarrow \infty} \Pr\left[ \tau_\ow > \tilde{\UU}_{\alpha}(\Lambda', \delta_n), \  \bmW_{1:n}^\tp \in \mathcal{W}_n(\Lambda, \delta_n) \right] & \le \alpha, \\
        \limsup_{n\rightarrow \infty} \Pr\left[ \tau_\ow < \tilde{\LL}_{\alpha}(\Lambda', \delta_n), \  \bmW_{1:n}^\tp \in \mathcal{W}_n(\Lambda, \delta_n) \right] & \le \alpha,
    \end{align*}
    where 
    \begin{align*}
        \tilde{\UU}_{\alpha}(\Lambda', \delta_n)
        & \equiv 
        \max_{i \in \mathcal{I}_n} \left\{ \UU_{\alpha}\left(\Lambda'\vec{\bm{1}}, \left( \frac{i}{n_1}, \frac{\lfloor n\delta_n\rfloor-i}{n_0} \right)\right)\right\}
        \le 
        \UU_{\alpha}(\Lambda', \delta_n), 
        \\
        \tilde{\LL}_{\alpha}(\Lambda', \delta_n) 
        & \equiv 
        \min_{i \in \mathcal{I}_n} \left\{ \LL_{\alpha}\left(\Lambda'\vec{\bm{1}}, \left( \frac{i}{n_1}, \frac{\lfloor n\delta_n\rfloor-i}{n_0} \right)\right)\right\} 
        \ge 
        \LL_{\alpha}(\Lambda', \delta_n), 
    \end{align*}
    and 
    $\mathcal{I}_n$ is the set of integers between $\max\{0, \lfloor n\delta_n\rfloor-n_0\}$ and $\min\{\lfloor n\delta_n\rfloor, n_1\}$. 
\end{theorem}

Theorem \ref{thm:ci_tau_pred_whole_sharper} is indeed related to a general strategy that can be used to narrow the bootstrap confidence bounds. We discuss this general strategy in Section \ref{sec:narrow_partition}. 
The validity of the confidence bounds from Theorem \ref{thm:ci_tau_pred_whole_sharper} and the fact that they are narrower than (or equal to) those from Theorem \ref{thm:ci_tau_pred_whole} follow immediately from the discussion in Section \ref{sec:narrow_partition}. 

\subsection{Optimization for the bootstrap confidence intervals}

In this subsection, we study how to compute the confidence bounds in Theorem \ref{thm:ci_tau_pred_whole} and Corollary \ref{cor:ci_tau_whole}, 
which all rely on the optimization in \eqref{eq:tau_boot_epsilon_pred_set_whole} based on the bootstrapped samples. 
Note that the computation for the confidence bounds from 
Theorem \ref{thm:ci_tau_pred_whole_sharper} has already been discussed in Sections \ref{sec:opt_per_boot} and \ref{sec:boot_milp_appdix}.

In the following, we focus on the computation of  
$\hat{\tau}_{\bstrap}^{\max} (\Lambda, \delta)$ and $\hat{\tau}_{\bstrap}^{\min} (\Lambda, \delta)$ defined as in \eqref{eq:tau_boot_epsilon_pred_set_whole} for any $\Lambda\ge 1$ and $\delta \in [0,1]$. 
To facilitate the discussion, 
we adopt the notation from Section \ref{sec:opt_per_boot}. 
The optimization for $\hat{\tau}_{\bstrap}^{\min}(\Lambda, \delta)$ and $\hat{\tau}_{\bstrap}^{\max}(\Lambda, \delta)$ as in \eqref{eq:tau_boot_epsilon_pred_set_whole} can be equivalently written in the following form:
\begin{align}\label{eq:double_frac_program_bootstrap_2_whole}
    \min 
    \text{ or } \max \quad 
    & 
    \dfrac{\sum_{i=1}^n k_{\bstrap i} Z_{i} (1-
    \hat{e}_{\bstrap i}^{(\psi_i)}
    )Y_{i}}{\sum_{i=1}^n k_{bi} Z_{i} (1-
    \hat{e}_{\bstrap i}^{(\psi_i)}
    ) }-\dfrac{\sum_{i=1}^n k_{\bstrap i}(1-Z_{i}) \hat{e}_{\bstrap i}^{(\psi_i)} Y_{i}}{\sum_{i=1}^n k_{\bstrap i} (1-Z_{bi}) \hat{e}_{\bstrap i}^{(\psi_i)}},
    \\
    \text{subject to} \quad 	&
    \sum_{i=1}^n \I \{ |\psi_i | \leq \log(\Lambda) \} \ge \lceil n
    (1-\delta)\rceil, 
    \nonumber
    \\
    & \dfrac{\sum_{i=1}^n k_{\bstrap i} Z_{i} (1-
    \hat{e}_{\bstrap i}^{(\psi_i)}
    ) g(\bmX_{i})}{\sum_{i=1}^n k_{\bstrap i} Z_{i} (1-
    \hat{e}_{\bstrap i}^{(\psi_i)}
    ) }
    = 
    \dfrac{\sum_{i=1}^n k_{\bstrap i}(1-Z_{i}) \hat{e}_{\bstrap i}^{(\psi_i)} g(\bmX_{i})}{\sum_{i=1}^n 
    k_{\bstrap i}(1-Z_{i}) \hat{e}_{\bstrap i}^{(\psi_i)}
    },
    \nonumber
\end{align}
and $\hat{\tau}_{\bstrap}^{\min}(\Lambda, \delta)$ and $\hat{\tau}_{\bstrap}^{\max}(\Lambda, \delta)$ are equivalently the minimum and maximum objective values from \eqref{eq:double_frac_program_bootstrap_2_whole}.

We assume, without loss of generality, that $Z_1=\cdots = Z_{n_1} = 1$ and $Z_{n_1+1} = \cdots = Z_{n} = 0$. 
By the same logic as Section \ref{sec:boot_milp_appdix}, we can 
 write the optimization in \eqref{eq:double_frac_program_bootstrap_2_whole} equivalently as: 
\begin{align}\label{eq:boot_frac_double_whole}
    \min \text{ or } \max \quad &\dfrac{\sum_{i=1}^{n_1}k_{\bstrap i}w_{\bstrap i}Y_i}{\sum_{i=1}^{n_1}k_{\bstrap i}w_{\bstrap i}}-\dfrac{\sum_{i=n_1+1}^nk_{\bstrap i}w_{\bstrap i}Y_i}{\sum_{i=n_1+1}^nk_{\bstrap i}w_{\bstrap i}}\\
    \text{subject to} \quad & w_{\bstrap i} \in [w_{\bstrap i}^{\low}, w_{\bstrap i}^{\up}]\equiv[a_{\bstrap i}^{\low}\Delta_i, 1-(1-a_{\bstrap i}^{\up})\Delta_i]
    \nonumber
    \\&\sum_{i=1}^n \Delta_i\geq \lceil n(1-\delta)\rceil, \quad \Delta_i = 0 \text{ or } 1 \text{ for all } i, 
    \nonumber
    \\
    & \dfrac{\sum_{i=1}^{n_1}k_{\bstrap i}w_{\bstrap i}g(\bmX_i)}{\sum_{i=1}^{n_1}k_{\bstrap i}w_{\bstrap i}}=\dfrac{\sum_{i=n_1+1}^nk_{\bstrap i}w_{\bstrap i}g(\bmX_i)}{\sum_{i=n_1+1}^nk_{\bstrap i}w_{\bstrap i}},\nonumber
\end{align}
where the optimization is over $\Delta_i$s and $w_{\bstrap i}$s, 
and the variables are defined the same as those in \eqref{eq:boot_frac_double}.

\subsubsection{Multiple mixed-integer linear programming for the sensitivity bounds}\label{sec:multiple_linear_appdix}

 Note that in \eqref{eq:boot_frac_double_whole}, 
the maximum or minimum objective value can always be obtained when $\sum_{i=1}^n \Delta_i =  \lceil n(1-\delta)\rceil$, since reducing $\Delta_i$ will increase or preserve the feasible region for $w_{\bstrap i}$. 
Thus, it is equivalent to consider the optimization in \eqref{eq:boot_frac_double_whole} with the constraint that $\sum_{i=1}^n \Delta_i =  \lceil n(1-\delta)\rceil$. 
To solve it,
it suffices to consider the following programming problems for all integer $l\in
[\max\{0, \lceil n(1-\delta)\rceil-n_0\}, \ \min\{n_1, \lceil n(1-\delta)\rceil\}]$ :
\begin{align}\label{eq:boot_frac_double_whole_milp}
    \min \text{ or } \max \quad &\dfrac{\sum_{i=1}^{n_1}k_{\bstrap i}w_{\bstrap i}Y_i}{\sum_{i=1}^{n_1}k_{\bstrap i}w_{\bstrap i}}-\dfrac{\sum_{i=n_1+1}^nk_{\bstrap i}w_{\bstrap i}Y_i}{\sum_{i=n_1+1}^nk_{\bstrap i}w_{\bstrap i}}\\
    \text{subject to} \quad & w_{\bstrap i} \in [w_{\bstrap i}^{\low}, w_{\bstrap i}^{\up}]\equiv[a_{\bstrap i}^{\low}\Delta_i, 1-(1-a_{\bstrap i}^{\up})\Delta_i]
    \nonumber
    \\&\sum_{i=1}^{n_1} \Delta_i =  l, \quad \sum_{i=n_1+1}^{n} \Delta_i =  \lceil n(1-\delta)\rceil-l, \quad\Delta_i = 0 \text{ or } 1 \text{ for all } i, 
    \nonumber
    \\
    & \dfrac{\sum_{i=1}^{n_1}k_{\bstrap i}w_{\bstrap i}g(\bmX_i)}{\sum_{i=1}^{n_1}k_{\bstrap i}w_{\bstrap i}}=\dfrac{\sum_{i=n_1+1}^nk_{\bstrap i}w_{\bstrap i}g(\bmX_i)}{\sum_{i=n_1+1}^nk_{\bstrap i}w_{\bstrap i}}.\nonumber
\end{align}
By the same logic as before, it is equivalent to consider \eqref{eq:boot_frac_double_whole_milp} with the constraints that $\sum_{i=1}^{n_1} \Delta_i \ge   l$ and $\sum_{i=n_1+1}^{n} \Delta_i \ge  \lceil n(1-\delta)\rceil-l$. 
Obviously, the resulting optimization has the same form as that in \eqref{eq:boot_frac_double}, and can be solved in the same way as described in Section \ref{sec:boot_milp_appdix}.

\subsubsection{Quadratic programming with integer constraints}\label{sec:quad_prog}

Consider the following Charnes-Cooper transformation: 
\begin{alignat*}{5}
   \bar{\omega}_{\bstrap i}&=\dfrac{k_{\bstrap i}\omega_{\bstrap i}}{\sum_{j:Z_j=Z_i}k_{\bstrap j}\omega_j}
   \quad
    \text{and}
    \quad
    t_{\bstrap z}&=\dfrac{1}{\sum_{j:Z_j=z}k_{\bstrap j}\omega_j} \ \text{ for }
    z=0,1.
\end{alignat*}
We can then transform the optimization in
\eqref{eq:boot_frac_double_whole} into the following 
one over  $\bar{\omega}_{\bstrap i}$s, $\Delta_i$s, $t_{\bstrap 1}$ and $t_{\bstrap 0}$:
\begin{align}\label{eq:boot_QCLP_joint}
\min 
 \text{ or } \max \quad & \sum_{i=1}^{n_1} \bar{\omega}_{\bstrap i} Y_i-\sum_{i=n_1+1}^n \bar{\omega}_{\bstrap i} Y_i \\
\text{subject to} \quad 	&
k_{\bstrap i}a_{\bstrap i}^{\low} \Delta_{i}t_{\bstrap Z_i} \leq \bar{\omega}_{\bstrap i} \leq k_{\bstrap i}t_{\bstrap Z_i}-k_{\bstrap i}(1-a_{\bstrap i}^{\up})\Delta_{i}t_{\bstrap Z_i}   \text{ for } 1\le i\le n,
\nonumber
\\
	&\sum_{i=1}^{n} \Delta_{i}\geq  
    \lceil n (1-\delta)\rceil, \quad \Delta_i=0 \text{ or }1 \text{ for all } i,
 \nonumber
 \\
 & \sum_{i=1}^{n_1} \bar{\omega}_{\bstrap i} g(\bmX_i) = \sum_{i=n_1+1}^n \bar{\omega}_{\bstrap i} g(\bmX_i), 
 \nonumber
 \\
& \sum_{i=1}^{n_1} \bar{\omega}_{\bstrap i}=1, \quad \sum_{i=n_1+1}^n \bar{\omega}_{\bstrap i}=1,
\nonumber
\quad t_{\bstrap 1}\ge 0, \quad t_{\bstrap 0}\ge 0,
\nonumber
\end{align}
which becomes a quadratically constrained linear programming (QCLP) problem with integer constraints. 
We can also relax the integer constraints on $\Delta_i$s and instead require them to be in the unit interval $[0,1]$, i.e., $0\le \Delta_i \le 1$ for all $i$. 
The QCLP problem can be solved using standard software such as Gurobi \citep{gurobi}, usually within a reasonable time once the integer constraints are relaxed. 
Importantly, the relaxed optimization can still lead to valid but potentially more conservative confidence bounds. 
From the numerical experiment, relaxing the integer constraints often lead to little loss in conservativeness.

Below we discuss an ad hoc procedure to efficiently get a lower (or upper) bound for the maximum (or minimum) objective value from \eqref{eq:boot_QCLP_joint}, which can be used to assess how conservative the relaxation of integer constraints can be. 
The main idea is to 
slightly adjust the solution of $\Delta_i$s from the relaxed QCLP. 
Specifically, let $\bar{\Delta}_i$s be the values of $\Delta_i$s from a solution of the relaxed QCLP that, say, maximizes the objective value in \eqref{eq:boot_QCLP_joint}. 
We can verify that $\sum_{i=1}^n \bar{\Delta}_i$ takes an integer value; otherwise we can always decrease certain $\Delta_i$s without changing the objective value, since the ranges of the $w_{1i}$s and $w_{0i}$s in \eqref{eq:boot_QCLP_joint} decrease with the $\Delta_i$s. 
Suppose that $\sum_{i=1}^n \bar{\Delta}_i=k$ for some integer $0\le k \le n$.
We then define $\bar{\Delta}_j'=1$ if $\bar{\Delta}_j$ is among the largest $k$ elements of $\bar{\Delta}_i$s, and $0$ otherwise. 
Finally, we fix the values of $\Delta_i$s at $\bar{\Delta}_j'$, and solve the optimization in \eqref{eq:boot_QCLP_joint} trying to maximize the objective value. Importantly, with the fixed $\Delta_i$s, the quadratic programming problem reduces to a linear programming problem and thus can be solved efficiently.
Denote the resulting maximum objective value by $\bar{U}'$. 
Following the same procedure we can also obtain $\bar{L}'$ 
for the minimum objective value. 
It is not difficult to see that $[\bar{L}', \bar{U}'] \subset [L, U]$, where $L$ and $U$ denote the optimized objective values from the origial QCLP problem. This is because the solution for $\bar{L}'$ and $\bar{U}'$ also fall in the feasible region of the original QCLP problem.

\subsection{Simultaneous sensitivity analyses}

In this subsection, we assume that the overlap-weighted average treatment effect is zero, i.e., $\tau_{\ow} = 0$, and use bootstrap two-sided confidence intervals of form $[\LL_{\alpha}(\Lambda, \delta), \UU_{\alpha}(\Lambda, \delta)]$ in Theorem \ref{thm:ci_tau_pred_whole} to conduct inference for the strength of unmeasured confounding. 
Similar logic applies when we assume $\tau_{\ow} \ge 0$ or $\tau_{\ow} \le 0$, and use upper or lower bootstrap confidence bounds.

For $0\le q\le 1$, 
let $\OR_{\m[q]}$ denote the $q$th sample quantile of the strength of unmeasured confounding within the sample.
Note that, by definition, $\bmW_{1:n}^\tp \in \mathcal{W}_n(\Lambda, \delta)$ 
if and only if 
$\OR_{\m[1-\delta]} \le \Lambda$. 
From the equivalent understanding of sensitivity analysis in Section \ref{sec:equiv_view_sa} and Theorems \ref{thm:ci_tau_pred_whole} and \ref{thm:ci_tau_pred_whole_sharper}, 
we can immediately derive asymptotically valid tests for the following null hypotheses about the sample quantiles of unmeasured confounding strength: 
\begin{equation}\label{eq:H_samp_U_whole}
    H_{q,\Lambda}^{\text{U}}: \OR_{\m[q]} \le \Lambda, 
    \quad
    \text{where $q\in [0,1]$ and $\Lambda \in [1, \infty]$}. 
\end{equation}

\begin{corollary}\label{cor:test_U_whole}
    Assume that Condition \ref{cond:logistic_Z_est} holds and the true overlap-weighted average treatment effect is zero, i.e., $\tau_{\ow} = 0$. 
    Consider the null hypothesis $H_{q,\Lambda}^{\text{U}}$ in \eqref{eq:H_samp_U_whole} for any $q \in [0,1]$ and $\Lambda\in [1, \infty]$. Let $\xi > 1$ be any fixed constant. 
    For any $\alpha \in (0,1)$ and $\xi > 1$, 
    the test that rejects $H_{q,\Lambda}^{\text{U}}$ if and only if $0\notin [\LL_{\alpha/2}(\xi \Lambda, 1-q), \UU_{\alpha/2}(\xi \Lambda, 1-q)]$ is an asymptotically valid level-$\alpha$ test for the null hypothesis $H_{q,\Lambda}$,    
    in the sense that 
    $$
    \limsup_{n\rightarrow \infty} \Pr\big\{
    H_{q,\Lambda}^{\text{U}} \text{ holds and } 0\notin [\LL_{\alpha/2}(\xi \Lambda, 1-q), \UU_{\alpha/2}(\xi \Lambda, 1-q)]
    \big\} \le \alpha.
    $$
\end{corollary}

Similar to Section \ref{eq:multiple_quantile}, by test inversion, we can then construct prediction sets for quantiles of unmeasured confounding strength within the sample. Importantly, such prediction sets will be simultaneously valid across all quantiles, which can strengthen our causal evidence from observational studies.
We summarize the results in the following theorem. 
Define 
\begin{align*}
        \mathcal{I}_{q}^\alpha \equiv \big\{ \Lambda \in [1, \infty]: 
        0 \in [\LL_{\alpha/2}(\Lambda, 1-q), \UU_{\alpha/2}(\Lambda, 1-q)] \big\}.
\end{align*}
Define further $\mathcal{I}_{q, \xi}^\alpha = \xi^{-1} \mathcal{I}_{q}^\alpha = \{\xi^{-1}\Lambda: \Lambda \in \mathcal{I}_{q}^\alpha\}$ for $\xi \ge 1$.

\begin{theorem}\label{thm:U_quant_pred_int_whole}
    Assume that Condition \ref{cond:logistic_Z_est} holds and the true overlap-weighted average treatment effect is zero, i.e., $\tau_{\ow} = 0$. 
    Let $\alpha \in (0,1)$ and $\xi >1$ be any fixed constants. 
    Then, for any $q\in [0,1]$,  $\mathcal{I}_{q, \xi}^\alpha$ is an asymptotic $1-\alpha$ prediction set for $\OR_{\m[q]}$. Moreover, it is simultaneously valid across all quantiles, in the sense that 
    $
        \liminf_{n \rightarrow \infty} \Pr\{ \OR_{\m[q]} \in \mathcal{I}_{q, \xi}^\alpha \text{ for all } q \in [0,1]\} \ge 1-\alpha. 
    $
\end{theorem}

The simultaneous prediction bounds in Theorem \ref{thm:U_quant_pred_int_whole} essentially provide a prediction band for the whole quantile function (and thus distribution function) of the unmeasured confounding strength within the sample. 
Consequently, they can also tell the proportions of units within the sample whose unmeasured confounding strength is greater than any thresholds, 
i.e., $n^{-1} \sum_{i=1}^n \I\{\OR_\m (\bmX_i, \bmU_i) > \Lambda \}$ for any $\Lambda \ge 1$. 
In addition, by the same logic as Remark \ref{remark:gen_sam_pop}, we can also use Theorem \ref{thm:U_quant_pred_int_whole} to construct (simultaneous) confidence intervals for (multiple) quantiles of unmeasured confounding strength in the whole population.

\begin{remark}
   In both Corollary \ref{cor:test_U_whole} and Theorem \ref{thm:U_quant_pred_int_whole}, the confidence bounds $\LL_{\alpha/2}(\cdot, \cdot)$ and $\UU_{\alpha/2}(\cdot, \cdot)$ can be replaced by the narrower bounds $\tilde{\LL}_{\alpha/2}(\cdot, \cdot)$ and $\tilde{\UU}_{\alpha/2}(\cdot, \cdot)$ from Theorem \ref{thm:ci_tau_pred_whole_sharper}, and the results will still hold.
   Moreover, the simultaneous prediction sets in Theorem \ref{thm:U_quant_pred_int_whole} essentially form a superset of the simultaneous prediction sets in Theorem \ref{thm:U_quant_pred_int}; this is analogous to the prediction sets in Corollary \ref{cor:U_quant_pred_int_max_tr_co}. 
   See Section \ref{sec:proof_simul_whole} for the technical details. 
\end{remark}

\section{Narrower bootstrap confidence bounds by decomposing the prediction set for unmeasured variables}\label{sec:narrow_partition}

Consider the same setting as in Section \ref{sec:pred_set_W} and suppose that the prediction set $\mathcal{W}_n$ can be decomposed as 
$\mathcal{W}_n=\bigcup_{j\in \mathcal{J}_n} \mathcal{W}_{n,j}$, where $\mathcal{W}_{n,j}$s may overlap with each other.  
Let $\mathcal{W}_{n,j}^{\varepsilon}$ be an $\varepsilon$-neighborhood of the set $\mathcal{W}_{n,j}$:
\begin{align*}
    \mathcal{W}_{n,j}^\varepsilon \equiv 
    \{
    (\bmW_1, \ldots, \bmW_n) \in (\mathbb{R}^{d_1})^{\otimes n}: 
    \max_{1\le i \le n} \|\bmW_i - \bmW_i'\| \le \varepsilon \text{ for some } (\bmW_1', \ldots, \bmW_n') \in \mathcal{W}_{n,j} 
    \}. 
\end{align*}
Similar to \eqref{eq:bootstrap_epsilon_pred_set}, we consider the following optimization based on bootstrapped samples for each $j\in\mathcal{J}_n$: 
\begin{align}\label{eq:bootstrap_epsilon_pred_set_sharp}
    \hat{\theta}_{H, \bstrap}^{(j)} = \ & \sup H({\bmtheta}) 
    \\
    \text{subject to} \quad & 
    \sum_{i=1}^n k_{\bstrap i} p_{\bmtheta}(\bmO_{i}, \bmW_i) =\bm{0},
    \nonumber
    \\
    & (\bmW_1, \bmW_2, \ldots, \bmW_n) \in \mathcal{W}_{n,j}^\varepsilon.
    \nonumber
\end{align}

\begin{theorem}\label{thm:bootstrap_over_identify_pred_set_sharp}
    Under the same setting as in Theorem \ref{thm:bootstrap_over_identify},  
    for any $\varepsilon > 0$, 
    \begin{align*}
        \limsup_{n\rightarrow \infty} \Pr[ H({\bmtheta}_0) > \max_{j\in\mathcal{J}_n}\{Q_{1-\alpha} ( \hat{\theta}_{H, \bstrap}^{(j)} )\}, \  (\bmW_1^\tp, \bmW_2^\tp, \ldots, \bmW_n^\tp) \in \mathcal{W}_n ] \le \alpha, 
    \end{align*}
    where $\hat{\theta}_{H, \bstrap}^{(j)}$ is defined in \eqref{eq:bootstrap_epsilon_pred_set_sharp}.
    Moreover,
    \begin{align}\label{eq:narrow_decomp}
        \max_{j\in\mathcal{J}_n}\big\{Q_{1-\alpha} ( \hat{\theta}_{H, \bstrap}^{(j)} ) \big\}\le Q_{1-\alpha}\big\{ \max_{j\in\mathcal{J}_n}\hat{\theta}_{H, \bstrap}^{(j)} \big\}= Q_{1-\alpha} ( \hat{\theta}_{H, \bstrap} ),
    \end{align}
    where $\hat{\theta}_{H, \bstrap}$ is defined in \eqref{eq:bootstrap_epsilon_pred_set}. 
\end{theorem}

Theorem \ref{thm:bootstrap_over_identify_pred_set_sharp} shows that $\max_{j\in\mathcal{J}_n}\{Q_{1-\alpha} ( \hat{\theta}_{H, \bstrap}^{(j)} )\}$ can lead to a narrower while still valid confidence bounds for $H({\bmtheta}_0)$. 
However, we note that the computation for $\max_{j\in\mathcal{J}_n}\{Q_{1-\alpha} ( \hat{\theta}_{H, \bstrap}^{(j)} )\}$ generally involves higher computational cost. 
Compared to the original confidence bound $Q_{1-\alpha} ( \hat{\theta}_{H, \bstrap} )$, 
we need to solve $|\mathcal{J}_n|$ times more optimization problems of a similar form, where $|\mathcal{J}_n|$ denotes the cardinality of the set $\mathcal{J}_n$.

\section{Additional numeric illustration and experiments}
In this section, we conduct simulations based on i.i.d.~samples from the following model: 
\begin{align}\label{eq:data_gen}
    X_i &\sim \text{Uniform}([0,1]),\\
    U_i\mid X_i &\sim 
    \begin{cases}
        \text{Uniform}([0,1]), & \text{if } 0\le X_i \le 0.7, \\
        \text{Uniform}([0,100]), & \text{if } 0.7< X_i \le 1,
    \end{cases}\nonumber\\
    Z_i\mid X_i,U_i &\sim \text{Bernoulli}\left(\frac{1}{1+\exp(X_i-0.1  U_i)}\right),\nonumber\\
    Y_i(0)=2&X_i+3U_i, \qquad Y_i(1)=Y_i(0)+5.\nonumber
\end{align}

\subsection{Visualization of prediction sets}

We generate data from the model in \eqref{eq:data_gen} with sample size $n = 1000$. 
Suppose that we are interested in the $90\%$ and $95\%$ sample quantiles of unmeasured confounding strength among the treated and control units, respectively, i.e., $(\OR^\treat_{\m[q_1]}, \OR^\control_{\m[q_0]})$ with  $\vec{\bm{q}}=(q_1,q_0)=(0.9, 0.95)$. 
From Theorem \ref{thm:U_quant_pred_int}, 
an asymptotic $1-\alpha$ prediction set for $(\OR^\treat_{\m[q_1]}, \OR^\control_{\m[q_0]})$ is 
$\mathcal{I}_{\vec{\bm{q}}, \xi}^\alpha = \xi^{-1} \mathcal{I}_{\vec{\bm{q}}}^\alpha$, where $\xi$ can be any constant greater than $1$ and  $\mathcal{I}_{\vec{\bm{q}}}^\alpha$ is defined in \eqref{eq:I_q_vec_alpha}. 
For convenience, we set $\xi = 1$, and present this prediction set $\mathcal{I}_{\vec{\bm{q}}}^\alpha$ for $\alpha = 0.05$ in 
Figure \ref{fig:Visualizing_the_Structure_of_Iq}.

\begin{figure}[htbp]
    \centering\includegraphics[width=0.6\textwidth]{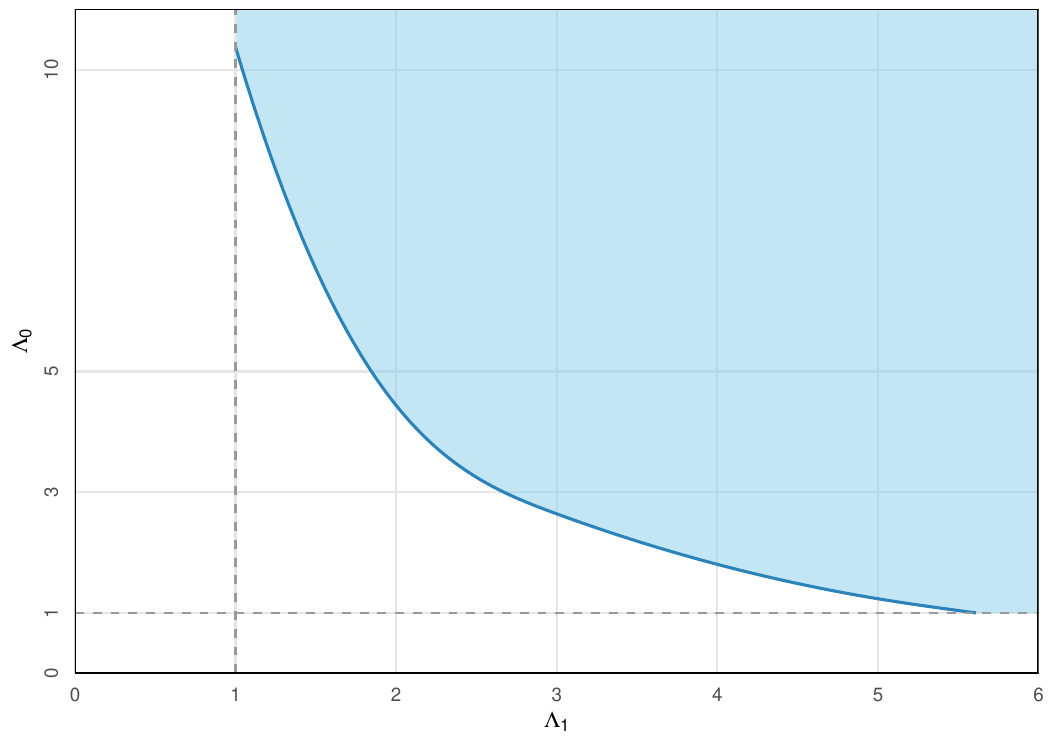}
    \caption{A $95\%$ prediction set for the sample quantiles of unmeasured confounding strength $(\OR^\treat_{\m[q_1]}, \OR^\control_{\m[q_0]})$, represented by the sky-blue shaded region represents.}
    \label{fig:Visualizing_the_Structure_of_Iq}
\end{figure}

\subsection{Sensitivity of confidence bounds to small variations in \texorpdfstring{$\vec{\bm{\Lambda}}$}{vector Lambda} or \texorpdfstring{$\Lambda$}{Lambda}}

We consider first the sensitivity analysis in Section \ref{sec:ci_for_overlap} with bounds on the strength of unmeasured confounding in both the treated and control populations. 
We generate $n=1000$ i.i.d.~samples from \eqref{eq:data_gen}, 
and consider the bounds $[\LL_{\alpha}(\vec{\bm{\Lambda}}, \vec{\bm{\delta}}), \UU_{\alpha}(\vec{\bm{\Lambda}}, \vec{\bm{\delta}})]$ defined in \eqref{eq:tau_boot_quantile_gen} for constructing confidence bounds for the overlap weighted average treatment effect.
Table \ref{table:small_variation_seperate} shows the values of  $[\LL_{\alpha}(\vec{\bm{\Lambda}}, \vec{\bm{\delta}}), \UU_{\alpha}(\vec{\bm{\Lambda}}, \vec{\bm{\delta}})]$ under small variations around $\vec{\bm{\Lambda}} = (2,3)$, at three different values of $\vec{\bm{\delta}}$. 
From Table \ref{table:small_variation_seperate}, the bounds are robust to small variations in $\vec{\bm{\Lambda}}$. 

\begin{table}[htbp]
\centering
\renewcommand{\arraystretch}{2}
\caption{Values of $[\LL_{\alpha}(\vec{\bm{\Lambda}}, \vec{\bm{\delta}}), \UU_{\alpha}(\vec{\bm{\Lambda}}, \vec{\bm{\delta}})]$ under small variations in $\vec{\bm{\Lambda}}$.}
\begin{tabular}{lccc}
\toprule[1.5pt]
                           & $\vec{\bm{\delta}}=(0.1, 0.05)$      & $\vec{\bm{\delta}}=(0.05, 0.1)$      & $\vec{\bm{\delta}}=(0.1, 0.2)$        \\ \hline
$\vec{\bm{\Lambda}}=(2,3)$               & {[}3.586, 129.571{]}  & {[}6.498, 123.691{]}  & {[}-5.073, 139.646{]} \\
$\vec{\bm{\Lambda}}=(2.001,3.001)$ & {[}3.578, 129.589{]} & {[}6.489, 123.711{]} & {[}-5.082, 139.665{]}    \\
$\vec{\bm{\Lambda}}=(2.01,3.01)$     & {[}3.511, 129.746{]} & {[}6.408, 123.888{]} & {[}-5.164, 139.834{]}\\ 
\bottomrule[1.5pt]
\end{tabular}
\label{table:small_variation_seperate}
\end{table}

We consider then the sensitivity analysis in Section \ref{sec:whole} with bounds on the unmeasured confounding strength in the whole population. Again, we generate $n=1000$ i.i.d.~samples from \eqref{eq:data_gen} and consider the bounds $[\LL_{\alpha}(\Lambda, \delta), \UU_{\alpha}(\Lambda, \delta)]$ defined as in \eqref{eq:tau_boot_quantile_gen_whole} for constructing confidence bounds for the overlap-weighted average treatment effect. 
Table \ref{table:small_variation_joint} shows the values of $[\LL_{\alpha}(\Lambda, \delta), \UU_{\alpha}(\Lambda, \delta)]$ under small variations around $\Lambda=2$, at three different values of $\delta$. 
From Table \ref{table:small_variation_joint}, the bounds are robust to small variations in $\Lambda$.

\begin{table}[htbp]
\centering
\renewcommand{\arraystretch}{2}
\caption{Different values of $[\LL_{\alpha}(\Lambda, \delta), \UU_{\alpha}(\Lambda, \delta)]$ for small variation in $\Lambda$.}
\begin{tabular}{lccc}
\toprule[1.5pt]
               & $\delta=0.05$              & $\delta=0.1$              & $\delta=0.2$                \\ \hline
$\Lambda=2$       & {[}10.237, 116.060{]} & {[}2.402, 130.814{]} & {[}-19.846, 155.882{]} \\
$\Lambda=2.001$ & {[}10.228, 116.082{]} & {[}2.396, 130.833{]} & {[}-19.867, 155.898{]} \\
$\Lambda=2.01$  & {[}10.153, 116.272{]} & {[}2.338, 131.008{]} & {[}-20.053, 156.043{]}\\
\bottomrule[1.5pt]
\end{tabular}
\label{table:small_variation_joint}
\end{table}

\subsection{Conservativeness in the bounds when relaxing integer-type constraints}

We consider first sensitivity analysis with bounds on unmeasured confounding strength in both treated and control populations. 
We generate data from \eqref{eq:data_gen} with $n=1000$, and compare the bounds from \eqref{eq:MILP_sep} with and without relaxing the integer constraints, which are equivalent to 
$[\hat{\tau}_{\bstrap}^{\max}(\vec{\bm{\Lambda}}, \vec{\bm{\delta}}), \hat{\tau}_{\bstrap}^{\min}(\vec{\bm{\Lambda}}, \vec{\bm{\delta}})]$ defined as in \eqref{eq:tau_boot_epsilon_pred_set} with 
$(k_{\bstrap 1}, k_{\bstrap 2}, \ldots, k_{\bstrap n}) = (1,1,\ldots, 1)$. 
Table 
\ref{table:milp_seperate_conservativeness} 
reports the optimized lower and upper bounds. 
It shows that relaxing the integer-type constraints lead to little increase in conservativeness. 
\begin{table}[htbp]
\centering
\renewcommand{\arraystretch}{2}
\caption{Values of $[\hat{\tau}_{\bstrap}^{\max}(\vec{\bm{\Lambda}}, \vec{\bm{\delta}}), \hat{\tau}_{\bstrap}^{\min}(\vec{\bm{\Lambda}}, \vec{\bm{\delta}})]$ with and without integer-type constraints.}
\begin{tabular}{lcc}
\toprule[1.5pt]
\textit{Method}                    & $\Lambda=(2,3),\ \delta=(0.1, 0.05)$ & $\Lambda=(5,2),\  \delta=(0.05, 0.1)$ \\ \hline
Glover's Linearization    & {[}15.006, 110.655{]}          & {[}7.074, 127.107{]}         \\
Relax integer constraints & {[}15.006, 110.655{]}          & {[}7.074, 127.107{]}         \\ 
\bottomrule[1.5pt]
\end{tabular}
\label{table:milp_seperate_conservativeness}
\end{table}

We consider then sensitivity analysis with bounds on unmeasured confounding strength in the whole population. 
We generate data from \eqref{eq:data_gen} with $n=50$, and compute bounds in \eqref{eq:boot_frac_double_whole} using the observed samples, i.e., $(k_{\bstrap 1}, k_{\bstrap 2}, \ldots, k_{\bstrap n}) = (1,1,\ldots, 1)$. 
We consider both the multiple linear programming approach in Section \ref{sec:multiple_linear_appdix} and the quadratic programming approach in \ref{sec:quad_prog}, with and without integer-type constraints.
Table \ref{table:conservative_whole} presents the bounds obtained by solving the optimization problem in \eqref{eq:boot_frac_double_whole}, both with and without integer constraints, as well as those from the ad hoc approach described in Section \ref{sec:quad_prog}. 
Note that the bounds from the relaxed programming and ad hoc approach provide, respectively, superset and subset of the exact bounds. 
From Table \ref{table:conservative_whole}, relaxing integer-type constraints lead to little conservativeness, and the ad hoc approach provides quite tight bounds for the exact solution.

\begin{table}[htbp]
    \centering
    \renewcommand{\arraystretch}{2}
    \caption{Sensitivity analysis bounds from the optimization in \eqref{eq:boot_frac_double_whole}, using both the multiple linear programming and quadratic programming formulations.
    ``exact'' denotes the optimization with integer-type constraints, 
    ``relaxed'' denotes the optimization without integer-type constraints, 
    and 
    ``ad hoc'' denotes the ad hoc approach discussed in Section \ref{sec:quad_prog}. 
    }
    \resizebox{0.96\columnwidth}{!}{%
    \begin{tabular}{cccccc}
    \toprule[1.5pt]
    \textit{Method}                                            & \multicolumn{1}{l}{} & \multicolumn{4}{c}{\textit{Point estimation of the bounds}}                                                         \\ \hline
    \textbf{}                                                  &                      & $\Lambda = \exp(1), \ \delta=0.1$ & $\Lambda = \exp(1), \ \delta=0.15$ & $\Lambda = \exp(2), \ \delta=0.1$ & $\Lambda = \exp(2), \ \delta=0.15$ \\  \cline{2-6} 
    \multirow{3}{*}{MMILP} & 
    {\text{ad hoc}}
    & {[}12.266, 105.329{]}   & {[}9.336, 110.047{]}    & {[}7.096, 140.769{]}   & {[}6.075, 144.247{]}     \\
    & 
    {\text{exact}}
    & {[}12.266, 105.329{]}   & {[}9.336, 110.047{]}    & {[}7.096, 140.769{]}   & {[}6.075, 144.247{]}     \\
    & 
    {\text{relaxed}}
    & {[}12.266, 105.332{]}   & {[}9.336, 110.047{]}    & {[}7.096, 140.769{]}   & {[}6.075, 144.247{]}     \\ \cline{2-6} 
    \multirow{3}{*}{QCLP}                     & 
    {\text{ad hoc}} & {[}12.266, 105.329{]}   & {[}9.336, 110.047{]}    & {[}7.096, 140.769{]}   & {[}6.075, 144.247{]}       \\
    & 
    {\text{exact}} & {[}12.266, 105.329{]}   & {[}9.336, 110.047{]}    & {[}7.096, 140.769{]}   & {[}6.075, 144.247{]}       \\
    & 
    {\text{relaxed}}
    & {[}12.266, 105.332{]}   & {[}9.336, 110.047{]}    & {[}7.096, 140.769{]}   & {[}6.075, 144.247{]}       \\
    \bottomrule[1.5pt]
    \end{tabular}%
    }
    \label{table:conservative_whole}
\end{table}

\section{Additional real data analyses}

We now present additional analyses for the application in Section \ref{sec:illu} using the robust sensitivity analysis in Section \ref{sec:whole}  that imposes bounds on the strength of unmeasured confounding across the whole population. 
Figure \ref{fig:ci_whole}(a) shows the bootstrap confidence bounds in Corollary \ref{cor:ci_tau_whole} under various values of $(\Lambda, \delta)$, and Figure \ref{fig:ci_whole}(b) shows the narrower bootstrap confidence bounds from Theorem \ref{thm:ci_tau_pred_whole_sharper}. 
Figure \ref{fig:strength_unmeaured_confounding_whole} shows the simultaneous prediction intervals for the sample quantiles of unmeasured confounding strength among all units from Theorem \ref{thm:U_quant_pred_int_whole} with $\xi = 1$.
The interpretation of these results are similar to that in Section \ref{sec:illu} and is thus omitted for conciseness.

\begin{figure}[htbp]
    \centering
    \begin{subfigure}{0.48\textwidth}
        \centering
        \includegraphics[width=\textwidth]{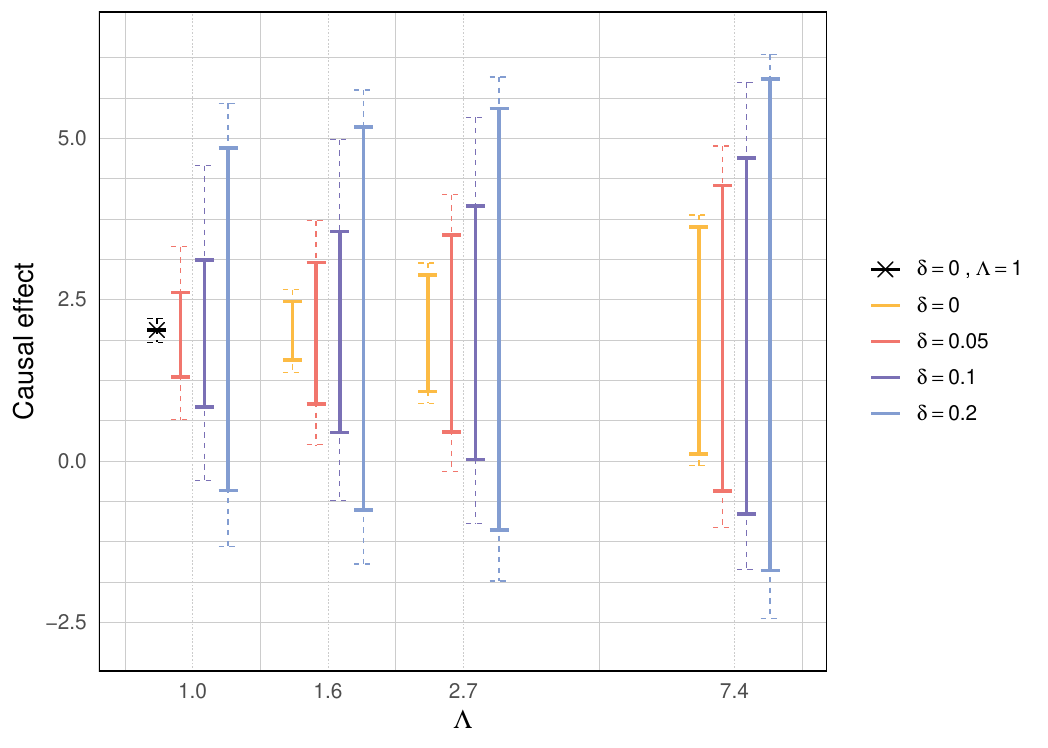}
        \caption{}
    \end{subfigure}
    \begin{subfigure}{0.48\textwidth}
        \centering
        \includegraphics[width=\textwidth]{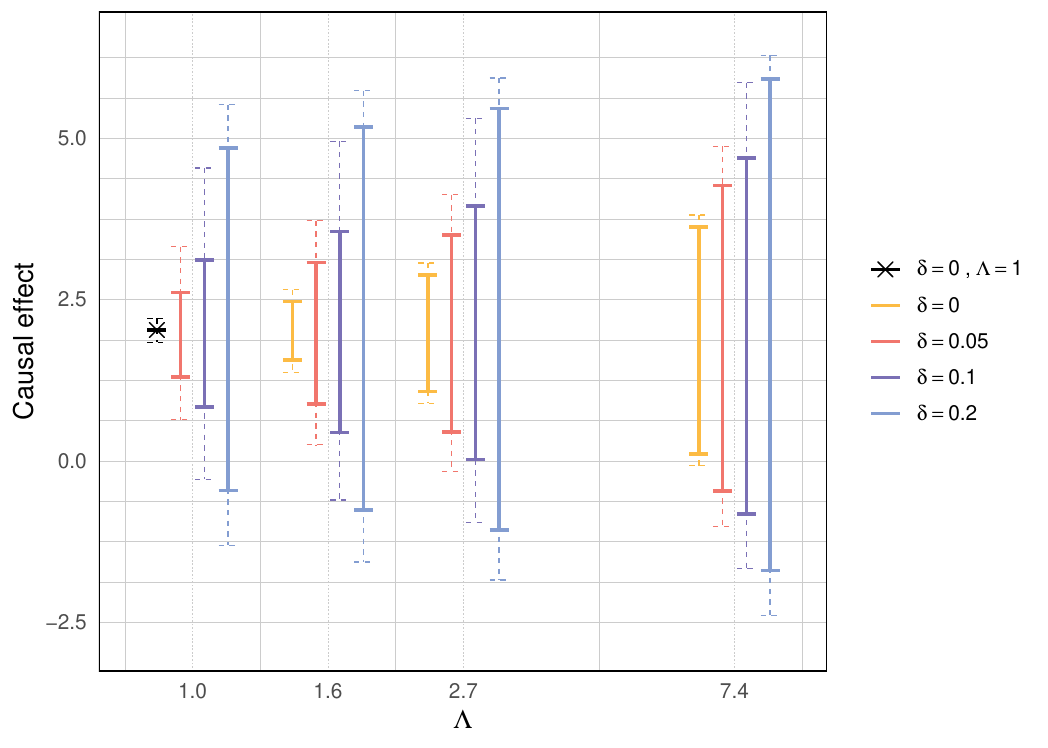}
        \caption{}
    \end{subfigure}
    \caption{
    Estimated bounds and 95\% confidence intervals for the overlap-weighted average causal effect $\tau_{\ow}$ of fish consumption on the blood mercury level.
    The solid intervals show the estimated bounds of $\tau_{\ow}$, and the dashed intervals show the 95\% confidence intervals for $\tau_{\ow}$, 
    under various sensitivity models in Assumption \ref{asmp:robust_sen_spec_whole} indexed by $(\Lambda, \delta)$. 
    The bounds in (a) and (b) are from Corollary \ref{cor:ci_tau_whole} and Theorem \ref{thm:ci_tau_pred_whole_sharper}, respectively. 
    }
    \label{fig:ci_whole}
\end{figure}

\begin{figure}[htbp]
    \centering\includegraphics[width=0.7\textwidth]{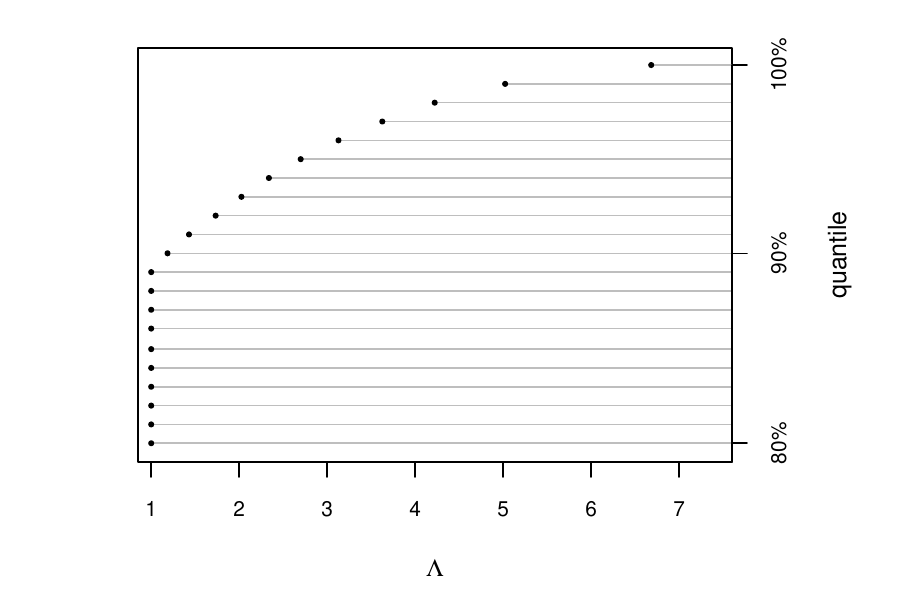}
    \caption{Lower prediction bounds for ${\OR}_{\m [q]}$s across various quantiles of unmeasured confounding strength, under the assumption of a zero overlap-weighted average treatment effect of the fish consumption on the blood mercury level.}
    \label{fig:strength_unmeaured_confounding_whole}
\end{figure}

\section{Proof of theorems}
\subsection{Proof of Proposition \ref{prop:tau_h_identifiable}}

In Proposition \ref{prop:tau_h_identifiable}, we have assumed that the potential outcomes $Y(1)$ and $Y(0)$, as well as the weighting function $h(\bmX,\bmU)$ under consideration, all have finite second moments. 
This ensures that the expectations in the following proof are well-defined. 

\begin{proof}[\bf Proof of Proposition \ref{prop:tau_h_identifiable}]

First, we prove the sufficiency. Suppose that $\E[h(\bmX,\bmU)\I\{e_{\tp}(\bmX,\bmU)=z\}]=0$ for $z=0,1$. 
Because $h(\cdot)$ is nonnegative, we must have $h(\bmX,\bmU)\I\{e_{\tp}(\bmX,\bmU)=z\}=0$ almost surely for $z=0,1$.
Define 
\begin{align*}
\xi_0(\bm{x},\bm{u})
= 
\begin{cases}
    \dfrac{1}{e_{\tp}(\bm{x},\bm{u})},  & \text{if } e_{\tp}(\bm{x},\bm{u})>0,\\
    0, & \text{if } e_{\tp}(\bm{x},\bm{u})=0,
\end{cases}
\ \ \text{ and } \ \ 
\xi_1(\bm{x},\bm{u})
=
\begin{cases}
    \dfrac{1}{1-e_{\tp}(\bm{x},\bm{u})},  & \text{if } e_{\tp}(\bm{x},\bm{u})<1,\\
    0, & \text{if } e_{\tp}(\bm{x},\bm{u})=1.
\end{cases}	
\end{align*}
We then have $\xi_0(\bmX,\bmU)e_{\tp}(\bmX,\bmU) = \I\{e_{\tp}(\bmX,\bmU)>0\}$, 
and 
\begin{align}\label{eq:h_xi0_YZ}
	\E\{h(\bmX,\bmU)\xi_0(\bmX,\bmU)YZ\} 
	=&\E[h(\bmX,\bmU)\xi_0(\bmX, \bmU)\E\{Y(1)Z|\bmX,\bmU\}] \nonumber \\
	=&\E[h(\bmX,\bmU)\xi_0(\bmX,\bmU)\E\{Y(1)|\bmX,\bmU\}\E(Z|\bmX,\bmU)] \nonumber \\
	=&\E[h(\bmX,\bmU)\xi_0(\bmX,\bmU)e_{\tp}(\bmX,\bmU)\E\{Y(1)|\bmX,\bmU\}] \nonumber \\
	=&\E[h(\bmX,\bmU) \I\{e_{\tp}(\bmX,\bmU)>0\} \E\{Y(1)|\bmX,\bmU\}] \nonumber \\
	=&\E\{h(\bmX,\bmU) \I\{e_{\tp}(\bmX,\bmU)>0\} Y(1)\},
\end{align}
where the first and last equalities follows from the law of total expectation, the second equality follows from the unconfoundedness assumption, 
the third equality follows from the definition of $e_\tp(\cdot)$, 
and the fourth equality follows from the definition of $\xi_0(\cdot)$. 
Because $h(\bmX,\bmU)\I\{e_{\tp}(\bmX,\bmU)=0\}=0$ almost surely, we then have 
\begin{align*}
    \E\{h(\bmX,\bmU)\xi_0(\bmX,\bmU)YZ\}
    & = \E\{h(\bmX,\bmU) \I\{e_{\tp}(\bmX,\bmU)>0\} Y(1)\}
    \\
    & = \E\{h(\bmX,\bmU) \I\{e_{\tp}(\bmX,\bmU)>0\} Y(1)\} + \E\{h(\bmX,\bmU) \I\{e_{\tp}(\bmX,\bmU)=0\} Y(1)\}
    \\
    & = \E\{h(\bmX,\bmU) Y(1)\}. 
\end{align*}
By the same logic, we have 
\begin{align*}
    & \E\{h(\bmX,\bmU)\xi_1(\bmX,\bmU)Y(1-Z)\}=\E\{h(\bmX,\bmU)Y(0)\}, &
    \E\{h(\bmX,\bmU)\xi_0(\bmX,\bmU)Z\}= \E\{h(\bmX,\bmU)\},
    \\
    & \E\{h(\bmX,\bmU)\xi_1(\bmX,\bmU)(1-Z)\}=\E\{h(\bmX,\bmU)\}. 
\end{align*}
Therefore, we can write $\tau_h$ equivalently as
\begin{align*}
	\tau_h=
    \dfrac{\E\{h(\bmX,\bmU)\{Y(1) - Y(0)\}\}}{\E\{h(\bmX,\bmU)\}} = 
    \dfrac{\E\{h(\bmX,\bmU)\xi_0(\bmX,\bmU)YZ\}}{\E\{h(\bmX,\bmU)\xi_0(\bmX,\bmU)Z\}}-\dfrac{\E\{h(\bmX,\bmU)\xi_1(\bmX,\bmU)Y(1-Z)\}}{\E\{h(\bmX,\bmU)\xi_1(\bmX,\bmU)(1-Z)\}}. 
\end{align*}
This indicates that $\tau_h$ is identifiable from the observable data.

Second, we prove the necessity. Suppose that $\E[h(\bmX,\bmU)\I\{e_{\tp}(\bmX,\bmU)=z\}] > 0$ for $z=0$ or $1$. 
Without loss of generality, we assume that $\E[h(\bmX,\bmU)\I\{e_{\tp}(\bmX,\bmU)=0\}]>0$. 
We then have $\Pr\{(\bmX,\bmU)\in \mathcal{A}_0\}>0$, where $\mathcal{A}_0=\{(\bm{x},\bm{u}): h(\bm{x},\bm{u})\I\{e_{\tp}(\bm{x},\bm{u})=0\}>0\}$. 
Note that 
\begin{align}\label{eq:hy_decom}
	\E\{h(\bmX,\bmU)Y(1)\}=\E[h(\bmX,\bmU)\I\{e_{\tp}(\bmX,\bmU)=0\}Y(1)]+\E[h(\bmX,\bmU)\I\{e_{\tp}(\bmX,\bmU)>0\}Y(1)],
\end{align}
where the second term is identifiable from the observable data as implied by \eqref{eq:h_xi0_YZ}. 
Below we show that the first term in \eqref{eq:hy_decom} is not identifiable.
Specifically, for any unit whose $(\bmX,\bmU)$ is in $\mathcal{A}_0$, we have probability zero to observe its treatment potential outcome $Y(1)$. 
That is, the distribution of $Y(1)$ among units whose covariate values fall in $\mathcal{A}_0$ is not identifiable from the observable data. 
Moreover, the first term in \eqref{eq:hy_decom} depends on such a distribution. 
For example, if $Y(1)$ equals a constant $c$ for units whose covariate values fall in $\mathcal{A}_0$, then we have 
$
    \E[h(\bmX,\bmU)\I\{e_{\tp}(\bmX,\bmU)=0\}Y(1)]
    = c \cdot \E[h(\bmX,\bmU)\I\{e_{\tp}(\bmX,\bmU)=0\}],
$
where $\E[h(\bmX,\bmU)\I\{e_{\tp}(\bmX,\bmU)=0\}] > 0$. 
Therefore, the quantity in \eqref{eq:hy_decom} and consequently $\tau_h$ are not identifiable from the observable data.

Finally, we show the equivalent conditions on $h(\cdot)$; that is, $\E[h(\bmX,\bmU)\I\{e_{\tp}(\bmX,\bmU)=z\}]=0$ for $z=0,1$ if and only if $h(\bmX, \bmU)=\tilde{h}(\bmX,\bmU)e_{\tp}(\bmX,\bmU)\{1-e_{\tp}(\bmX,\bmU)\}$ almost surely for some function $\tilde{h}(\cdot)$. 
When $h(\bmX,\bmU)\overset{\text{a.s.}}{=}\tilde{h}(\bmX,\bmU)e_{\tp}(\bmX,\bmU)\{1-e_{\tp}(\bmX,\bmU)\}$ for some $\tilde{h}(\cdot)$, it is easy to see that $\E[h(\bmX,\bmU)\I\{e_{\tp}(\bmX,\bmU)=z\}]=0$ for $z=0,1$. 
Below we consider the case where $\E[h(\bmX,\bmU)\I\{e_{\tp}(\bmX,\bmU)=z\}]=0$ for $z=0,1$.
In this case, we have $h(\bmX,\bmU)\I\{e_{\tp}(\bmX,\bmU)=z\} \overset{\text{a.s.}}{=} 0$ for $z=0,1$. 
Let 
\begin{align*}
    \tilde{h}(\bm{x},\bm{u}) = 
    \begin{cases}
        \dfrac{h(\bm{x},\bm{u})}{e_{\tp}(\bm{x},\bm{u})\{1-e_{\tp}(\bm{x},\bm{u})\}}, & 0 < e_{\tp}(\bm{x},\bm{u}) < 1,\\
        0, & e_{\tp}(\bm{x},\bm{u}) = 0 \text{ or } 1.
    \end{cases}
\end{align*}
We then have 
\begin{align*}
    & \quad \ h(\bmX,\bmU)
    \\
    & = 
    h(\bmX,\bmU)\I\{e_{\tp}(\bmX, \bmU)=0\}+ h(\bmX,\bmU) \I\{e_{\tp}(\bmX, \bmU)=1\}
    + 
    h(\bmX, \bmU)\I\{0<e_{\tp}(\bmX, \bmU)<1\}\\
    & \overset{\text{a.s.}}{=}
    h(\bmX, \bmU)\I\{0<e_{\tp}(\bmX, \bmU)<1\}
    \\
    & =
    \tilde{h}(\bmX, \bmU)e_{\tp}(\bmX, \bmU)\{1-e_{\tp}(\bmX, \bmU)\}\I\{0<e_{\tp}(\bmX, \bmU)<1\}
    \\
    & =
    \tilde{h}(\bmX, \bmU)e_{\tp}(\bmX, \bmU)\{1-e_{\tp}(\bmX, \bmU)\} [\I\{e_{\tp}(\bmX, \bmU)=0\}+ \I\{e_{\tp}(\bmX, \bmU)=1\}]\\
    &\ \ \ \ + 
    \tilde{h}(\bmX, \bmU)e_{\tp}(\bmX, \bmU)\{1-e_{\tp}(\bmX, \bmU)\} \I\{0<e_{\tp}(\bmX, \bmU)<1\}
    \\
    & = \tilde{h}(\bmX, \bmU)e_{\tp}(\bmX, \bmU)\{1-e_{\tp}(\bmX, \bmU)\},
\end{align*}
where the second last equality follows from the fact that $ e_{\tp}(\bmX, \bmU)(1-e_{\tp}(\bmX, \bmU)) [\I\{e_{\tp}(\bmX, \bmU)=0\}+ \I\{e_{\tp}(\bmX, \bmU)=1\}] \equiv 0$.

From the above, Proposition \ref{prop:tau_h_identifiable} holds.
\end{proof}

\subsection{Proof of Theorem \ref{thm:bootstrap_over_identify}}
To prove Theorem \ref{thm:bootstrap_over_identify}, we need the following lemmas.

\begin{lemma}\label{lemma:quantile_compare_gen}
    Let $\Omega \equiv B^{-1}\sum_{b=1}^B\I\{\max_{1\le i\le n}\|\bm{R}_i^\tp\hat{{\bmrho}}_b\|>\varepsilon\}$,
    where $B\equiv n^n$ is the total number of possible bootstrap samples and we use $b\in \{1,2,\ldots, B\}$ to index all the bootstrap samples. 
    We then have, for any $\alpha \ge \Omega$,
    \begin{align*}
        Q_{1-\alpha+\Omega}(\tilde{\theta}_{H, \bstrap})\ge Q_{1-\alpha}(\check{\theta}_{H, \bstrap})
        \quad
        \text{and}
        \quad
        Q_{\alpha-\Omega}(\tilde{\theta}_{H, \bstrap})\le Q_{\alpha}(\check{\theta}_{H, \bstrap}),
    \end{align*}
    where $\check{\theta}_{H, \bstrap}$ and $\tilde{\theta}_{H, \bstrap}$ are defined in \eqref{eq:bootstrap_R_rho} and \eqref{eq:bootstrap_R_epsilon}. 
\end{lemma}

 \begin{proof}[\bf Proof of Lemma \ref{lemma:quantile_compare_gen}]
    Consider any $1\le b \le B$ such that $\max_{1\le i\le n}\|\bm{R}_i^\tp\hat{{\bmrho}}_b\|\le\varepsilon$. 
    Then, for any $({\bmtheta}, \bmW_1, \ldots, \bmW_n)$ in the feasible region of \eqref{eq:bootstrap_R_rho}, $({\bmtheta}, \bmW_1+\bm{R}_1^\tp\hat{{\bmrho}}_b, \ldots, \bmW_n+\bm{R}_n^\tp\hat{{\bmrho}}_b)$ must be in the feasible region of \eqref{eq:bootstrap_R_epsilon}. 
    Therefore, once $\max_{1\le i\le n}\|\bm{R}_i^\tp\hat{{\bmrho}}_b\|\le\varepsilon$,     
    we must have $\tilde{\theta}_{H, b}\ge \check{\theta}_{H, b}$.

    Assume, without loss of generality, that $\check{\theta}_{H, 1}\le\check{\theta}_{H, 2}\le\cdots \le\check{\theta}_{H, B}$. 
    Let $k=B-\lfloor \alpha B\rfloor$ and $\mathcal{S}=\{s:  k\le s \le B, \max_{1\le i\le n}\|\bm{R}_i^\tp\hat{{\bmrho}}_s\|\le\varepsilon\}$.
    Then $Q_{1-\alpha}(\check{\theta}_{H, \bstrap})=\check{\theta}_{H, k}$. In addition, by the definition of $\Omega$, we know that $|\mathcal{S}|\ge B-k+1-\Omega B=\lfloor \alpha B\rfloor+1-\Omega B = \lfloor (\alpha - \Omega) B\rfloor+1$.
    Note that  
    the $(1-\beta)$th quantile of a sequence of length $B$ is greater than or equal to the minimum of any subsequence of length at least $\lfloor \beta B\rfloor+1$. 
    We then have
    \begin{align*}
        Q_{1-\alpha+\Omega}(\tilde{\theta}_{H, \bstrap}) \ge& \min_{s\in \mathcal{S}}\tilde{\theta}_{H, s}\ge \min_{s\in \mathcal{S}}\check{\theta}_{H, s}
        \ge
        \check{\theta}_{H, k}=Q_{1-\alpha}(\check{\theta}_{H, \bstrap}),
    \end{align*}
    where the second inequality follows from the definitions in  \eqref{eq:bootstrap_R_rho} and \eqref{eq:bootstrap_R_epsilon} and the fact that $\max_{1\le i\le n}\|\bm{R}_i^\tp\hat{{\bmrho}}_s\|\le\varepsilon$ for $s\in \mathcal{S}$. 
    Analogously, we can show that $Q_{\alpha-\Omega}(\tilde{\theta}_{H, \bstrap})\le Q_{\alpha}(\check{\theta}_{H, \bstrap})$.
    Thus, we derive Lemma \ref{lemma:quantile_compare_gen}. 
\end{proof}

\begin{lemma}\label{lemma:order_rho_b}
    Under Assumption \ref{asmp:valid_bootstrap_z_est}, we have 
    \begin{align*}
        \|\hat{{\bmrho}}_\bstrap\|=O_{\Pr}(n^{-1/2}).
    \end{align*}
\end{lemma}

\begin{proof}[\bf Proof of Lemma \ref{lemma:order_rho_b}]
Let $\mathcal{D} = \{\bmO_i, \bmW_i^\tp\}_{i=1}^n$ represent the $n$ i.i.d.~samples.
From Assumption \ref{asmp:valid_bootstrap_z_est}, we can know that, 
for any fixed $M>0$,
\begin{align}
\label{eq:rho_b_con_prob}
    \Pr(\|\sqrt{n}(\hat{{\bmrho}}_{\mathsf{b}} - 
\hat{{\bmrho}})\|>M\mid\mathcal{D})\overset{\Pr}{\rightarrow}\Pr(\|\mathcal{N}\left(\bm{0}, \bm{\Sigma}_{{\bmrho}}\right)\|>M),
\end{align}
where $\bm{\Sigma}_{{\bmrho}}$ is the submatrix of $\bm{\Sigma}$ corresponding to the asymptotic variance of $\hat{{\bmrho}}_{\mathsf{b}}$, and $\mathcal{N}\left(\bm{0}, \bm{\Sigma}_{{\bmrho}}\right)$ denotes a Gaussian random vector with mean zero and covariance matrix $\bm{\Sigma}_{{\bmrho}}$.
Note that for any $\varepsilon>0$, there exists $M_{\varepsilon}>0$ such that $\Pr(\|\mathcal{N}\left(\bm{0}, \bm{\Sigma}_{{\bmrho}}\right)\|>M_{\varepsilon})\le \varepsilon/2$. 

From \eqref{eq:rho_b_con_prob} and by the dominated convergence theorem,  we have 
\begin{align}\label{eq:expect_rho_b_con_prob}
    \Pr(\|\sqrt{n}(\hat{{\bmrho}}_{\mathsf{b}} - 
\hat{{\bmrho}})\|>M)\rightarrow\Pr(\|\mathcal{N}\left(\bm{0}, \bm{\Sigma}_{{\bmrho}}\right)\|>M).
\end{align}
Then for any given $\varepsilon$ and $M_{\varepsilon}$, there exists $n_{\varepsilon}>0$ such that $\Pr(\|\sqrt{n}(\hat{{\bmrho}}_{\mathsf{b}} - 
\hat{{\bmrho}})\|>M_{\varepsilon}) \le \Pr(\|\mathcal{N}\left(\bm{0}, \bm{\Sigma}_{{\bmrho}}\right)\|>M_{\varepsilon})+\varepsilon/2 \le \varepsilon$ for all $n>n_{\varepsilon}$.
Thus, 
$\|\hat{{\bmrho}}_{\mathsf{b}} - 
\hat{{\bmrho}}\|=O_{\Pr}(n^{-1/2})$.
In addition, 
Assumption \ref{asmp:valid_bootstrap_z_est} also implies that 
$\|\hat{{\bmrho}}\|=O_{\Pr}(n^{-1/2})$.

From the above, 
$\|\hat{{\bmrho}}_\bstrap\|\le\|\hat{{\bmrho}}_{\mathsf{b}} - 
\hat{{\bmrho}}\|+\|\hat{{\bmrho}}\|=O_{\Pr}(n^{-1/2})$, i.e., Lemma \ref{lemma:order_rho_b} holds.
\end{proof}

\begin{proof}[\bf Proof of Theorem \ref{thm:bootstrap_over_identify}]

    First, we prove that $\max_{1\le i \le n}\|\bm{R}_i^\tp\hat{{\bmrho}}_\bstrap\|=o_{\Pr}(1)$. 
    By Lemma \ref{lemma:order_rho_b}, 
    we have 
    $\|\hat{{\bmrho}}_\bstrap\|=O_{\Pr}(n^{-1/2})$.
    Since $\E(\|\bm{R}^\tp\|^{2+\nu})<\infty$ for some $\nu>0$, for any $\tilde{\varepsilon}>0$, we then have 
    \begin{align*}
        \Pr(\max_{1\le i \le n}\|\bm{R}_i^\tp\|\ge\tilde{\varepsilon} n^{1/2}) & \le n \Pr(\|\bm{R}^\tp\|\ge\tilde{\varepsilon} n^{1/2}) \le n \E(\|\bm{R}^\tp\|^{2+\nu})/(\tilde{\varepsilon}^{2+\nu}n^{1+\nu/2})\\
        & =n^{-\nu/2}\E(\|\bm{R}^\tp\|^{2+\nu})/\tilde{\varepsilon}^{2+\nu}\rightarrow 0,
    \end{align*}
    where the first inequality follows from the union bound, and the second inequality follows from the Markov inequality. 
    This immediately implies that     
    $\max_{1\le i \le n}\|\bm{R}_i^\tp\|=o_{\Pr}(n^{1/2})$, and consequently $\max_{1\le i \le n}\|\bm{R}_i^\tp\hat{{\bmrho}}_\bstrap\|=o_{\Pr}(1)$.

    Second,  
    for any given $\varepsilon>0$, we define $\Omega$ as in Lemma \ref{lemma:quantile_compare_gen}: 
    \begin{align*}
        \Omega=B^{-1}\sum_{b=1}^B\I\left( \max_{1\le i\le n}\|\bm{R}_i^\tp\hat{{\bmrho}}_b\|>\varepsilon \right)
        = 
        \Pr\left( \max_{1\le i\le n}\|\bm{R}_i^\tp\hat{{\bmrho}}_\bstrap\|>\varepsilon \mid \mathcal{D}\right), 
    \end{align*}
    where $\mathcal{D} = \{\bmO_i, \bmW_i^\tp\}_{i=1}^n$ represent the $n$ i.i.d.~samples.
    Then by the Markov inequality, for any $\eta>0$, 
    \begin{align}\label{eq:Omega_o1}
        \Pr(\Omega>\eta)\le\E(\Omega)/\eta=\Pr\left(\max_{1\le i \le n}\|\bm{R}_i^\tp\hat{{\bmrho}}_\bstrap\|>\varepsilon \right)/\eta\rightarrow 0 
        \ \ \text{as} \ \ 
        n\rightarrow\infty.
    \end{align}
    
    Third, we prove that $Q_{1-\alpha} ( \tilde{\theta}_{H, \bstrap} )$ is an asymptotic $1-\alpha$ upper confidence bound for $H({\bmtheta}_0)$. 
    From Lemma \ref{lemma:quantile_compare_gen} and the discussion before, for any $0 < \eta \le \alpha' < 1$, 
    \begin{align*}
        & \quad \ \limsup_{n\rightarrow \infty} \Pr\{ H({\bmtheta}_0) > Q_{1-\alpha'+\eta} ( \tilde{\theta}_{H, \bstrap} ) \}\\
        &\le
        \limsup_{n\rightarrow \infty} \Pr\{ H({\bmtheta}_0) > Q_{1-\alpha'+\eta} ( \tilde{\theta}_{H, \bstrap} ), \Omega\le\eta \}+\limsup_{n\rightarrow \infty} \Pr\{ H({\bmtheta}_0) > Q_{1-\alpha'+\eta} ( \tilde{\theta}_{H, \bstrap} ), \Omega>\eta \}\\
        &\le
        \limsup_{n\rightarrow \infty} \Pr\{ H({\bmtheta}_0) > Q_{1-\alpha'+\Omega} ( \tilde{\theta}_{H, \bstrap} ) , 
        \Omega \le \alpha'
        \}+\limsup_{n\rightarrow \infty} \Pr(\Omega>\eta)\\
        &\le
        \limsup_{n\rightarrow \infty} \Pr\{ H({\bmtheta}_0) > Q_{1-\alpha'}(\check{\theta}_{H, \bstrap}) \}\\
        &\le
        \limsup_{n\rightarrow \infty} \Pr[ H({\bmtheta}_0) > Q_{1-\alpha'}\{H(\hat{{\bmtheta}}_\bstrap)\} ]\\
        &=\alpha',
    \end{align*}
    where the third inequality uses Lemma \ref{lemma:quantile_compare_gen} and \eqref{eq:Omega_o1}, 
    the fourth inequality follows from the definition in \eqref{eq:bootstrap_R_rho}, and the last equality follows from the property of the standard bootstrap. 
    Let $\alpha'=\alpha+\eta$. We then have 
    \begin{align*}
        \limsup_{n\rightarrow \infty} \Pr\{ H({\bmtheta}_0) > Q_{1-\alpha} ( \tilde{\theta}_{H, \bstrap} ) \}\le \alpha+\eta.
    \end{align*}
    Letting $\eta\rightarrow 0$, this then implies that 
    $
        \limsup_{n\rightarrow \infty} \Pr\{ H({\bmtheta}_0) > Q_{1-\alpha} ( \tilde{\theta}_{H, \bstrap} ) \}\le \alpha.
    $ Thus, $Q_{1-\alpha} ( \tilde{\theta}_{H, \bstrap} )$ is an asymptotic $1-\alpha$ upper confidence bound for $H({\bmtheta}_0)$. 

    From the above, Theorem \ref{thm:bootstrap_over_identify} holds. 
\end{proof}

\begin{proof}[\bf Comment on a sample-size-dependent $\varepsilon$ for Theorem \ref{thm:bootstrap_over_identify}]
    We can replace $\varepsilon$ in \eqref{eq:bootstrap_R_epsilon} by $\varepsilon_n$ that can decrease with the sample size $n$, with rate depending on the order of $\max_{1\le i\le n} \|\bm{R}_i^\tp\|$, such that Theorem \ref{thm:bootstrap_over_identify} still holds.

    Specifically, suppose that $\max_{1\le i\le n} \|\bm{R}_i^\tp\|=o_{\Pr}(n^\gamma)$ for some $\gamma\le 1/2$.
    Then $\max_{1\le i \le n}\|\bm{R}_i^\tp\hat{{\bmrho}}_\bstrap\|=o_{\Pr}(n^{\gamma-1/2})$.
    Let $\varepsilon_n=n^{\gamma-1/2}\varepsilon_0$ for a positive constant $\varepsilon_0$ and $\Omega_n=B^{-1}\sum_{b=1}^B\I\{\max_{1\le i\le n}\|\bm{R}_i^\tp\hat{{\bmrho}}_b\|>\varepsilon_n\}$,
    we then have, for any $\eta>0$, 
    \begin{align*}
        \Pr(\Omega_n>\eta)\le\E(\Omega_n)/\eta=\Pr\left(\max_{1\le i \le n}\|\bm{R}_i^\tp\hat{{\bmrho}}_\bstrap\|>n^{\gamma-1/2}\varepsilon_0 \right)/\eta\rightarrow 0 
        \ \ \text{as} \ \ 
        n\rightarrow\infty.
    \end{align*}
    We can then derive Theorem \ref{thm:bootstrap_over_identify} by the same logic as its proof before. 
    
    Analogously, when $\bm{R}^\tp$ is bounded, we can choose ${\varepsilon_n}$ such that $\varepsilon_n \gg n^{-1/2}$, under which Theorem \ref{thm:bootstrap_over_identify} still holds. 
\end{proof}

\subsection{Proof of Theorem \ref{thm:bootstrap_over_identify_pred_set}}
\begin{proof}[\bf Proof of Theorem \ref{thm:bootstrap_over_identify_pred_set}]
    Similar to \eqref{eq:bootstrap_R_rho}, we consider the following optimization  based on bootstrap samples:
    \begin{align}\label{eq:bootstrap_R_rho_pred_set}
    \breve{\theta}_{H, \bstrap} = \ & \sup H({\bmtheta}) 
    \\
    \text{subject to} \quad & 
    \sum_{i=1}^n k_{\bstrap i} p_{\bmtheta}(\bmO_{i}, \bmW_i + \bm{R}^\tp_i \hat{{\bmrho}}_{\bstrap} ) =\bm{0},
    \nonumber
    \\
    & (\bmW_1, \bmW_2, \ldots, \bmW_n) \in \mathcal{W}_n.
    \nonumber
\end{align}
    For any given $\varepsilon> 0$, define  $\Omega=B^{-1}\sum_{b=1}^B\I\{\max_{1\le i\le n}\|\bm{R}_i^\tp\hat{{\bmrho}}_b\|>\varepsilon\}$ as in Lemma \ref{lemma:quantile_compare_gen}.

    By the same logic as the second step in the proof of Theorem \ref{thm:bootstrap_over_identify}, we have 
    $\Pr(\Omega>\eta)\rightarrow 0$ as $n\rightarrow \infty$ for any $\eta>0$. 
    Additionally, by the same logic as the proof of Lemma \ref{lemma:quantile_compare_gen}, we have $Q_{1-\alpha+\Omega}(\hat{\theta}_{H, \bstrap})\ge Q_{1-\alpha}(\breve{\theta}_{H, \bstrap})$ for any $\alpha \ge \Omega$.
    Moreover, note that if $(\bmW_1^\tp, \bmW_2^\tp, \ldots, \bmW_n^\tp)\in \mathcal{W}_n$, then $\breve{\theta}_{H, \bstrap}\ge H(\hat{{\bmtheta}}_\bstrap)$. This is because 
    $(\hat{{\bmtheta}}_{\bstrap}, \bmW_1^\tp, \ldots, \bmW_n^\tp)$ is in the feasible region of \eqref{eq:bootstrap_R_rho_pred_set}.
    
    By the same logic as the third step in the proof of Theorem \ref{thm:bootstrap_over_identify}, for any
    $0 < \eta \le \alpha' < 1$, 
    \begin{align*}
        &\quad \ \limsup_{n\rightarrow \infty} \Pr\{ H({\bmtheta}_0) > Q_{1-\alpha'+\eta} ( \hat{\theta}_{H, \bstrap} ), \  (\bmW_1^\tp, \bmW_2^\tp, \ldots, \bmW_n^\tp) \in \mathcal{W}_n \}\\
        & \le
        \limsup_{n\rightarrow \infty} \Pr\{ H({\bmtheta}_0) > Q_{1-\alpha'+\eta} ( \hat{\theta}_{H, \bstrap} ), \  (\bmW_1^\tp, \bmW_2^\tp, \ldots, \bmW_n^\tp) \in \mathcal{W}_n, \ \Omega\le\eta \}\\
        &\quad \ +\limsup_{n\rightarrow \infty} \Pr\{ H({\bmtheta}_0) > Q_{1-\alpha'+\eta} ( \hat{\theta}_{H, \bstrap} ), \  (\bmW_1^\tp, \bmW_2^\tp, \ldots, \bmW_n^\tp) \in \mathcal{W}_n , \ \Omega>\eta\}\\
        & \le
        \limsup_{n\rightarrow \infty} \Pr\{ H({\bmtheta}_0) > Q_{1-\alpha'+\Omega} ( \hat{\theta}_{H, \bstrap}),
        \alpha'\ge \Omega, 
        \  (\bmW_1^\tp, \bmW_2^\tp, \ldots, \bmW_n^\tp) \in \mathcal{W}_n \}+\limsup_{n\rightarrow \infty} \Pr(\Omega>\eta)\\
        & \le
        \limsup_{n\rightarrow \infty} \Pr\{ H({\bmtheta}_0) > Q_{1-\alpha'}(\breve{\theta}_{H, \bstrap}), \  (\bmW_1^\tp, \bmW_2^\tp, \ldots, \bmW_n^\tp)\in \mathcal{W}_n \}\\
        & \le
        \limsup_{n\rightarrow \infty} \Pr[ H({\bmtheta}_0) > Q_{1-\alpha'}\{H(\hat{{\bmtheta}}_\bstrap)\} ]\\
        & =\alpha'.
    \end{align*}
    Letting $\alpha'=\alpha+\eta$, then 
    \begin{align*}
        \limsup_{n\rightarrow \infty} \Pr\{ H({\bmtheta}_0) > Q_{1-\alpha} ( \hat{\theta}_{H, \bstrap} ), \  (\bmW_1^\tp, \bmW_2^\tp, \ldots, \bmW_n^\tp) \in \mathcal{W}_n \} & \le\alpha+\eta.
    \end{align*}
    Letting $\eta\rightarrow 0$, this then implies that 
    \begin{align*}
        \limsup_{n\rightarrow \infty} \Pr\{ H({\bmtheta}_0) > Q_{1-\alpha} ( \hat{\theta}_{H, \bstrap} ), \  (\bmW_1^\tp, \bmW_2^\tp, \ldots, \bmW_n^\tp) \in \mathcal{W}_n \}& \le\alpha.
    \end{align*}

    If further $\mathcal{W}_n$ is an asymptotic $1-\zeta$ prediction set for $(\bmW_1^\tp, \bmW_2^\tp, \ldots, \bmW_n^\tp)$, then 
    \begin{align*}
     \limsup_{n\rightarrow \infty} \Pr\{ H({\bmtheta}_0) > Q_{1-\alpha} ( \hat{\theta}_{H, \bstrap} ) \}
    & \le 
         \limsup_{n\rightarrow \infty} \Pr\{ H({\bmtheta}_0) > Q_{1-\alpha} ( \hat{\theta}_{H, \bstrap} ), \  (\bmW_1^\tp, \bmW_2^\tp, \ldots, \bmW_n^\tp) \in \mathcal{W}_n \}\\
         & \quad \ + 
         \limsup_{n\rightarrow \infty} \Pr\{ (\bmW_1^\tp, \bmW_2^\tp, \ldots, \bmW_n^\tp) \notin \mathcal{W}_n \}\\
         & \le \alpha + \limsup_{n\rightarrow \infty} \Pr\{ (\bmW_1^\tp, \bmW_2^\tp, \ldots, \bmW_n^\tp) \notin \mathcal{W}_n \}
         \\
         & \le \alpha + \zeta. 
    \end{align*}

    From the above, Theorem \ref{thm:bootstrap_over_identify_pred_set} holds. 
\end{proof}

\subsection{Proof of Theorem \ref{thm:valid_bootstrap_R_g}}
We consider a logistic regression model for the propensity model, 
where $e(x;\bm{\beta}) = \logit^{-1} \{\bm{\beta}^\top s(\bmX)\}$. 
For any parameter vector $({\bmtheta},{\bmrho})=(\tau, \bm{\beta}, \kappa_{\treat}, \kappa_{\control}, {\bmrho})$, 
consider the following estimating equations as in \eqref{eq:est_equ_ow}:
\begin{align*}
    p_{{\bmtheta}}(\bmO, W^\tp+\bm{R}^\tp{\bmrho})
    &=
    \begin{pmatrix}
        p_{{\bmtheta}, 1}(\bmO, W^\tp+\bm{R}^\tp{\bmrho})\\
        p_{{\bmtheta}, 2}(\bmO, W^\tp+\bm{R}^\tp{\bmrho})\\
        p_{{\bmtheta}, 3}(\bmO, W^\tp+\bm{R}^\tp{\bmrho})\\
        p_{{\bmtheta}, 4}(\bmO, W^\tp+\bm{R}^\tp{\bmrho})\\
        p_{{\bmtheta}, 5}(\bmO, W^\tp+\bm{R}^\tp{\bmrho})
    \end{pmatrix}\\
    &=
    \begin{bmatrix}
        \left(Z-\dfrac{e^{\bm{\beta}^\top s(\bmX)}}{1+e^{\bm{\beta}^\top s(\bmX)}}\right)s(\bmX)
        \\
        \dfrac{Z}{1+e^{W^\tp+\bm{R}^\tp{\bmrho}+\bm{\beta}^\top s(\bmX)}} - \kappa_{\treat}\\
        \dfrac{1-Z}{1+e^{-W^\tp-\bm{R}^\tp{\bmrho}-\bm{\beta}^\top s(\bmX)}} - \kappa_{\control}\\
        \dfrac{ZY}{\{1+e^{W^\tp+\bm{R}^\tp{\bmrho}+\bm{\beta}^\top s(\bmX)}\}\kappa_{\treat}}- \dfrac{(1-Z)Y}{\{1+e^{-W^\tp-\bm{R}^\tp{\bmrho}-\bm{\beta}^\top s(\bmX)}\}\kappa_{\control}}-\tau\\
        \dfrac{Zg(\bmX)}{\{1+e^{W^\tp+\bm{R}^\tp{\bmrho}+\bm{\beta}^\top s(\bmX)}\}\kappa_{\treat}}- \dfrac{(1-Z)g(\bmX)}{\{1+e^{-W^\tp-\bm{R}^\tp{\bmrho}-\bm{\beta}^\top s(\bmX)}\}\kappa_{\control}}
    \end{bmatrix}.
\end{align*}

Define $\Phi({\bmtheta}, {\bmrho})=\E\{p_{{\bmtheta}}(\bmO, W^\tp+\bm{R}^\tp{\bmrho})\}$. Note that $\Phi({\bmtheta}_0, {\bmrho}_0)=\bm{0}$, where $({\bmtheta}_0, {\bmrho}_0)=(\tau_{\ow}, \bm{\beta}_0, \kappa_{\treat 0}, \kappa_{\control 0}, {\bmrho}_0)$. The Z-estimates $(\hat{{\bmtheta}}, \hat{{\bmrho}})=(\hat{\tau}, \hat{\bm{\beta}}, \hat{\kappa}_{\treat}, \hat{\kappa}_{\control}, \hat{{\bmrho}})$ are obtained by solving the equations:
\begin{align*}
    \Phi_n(\hat{{\bmtheta}}, \hat{{\bmrho}})\equiv\frac{1}{n}\sumn p_{\hat{{\bmtheta}}}(\bmO_i, W_i^\tp+\bm{R}_i^\tp\hat{{\bmrho}})=\bm{0}.
    \end{align*}

Let $\Pr_n$ be the empirical measure of the sample $\{(\bmO_i, W_i^\tp, \bm{R}_i^\tp)\}_{i=1}^n$, and $\{(\hat{\bmO}_{\bstrap i}, \hat{W}_{\bstrap i}^\tp, \hat{\bm{R}}_{\bstrap i}^\tp)\}_{i=1}^n$
be the corresponding bootstrap samples or equivalently i.i.d.~samples from the empirical distribution $\Pr_n$. The bootstrap empirical distribution is $\hat{\Pr}_n=n^{-1}\sumn \delta_{(\hat{\bmO}_{\bstrap i}, \hat{W}_{\bstrap i}^\tp, \hat{\bm{R}}_{\bstrap i}^\tp)}$, where $\delta_{(\hat{\bmO}_{\bstrap i}, \hat{W}_{\bstrap i}^\tp, \hat{\bm{R}}_{\bstrap i}^\tp)}$ denotes a point mass at $(\hat{\bmO}_{\bstrap i}, \hat{W}_{\bstrap i}^\tp, \hat{\bm{R}}_{\bstrap i}^\tp)$. 
Then the bootstrap Z-estimator $(\hat{{\bmtheta}}_{\bstrap}, \hat{{\bmrho}}_{\bstrap})=(\hat{\tau}_{\bstrap}, \hat{\bm{\beta}}_{\bstrap}, \hat{\kappa}_{\treat\bstrap}, \hat{\kappa}_{\control\bstrap}, \hat{{\bmrho}}_{\bstrap})$ are obtained from the equations:
\begin{align*}
    \hat{\Phi}_n(\hat{{\bmtheta}}_{\bstrap}, \hat{{\bmrho}}_{\bstrap})\equiv\frac{1}{n}\sumn p_{\hat{{\bmtheta}}_{\bstrap}}(\hat{\bmO}_{\bstrap i}, \hat{W}_{\bstrap i}^\tp+\hat{\bm{R}}_{\bstrap i}^\tp\hat{{\bmrho}}_{\bstrap})=\bm{0}.
\end{align*}

To prove Theorem \ref{thm:valid_bootstrap_R_g}, we need the following seven lemmas. 

\begin{lemma}[\citet{KosorokEmp:2006}, Theorem 10.16]\label{lemma:asy_ci}
    Consider $Z$-estimation based on the estimating equation ${\bmtheta}\mapsto\Psi_n({\bmtheta})\equiv \mathbb{P}_n\psi_{\bmtheta}$, where ${\bmtheta}\in{\bm{\Theta}}\subset\mathbb{R}^p$ and $\bm{x}\mapsto\psi_{\bmtheta}(\bm{x})$ is a measurable $p$-vector valued function for each ${\bmtheta}$.
    Define the map ${\bmtheta}\mapsto\Psi({\bmtheta})\equiv \Pr\psi_{\bmtheta}$ and assume ${\bmtheta}_0\in{\bm{\Theta}}$ satisfies $\Psi({\bmtheta}_0)=\bm{0}$. Let $\hat{{\bmtheta}}_n$ be an approximate zero of $\Psi_n$, and let $\tilde{{\bmtheta}}_n$ be an approximate zero of the bootstrapped estimating equation ${\bmtheta}\mapsto\hat\Psi_n({\bmtheta})\equiv\hat{\mathbb{P}}_n\psi_{\bmtheta}$. Let ${\bm{\Theta}}$ be open and assume the following:
    \begin{enumerate}
        \item[(i)] For any sequence $\{{\bmtheta}_n\}\in\bm{\Theta}, \Psi({\bmtheta}_n)\rightarrow \bm{0}$ implies $\|{\bmtheta}_n-{\bmtheta}_0\|\rightarrow 0$;
        \item[(ii)] The class $\{\psi_{\bmtheta}:{\bmtheta}\in{\bm{\Theta}}\}$ is strong Glivenko-Cantelli;
        \item[(iii)] For any $\eta>0$, the class $\mathcal{F}\equiv\{\psi_{\bmtheta}:{\bmtheta}\in{\bm{\Theta}}, \|{\bmtheta}-{\bmtheta}_0\|\leq \eta\}$ is Donsker and $\mathbb{P}\|\psi_{\bmtheta}-\psi_{{\bmtheta}_0}\|^2\rightarrow 0$ as $\|{\bmtheta}-{\bmtheta}_0\|\rightarrow 0$;
        \item[(iv)] $\mathbb{P}\|\psi_{{\bmtheta}_0}\|^2<\infty$ and $\Psi({\bmtheta})$ is differentiable at ${\bmtheta}_0$ with nonsingular derivative matrix $\bm{V}_{{\bmtheta}_0};$
        \item[(v)] $\Psi_n(\hat{{\bmtheta}}_n)=o_{\Pr}(n^{-1/2})$ and $\hat{\Psi}_n(\tilde{{\bmtheta}}_n)=o_{\Pr}(n^{-1/2})$.
    \end{enumerate}
    Then 
    \begin{align*}
        n^{1/2}(\hat{{\bmtheta}}_n-{\bmtheta}_0)\converged \mathcal{N}(\bm{0}, \bm{V}_{{\bmtheta}_0}^{-1}\mathbb{P}[\psi_{{\bmtheta}_0}\psi_{{\bmtheta}_0}^{\top}](\bm{V}_{{\bmtheta}_0}^{-1})^{\top}), \ 
        \text{and}
        \ 
        n^{1/2}(\tilde{{\bmtheta}}_n-\hat{{\bmtheta}}_n)\overset{\Pr}{\converged}\mathcal{N}(\bm{0}, \bm{V}_{{\bmtheta}_0}^{-1}\mathbb{P}[\psi_{{\bmtheta}_0}\psi_{{\bmtheta}_0}^{\top}](\bm{V}_{{\bmtheta}_0}^{-1})^{\top}).
    \end{align*}
\end{lemma}

\begin{proof}[\bf Proof of Lemma \ref{lemma:asy_ci}]
    Lemma \ref{lemma:asy_ci} is from \citet[][Theorem 10.16]{KosorokEmp:2006}. 
\end{proof}

\begin{lemma}\label{lemma:unqiue_root}
    If 
    Condition \ref{cond:logistic_Z_est} holds and $\bm{R}^\tp=g^\top(\bmX)$ ,
    the following equation has a unique solution ${\bmrho}={\bmrho}_0$ in a sufficiently small open neighborhood around ${\bmrho}_0\equiv\bm{0}$:
    \begin{align*}
        \frac{\E[g(\bmX)Z/\{1+e^{W^\tp+\bm{R}^\tp{\bmrho}+\bm{\beta}_0^{\top}s(\bmX)}\}]}{\E[Z/\{1+e^{W^\tp+\bm{R}^\tp{\bmrho}+\bm{\beta}_0^{\top}s(\bmX)}\}]}
        =
        \frac{\E[g(\bmX)(1-Z)/\{1+e^{-W^\tp-\bm{R}^\tp{\bmrho}-\bm{\beta}_0^{\top}s(\bmX)}\}]}{\E[(1-Z)/\{1+e^{-W^\tp-\bm{R}^\tp{\bmrho}-\bm{\beta}_0^{\top}s(\bmX)}\}]}.
    \end{align*}    
\end{lemma}

\begin{proof}[\bf Proof of Lemma \ref{lemma:unqiue_root}]
    Let \begin{align*}
        f({\bmrho})\equiv
        &\frac{\E[g(\bmX)Z/\{1+e^{W^\tp+\bm{R}^\tp{\bmrho}+\bm{\beta}_0^{\top}s(\bmX)}\}]}{\E[Z/\{1+e^{W^\tp+\bm{R}^\tp{\bmrho}+\bm{\beta}_0^{\top}s(\bmX)}\}]}\\
        &-\frac{\E[g(\bmX)(1-Z)/\{1+e^{-W^\tp-\bm{R}^\tp{\bmrho}-\bm{\beta}_0^{\top}s(\bmX)}\}]}{\E[(1-Z)/\{1+e^{-W^\tp-\bm{R}^\tp{\bmrho}-\bm{\beta}_0^{\top}s(\bmX)}\}]}.
    \end{align*}
    By the inverse function theorem, it suffices to show that $\nabla_{{\bmrho}=\bm{0}}f({\bmrho})$ is nonsingular. 
    In fact, we can establish an even stronger result that 
    $\nabla_{{\bmrho}=\bm{0}}f({\bmrho})\prec 0$, i.e., $\nabla_{{\bmrho}=\bm{0}}f({\bmrho})$ is negative definite. 

    By the law of total expectation and using the fact that 
    $
        e^{W^\tp+\bm{\beta}_0^\top s(\bmX)}={e_\tp(\bmX,\bmU)}/\{1-e_\tp(\bmX,\bmU)\}, 
    $
    we have 
    \begin{align*}
        f({\bmrho})  
        & = 
        \frac{\E\left\{\dfrac{g e_\tp }{1+\dfrac{e_\tp}{1-e_\tp}e^{\bm{R}^\tp{\bmrho}}}\right\}}{\E\left\{\dfrac{e_\tp}{1+\dfrac{e_\tp}{1-e_\tp}e^{\bm{R}^\tp{\bmrho}}}\right\}}
        -
        \frac{\E\left\{ \dfrac{g (1-e_\tp)}{1+\dfrac{1-e_\tp}{e_\tp}e^{-\bm{R}^\tp{\bmrho}}}\right\} }{\E\left\{\dfrac{1-e_\tp}{1+\dfrac{1-e_\tp}{e_\tp}e^{-\bm{R}^\tp{\bmrho}}}\right\}},
    \end{align*}
    where we use $g$ and $e_\tp$ to denote $g(\bmX)$ and $e_\tp(\bmX,\bmU)$ for notational simplicity. 
    Note that for $\tilde{g}$ equals $g$ or $1$, we have 
    \begin{align*}
        \nabla_{{\bmrho}} \left\{\dfrac{\tilde{g} e_\tp }{1+\dfrac{e_\tp}{1-e_\tp}e^{\bm{R}^\tp{\bmrho}}}\right\}
        & = - \frac{\tilde{g}  \bm{R}^\tp e_\tp \frac{e_\tp}{1-e_\tp} e^{\bm{R}^\tp {\bmrho}}}{
        \left( 1+\dfrac{e_\tp}{1-e_\tp}e^{\bm{R}^\tp{\bmrho}} \right)^2
        }
        = 
        - \frac{\tilde{g}  g^\top e_\tp \frac{e_\tp}{1-e_\tp} e^{g^\top {\bmrho}}}{
        \left( 1+\dfrac{e_\tp}{1-e_\tp}e^{g^\top  {\bmrho}} \right)^2
        }, 
        \\ 
        \nabla_{{\bmrho}} \left\{ \dfrac{\tilde{g} (1-e_\tp)}{1+\dfrac{1-e_\tp}{e_\tp}e^{-\bm{R}^\tp{\bmrho}}}\right\}
        & =  \frac{\tilde{g} \bm{R}^\tp (1-e_\tp) \frac{1-e_\tp}{e_\tp} e^{-\bm{R}^\tp{\bmrho}}}{
        \left( 1+\dfrac{1-e_\tp}{e_\tp}e^{-\bm{R}^\tp{\bmrho}} \right)^2
        } 
        =  \frac{\tilde{g} g^\top (1-e_\tp) \frac{1-e_\tp}{e_\tp} e^{-g^\top{\bmrho}}}{
        \left( 1+\dfrac{1-e_\tp}{e_\tp}e^{-g^\top{\bmrho}} \right)^2
        }, 
    \end{align*}
    where we use the fact that $\bm{R}^\tp = g^\top$. These immediately imply that for $\tilde{g}$ equals $g$ or $1$,    
    \begin{align*}
        \nabla_{{\bmrho}=\bm{0}} \left\{\dfrac{\tilde{g} e_\tp }{1+\dfrac{e_\tp}{1-e_\tp}e^{\bm{R}^\tp{\bmrho}}}\right\}
        & 
        = - \tilde{g} g^\top e_\tp^2 (1-e_\tp), 
        \nabla_{{\bmrho}=\bm{0}} \left\{ \dfrac{\tilde{g} (1-e_\tp)}{1+\dfrac{1-e_\tp}{e_\tp}e^{-\bm{R}^\tp{\bmrho}}}\right\}
        = \tilde{g} g^\top e_\tp (1-e_\tp)^2, 
    \end{align*}
    Moreover, for $\tilde{g}$ equals $g$ or $1$,
    \begin{align*}
        \left\| \nabla_{{\bmrho}} \left\{\dfrac{\tilde{g} e_\tp }{1+\dfrac{e_\tp}{1-e_\tp}e^{\bm{R}^\tp{\bmrho}}}\right\} \right\|
        & 
        = 
        \left\| \tilde{g}  g^\top \right\| 
        \frac{e_\tp}{1+\dfrac{e_\tp}{1-e_\tp}e^{g^\top  {\bmrho}}}
         \frac{  \frac{e_\tp}{1-e_\tp} e^{g^\top {\bmrho}}}{
        1+\dfrac{e_\tp}{1-e_\tp}e^{g^\top  {\bmrho}} 
        }
        \le \| \tilde{g} \|   \| g\| 
        \le \frac{\| \tilde{g} \|^2 + \| g\|^2}{2}, 
    \end{align*}
    and analogously, 
    \begin{align*}
        \left\| \nabla_{{\bmrho}} \left\{ \dfrac{\tilde{g} (1-e_\tp)}{1+\dfrac{1-e_\tp}{e_\tp}e^{-\bm{R}^\tp{\bmrho}}}\right\} \right\|
        & \le \frac{\| \tilde{g} \|^2 + \| g\|^2}{2}. 
    \end{align*}
    Using the dominated convergence theorem, we can then interchange the differentiation and expectation, and derive that  
    \begin{align}\label{eq:grad_f_gamma}
        \nabla_{{\bmrho}=\bm{0}}f({\bmrho}) & = 
        \frac{
        - \E \{ gg^\top e_\tp^2 (1-e_\tp) \} \E\{ e_\tp (1-e_\tp)\} 
        + \E\{g e_\tp (1-e_\tp) \} \E\{ g^\top e_\tp^2 (1-e_\tp) \}
        ]
        }{[ \E\{ e_\tp (1-e_\tp)\} ]^2 }
        \nonumber
        \\
        & \quad 
        - 
        \frac{
        \E\{ g g^\top e_\tp (1-e_\tp)^2 \} \E\{ e_\tp (1-e_\tp)\} 
        - 
        \E\{g e_\tp (1-e_\tp) \}  \E\{g^\top e_\tp (1-e_\tp)^2 \}
        }{
        [\E\{ e_\tp (1-e_\tp)\} ]^2
        }
        \nonumber
        \\
        & = 
        \frac{
        - \E \{ gg^\top e_\tp (1-e_\tp) \} \E\{ e_\tp (1-e_\tp)\}
        + \E\{g e_\tp (1-e_\tp) \} \E\{ g^\top e_\tp (1-e_\tp) \}
        }{
        [\E\{ e_\tp (1-e_\tp)\} ]^2
        },
    \end{align}
    where the last equality uses the fact that $e_\tp^2 (1-e_\tp) + e_\tp (1-e_\tp)^2 = e_\tp (1-e_\tp)$.

    To prove that $\nabla_{{\bmrho}=\bm{0}}f({\bmrho})\prec 0$, it suffices to prove that the numerator in \eqref{eq:grad_f_gamma} is negative definite, since the denominator there is positive under Condition \ref{cond:logistic_Z_est}(iii).
    Denote the negative numerator in \eqref{eq:grad_f_gamma} by 
    \begin{align*}
        \bm{M} & = \E \{ gg^\top e_\tp (1-e_\tp) \} \E\{ e_\tp (1-e_\tp)\}
        - \E\{g e_\tp (1-e_\tp) \} \E\{ g^\top e_\tp (1-e_\tp) \}.
    \end{align*}
    It suffices to prove that $\bm{M}$ is positive definite. Note that 
    for any $\bm{u}\in\mathbb{R}^p\backslash\{\bm{0}\}$, with $A=\bm{u}^\top g\sqrt{e_\tp(1-e_\tp)}, B=\sqrt{e_\tp(1-e_\tp)}$, we have
    \begin{align*}
        \bm{u}^\top \bm{M} \bm{u}=&\E(A^2)\E(B^2)-\{E(AB)\}^2\ge 0,
    \end{align*}
    where the last inequality follows from the Cauchy–Schwarz inequality. 
    Thus, $\bm{M}$ must be positive semidefinite. 

    We finally prove that $M$ is positive definite by contradiction. 
    If $\bm{M}$ is not positive definite, then there exists $\bm{u}\in\mathbb{R}^p$ such that $\bm{u}^\top \bm{M} \bm{u}=0$. 
    This implies that the corresponding $A$ and $B$ satisfy that $A=\lambda B$ almost surely for some constant $\lambda\in\mathbb{R}$, i.e., $\bm{u}^\top g\sqrt{e_\tp(1-e_\tp)}=\lambda \sqrt{e_\tp(1-e_\tp)}$ almost surely.
    Consequently, 
    \begin{align*}
        \Pr(\bm{u}^\top g \ne \lambda \mid 0 < e_\tp < 1 )
        & = 
        \frac{\Pr(\bm{u}^\top g \ne \lambda, 0 < e_\tp < 1 )}{\Pr(0 < e_\tp < 1)}
        \le 
        \frac{\Pr(\bm{u}^\top g\sqrt{e_\tp(1-e_\tp)}\ne \lambda \sqrt{e_\tp(1-e_\tp)})}{\Pr(0 < e_\tp < 1)}\\
        & = 0, 
    \end{align*}
    which is in contradiction to  Condition \ref{cond:logistic_Z_est}(iii). 
    Therefore, $\bm{M}$ must be positive semidefinite.

    From the above, Lemma \ref{lemma:unqiue_root} holds. 
\end{proof}

\begin{lemma}\label{lemma:str_posi}
    $\|\Phi({\bmtheta}, {\bmrho})\|_1$ is strictly positive for any   
    $(\bmtheta, \bmrho) \ne ({\bmtheta}_0, {\bmrho}_0)$ such that 
    ${\bmrho}$ is in a sufficiently small open neighborhood of ${\bmrho}_0\equiv 0$.
\end{lemma}

\begin{proof}[\bf Proof of Lemma \ref{lemma:str_posi}]
    Suppose there exists $({\bmtheta}', {\bmrho}')=(\tau', \bm{\beta}', \kappa'_{\treat}, \kappa'_{\control}, {\bmrho}')$ such that $\Phi({\bmtheta}', {\bmrho}')=0$, i.e., $\E\{p_{{\bmtheta}', i}(\bmO, W^\tp+\bm{R}^\tp {\bmrho}')\}=0$ for $1\leq i \leq 5$, for ${\bmrho}'$ that is in a sufficiently small open neighborhood of $0$ as in Lemma \ref{lemma:unqiue_root}. It suffices to prove that $({\bmtheta}', {\bmrho}')=({\bmtheta}_0, {\bmrho}_0)$.
    
    From Condition \ref{cond:logistic_Z_est}(ii) and using the mean value theorem, we can show that 
    $$
        \E\left(Zs(\bmX)-\dfrac{e^{\bm{\beta}^{\top}s(\bmX)}}{1+e^{\bm{\beta}^{\top}s(\bmX)}}s(\bmX)\right)=\bm{0}
    $$
    has a unique root $\bm{\beta}=\bm{\beta}_0$. 
    Thus, $\bm{\beta}'=\bm{\beta}_0$.  
    Moreover, from equations $\E\{p_{{\bmtheta}', i}(\bmO, W^\tp+\bm{R}^\tp {\bmrho}')\}=0$ for $i\in \{ 2,3,5 \}$, we can derive that 
    \begin{align*}
        \frac{\E[g(\bmX)Z/\{1+e^{W^\tp+\bm{R}^\tp{\bmrho'}+\bm{\beta}_0^{\top}s(\bmX)}\}]}{\E[Z/\{1+e^{W^\tp+\bm{R}^\tp{\bmrho'}+\bm{\beta}_0^{\top}s(\bmX)}\}]}
        =
        \frac{\E[g(\bmX)(1-Z)/\{1+e^{-W^\tp-\bm{R}^\tp{\bmrho'}-\bm{\beta}_0^{\top}s(\bmX)}\}]}{\E[(1-Z)/\{1+e^{-W^\tp-\bm{R}^\tp{\bmrho'}-\bm{\beta}_0^{\top}s(\bmX)}\}]},
    \end{align*}
    which, by Lemma \ref{lemma:unqiue_root}, has a unique root at ${\bmrho}_0$ in a small open neighborhood of $\bm{0}$. 
    Thus, we also have ${\bmrho}'={\bmrho}_0$. 
    From 
    $\E\{p_{{\bmtheta}', i}(\bmO, W^\tp+\bm{R}^\tp {\bmrho}')\}=0$ for $i\in \{ 2,3,4\}$, we can then derive that $\kappa_{\treat}'=\kappa_{\treat 0}, \kappa_{\control}'=\kappa_{\control 0}, \tau'=\tau_{\ow}$. 

    From the above, we must have 
    $({\bmtheta}', {\bmrho}')=({\bmtheta}_0, {\bmrho}_0)$. 
    Therefore, Lemma \ref{lemma:str_posi} holds. 
\end{proof}

\begin{lemma}[\citet{WellnerEmp:2005}, Lemma 6.1]\label{lemma:glivenko_cantelli}
    Suppose that $\mathcal{F}=\{f(\cdot, t): t \in T\}$ where the functions $f: \mathcal{X}\times T\rightarrow \mathbb{R}$ are continuous in $t$ for $\mathbb{P}-$almost all $\bm{x}\in\mathcal{X}$. Suppose that $T$ is compact and that the envelope function $F$ defined by $F(\bm{x})=\sup_{t\in T}|f(\bm{x}, t)|$ satisfies the outer expectation $\E^*F<\infty$. Then $\mathcal{F}$ is $\mathbb{P}$-Glivenko-Cantelli.
\end{lemma}

\begin{proof}[\bf Proof of Lemma \ref{lemma:glivenko_cantelli}]
    Lemma \ref{lemma:glivenko_cantelli} is from \citet[][Lemma 6.1]{WellnerEmp:2005}. 
\end{proof}

\begin{lemma}\label{lemma:p_gli_can}
    The class $\big\{p_{{\bmtheta}}(\bmO, W^\tp+\bm{R}^\tp{\bmrho}): ({\bmtheta}, {\bmrho})\in{\bm{\Theta}}\big\}$ is $\mathbb{P}$-Glivenko--Cantelli.
\end{lemma}

\begin{proof}[\bf Proof of Lemma \ref{lemma:p_gli_can}]
    Define the envelope function $B(\bmO, W^\tp, \bm{R}^\tp)\equiv\sup_{({\bmtheta}, {\bmrho})\in{\bm{\Theta}}}\|p_{{\bmtheta}}(\bmO, W^\tp+\bm{R}^\tp{\bmrho})\|_1$. 
    Since $({\bmtheta}, {\bmrho})\equiv(\tau, \bm{\beta}, \kappa_{\treat}, \kappa_{\control}, {\bmrho})$ is in a compact parameter space $\bm{\Theta}\subset \mathbb{R}\times\mathbb{R}^{\dim(\bm{\beta})}\times[\kappa_{\min}, \infty)\times[\kappa_{\min}, \infty)\times\mathbb{R}^{\dim({\bmrho})}$, we have 
    \begin{align*}
        \|p_{{\bmtheta}}(\bmO, W^\tp+\bm{R}^\tp{\bmrho})\|_1
        &\le \sum_{i=1}^5 \|p_{{\bmtheta}, i}(\bmO, W^\tp+\bm{R}^\tp{\bmrho})\|_1 \\
        &\le\|s(\bmX)\|_1+\kappa_{\min}^{-1}(|Y|+\|g(\bmX)\|_1)+M,
    \end{align*}
    where $M$ is some positive constant.
    Using condition \ref{cond:logistic_Z_est}(i), we know that 
    $\E\{B(\bmO, W^\tp, \bm{R}^\tp)\}<\infty$. 
    Therefore, from Lemma \ref{lemma:glivenko_cantelli}, the class $\big\{p_{{\bmtheta}}(\bmO, W^\tp+\bm{R}^\tp{\bmrho}): ({\bmtheta}, {\bmrho})\in{\bm{\Theta}}\big\}$ is $\mathbb{P}$-Glivenko-Cantelli.
\end{proof}

\begin{lemma}\label{lemma:p_donsker}
    The class $\big\{p_{{\bmtheta}}(\bmO, W^\tp+\bm{R}^\tp{\bmrho}): ({\bmtheta}, {\bmrho})\in{\bm{\Theta}}\big\}$ is $\mathbb{P}$-Donsker and $\E\{\|p_{{\bmtheta}_n}(\bmO, W^\tp+\bm{R}^\tp{\bmrho}_n)-p_{{\bmtheta}_0}(\bmO, W^\tp+\bm{R}^\tp{\bmrho}_0)\|_2^2\}\rightarrow0$, whenever $\|({\bmtheta}_n, {\bmrho}_n)-({\bmtheta}_0, {\bmrho}_0)\|_2 \rightarrow 0.$
\end{lemma}

\begin{proof}[\bf Proof of Lemma \ref{lemma:p_donsker}]
    We consider the difference in the values of the estimating equations $p_{{\bmtheta}}(\bmO, W^\tp+\bm{R}^\tp{\bmrho})$ evaluated at two parameter values. 
    Recall that $p_{{\bmtheta},i}(\bmO, W^\tp+\bm{R}^\tp{\bmrho})$ denotes the $i$-th component of $p_{{\bmtheta}}(\bmO, W^\tp+\bm{R}^\tp{\bmrho})$, for $1\le i \le 5$. 
    Let $({\bmtheta}_1,{\bmrho}_1)=(\tau_1, \bm{\beta}_1, \kappa_{\treat 1}, \kappa_{\control 1}, {\bmrho}_1)$ and $({\bmtheta}_2,{\bmrho}_2)=(\tau_2, \bm{\beta}_2, \kappa_{\treat 2}, \kappa_{\control 2}, {\bmrho}_2)$ be two points in the parameter space ${\bm{\Theta}}$.
    We further introduce the notation $\lesssim$, where $a\lesssim b $ if and only if $a\le Cb$ for some absolute constant $C>0$. 
    
    For the first estimating equation, 
    by definition and the mean value theorem, we have 
    \begin{align*}
        &\quad \ \|p_{{\bmtheta}_1, 1}(\bmO, W^\tp+\bm{R}^\tp{\bmrho}_1)-p_{{\bmtheta}_2, 1}(\bmO, W^\tp+\bm{R}^\tp{\bmrho}_2)\|_2\\
        &=\left\|\left(\dfrac{e^{\bm{\beta}_2^{\top}s(\bmX)}}{1+e^{\bm{\beta}_2^{\top}s(\bmX)}}-\dfrac{e^{\bm{\beta}_1^{\top}s(\bmX)}}{1+e^{\bm{\beta}_1^{\top}s(\bmX)}}\right)s(\bmX)\right\|_2 \\
        &=\left\|(\bm{\beta}_2-\bm{\beta}_1)^{\top}s(\bmX)s(\bmX)\dfrac{e^{t_*}}{(1+e^{t_*})^2}\right\|_2  \quad \ \ (\text{for some } t_* \text{ between } \bm{\beta}_1^{\top}s(\bmX) \text{ and } \bm{\beta}_2^{\top}s(\bmX) )\\
        &\leq |(\bm{\beta}_2-\bm{\beta}_1)^{\top}s(\bmX)|\|s(\bmX)\|_2 
        \leq\|\bm{\beta}_2-\bm{\beta}_1\|_2\|s(\bmX)\|_2\|s(\bmX)\|_2\\
        &\leq M_1(\bmX)\|\bm{\beta}_2-\bm{\beta}_1\|_2,
    \end{align*}
    where $M_1(\bmX)=\|s(\bmX)\|_2^2$.

    For the second estimating equation, 
    by definition and the mean value theorem, we have 
    \begin{align*}
        &\quad \ \|p_{{\bmtheta}_1, 2}(\bmO, W^\tp+\bm{R}^\tp{\bmrho}_1)-p_{{\bmtheta}_2, 2}(\bmO, W^\tp+\bm{R}^\tp{\bmrho}_2)\|_2\\
        &\le|\kappa_{\treat 2}-\kappa_{\treat 1}|+\left|\dfrac{1}{1+e^{W^\tp+\bm{R}^\tp{\bmrho}_2+\bm{\beta}_2^{\top}s(\bmX)}}-\dfrac{1}{1+e^{W^\tp+\bm{R}^\tp{\bmrho}_1+\bm{\beta}_1^{\top}s(\bmX)}}\right|\\
        & 
        \le 
        |\kappa_{\treat 2}-\kappa_{\treat 1}| 
        + 
        | \bm{R}^\tp({\bmrho}_2-{\bmrho}_1) +(\bm{\beta}_2-\bm{\beta}_1)^{\top}s(\bmX) | 
        \frac{e^{t_*}}{(1+e^{t_*})^2} 
        \\
        & \quad \ \ 
        \text{(for some $t_*$ between $W^\tp+\bm{R}^\tp{\bmrho}_2+\bm{\beta}_2^{\top}s(\bmX)$ and $W^\tp+\bm{R}^\tp{\bmrho}_1+\bm{\beta}_1^{\top}s(\bmX)$)}\\
        & \le 
        |\kappa_{\treat 2}-\kappa_{\treat 1}|+\|\bm{R}^\tp\|_2\|{\bmrho}_2-{\bmrho}_1\|_2 +\|\bm{\beta}_2-\bm{\beta}_1\|_2\|s(\bmX)\|_2\\
        & \le 
        M_2(\bmX, \bm{R}^\tp)(|\kappa_{\treat 2}-\kappa_{\treat 1}|+\|\bm{\beta}_2-\bm{\beta}_1\|_2+\|{\bmrho}_2-{\bmrho}_1\|_2),
    \end{align*}
    where 
    where $M_2(\bmX, \bm{R}^\tp)=1+\|s(\bmX)\|_2+\|\bm{R}^\tp\|_2$.
    
    For the third estimating equation, by the same as logic, we can derive that 
    $$
    \|p_{{\bmtheta}_1, 3}(\bmO, W^\tp+\bm{R}^\tp{\bmrho}_1)-p_{{\bmtheta}_2, 3}(\bmO, W^\tp+\bm{R}^\tp{\bmrho}_2)\|_2\le 
        M_2(\bmX, \bm{R}^\tp)(|\kappa_{\control 2}-\kappa_{\control 1}|+\|\bm{\beta}_2-\bm{\beta}_1\|_2+\|{\bmrho}_2-{\bmrho}_1\|_2).
    $$

    For the fourth estimating equation, by definition and the mean value theorem, we have 
    \begin{align*}
        &\quad \ \|p_{{\bmtheta}_1, 4}(\bmO, W^\tp+\bm{R}^\tp{\bmrho}_1)-p_{{\bmtheta}_2, 4}(\bmO, W^\tp+\bm{R}^\tp{\bmrho}_2)\|_2\\
        &\le
        |\tau_2-\tau_1|+\left|\dfrac{Y}{\{1+e^{W^\tp+\bm{R}^\tp{\bmrho}_2+\bm{\beta}_2^\top s(\bmX)}\}\kappa_{\treat 2}}-\dfrac{Y}{\{1+e^{W^\tp+\bm{R}^\tp{\bmrho}_1+\bm{\beta}_1^\top s(\bmX)}\}\kappa_{\treat 1}}\right|\\
        &\quad \ +\left|\dfrac{Y}{\{1+e^{-W^\tp-\bm{R}^\tp{\bmrho}_2-\bm{\beta}_2^\top s(\bmX)}\}\kappa_{\control 2}}-\dfrac{Y}{\{1+e^{-W^\tp-\bm{R}^\tp{\bmrho}_1-\bm{\beta}_1^\top s(\bmX)}\}\kappa_{\control 1}}\right|\\
        &\le 
        |\tau_2-\tau_1|+\left|\dfrac{Y}{\{1+e^{W^\tp+\bm{R}^\tp{\bmrho}_2+\bm{\beta}_2^\top s(\bmX)}\}\kappa_{\treat 2}}-\dfrac{Y}{\{1+e^{W^\tp+\bm{R}^\tp{\bmrho}_2+\bm{\beta}_2^\top s(\bmX)}\}\kappa_{\treat 1}}\right|\\
        &\quad \ +\left|\dfrac{Y}{\{1+e^{W^\tp+\bm{R}^\tp{\bmrho}_2+\bm{\beta}_2^\top s(\bmX)}\}\kappa_{\treat 1}}-\dfrac{Y}{\{1+e^{W^\tp+\bm{R}^\tp{\bmrho}_1+\bm{\beta}_1^\top s(\bmX)}\}\kappa_{\treat 1}}\right|\\
        &\quad \ +\left|\dfrac{Y}{\{1+e^{-W^\tp-\bm{R}^\tp{\bmrho}_2-\bm{\beta}_2^\top s(\bmX)}\}\kappa_{\control 2}}-\dfrac{Y}{\{1+e^{-W^\tp-\bm{R}^\tp{\bmrho}_2-\bm{\beta}_2^\top s(\bmX)}\}\kappa_{\control 1}}\right|\\
        &\quad \ +\left|\dfrac{Y}{\{1+e^{-W^\tp-\bm{R}^\tp{\bmrho}_2-\bm{\beta}_2^\top s(\bmX)}\}\kappa_{\control 1}}-\dfrac{Y}{\{1+e^{-W^\tp-\bm{R}^\tp{\bmrho}_1-\bm{\beta}_1^\top s(\bmX)}\}\kappa_{\control 1}}\right|\\
        &\lesssim|\tau_2-\tau_1|+|Y||\kappa_{\treat 2}-\kappa_{\treat 1}|+|Y||\kappa_{\control 2}-\kappa_{\control 1}|+\|\bm{\beta}_2-\bm{\beta}_1\|_2\|s(\bmX)\|_2 |Y|+\|\bm{R}^\tp \|_2 |Y|\|{\bmrho}_2-{\bmrho}_1\|_2\\
        &\lesssim M_3(\bmX, \bm{R}^\tp)(|\tau_2-\tau_1|+|\kappa_{\treat 2}-\kappa_{\treat 1}|+|\kappa_{\control 2}-\kappa_{\control 1}|+\|\bm{\beta}_2-\bm{\beta}_1\|_2+\|{\bmrho}_2-{\bmrho}_1\|_2),
    \end{align*}
    where     
    $M_3(\bmX, \bm{R}^\tp)=1+|Y|+\|s(\bmX)\|_2|Y|+\|\bm{R}^\tp \|_2|Y|$, 
    and the second last inequality uses similar bounds as the second estimating equation and the fact that the parameter space is compact.

    For the fifth estimating equation, by the similar logic as the bound for the fourth estimating equation, we have  
    \begin{align*}
        &\quad \ \|p_{{\bmtheta}_1, 5}(\bmO, W^\tp+\bm{R}^\tp{\bmrho}_1)-p_{{\bmtheta}_2, 5}(\bmO, W^\tp+\bm{R}^\tp{\bmrho}_2)\|_2\\
        &\le
        \left\|\dfrac{g(\bmX)}{\{1+e^{W^\tp+\bm{R}^\tp{\bmrho}_2+\bm{\beta}_2^\top s(\bmX)}\}\kappa_{\treat 2}}-\dfrac{g(\bmX)}{\{1+e^{W^\tp+\bm{R}^\tp{\bmrho}_1+\bm{\beta}_1^\top s(\bmX)}\}\kappa_{\treat 1}}\right\|_2\\
        &\quad \ +\left\|\dfrac{g(\bmX)}{\{1+e^{-W^\tp-\bm{R}^\tp{\bmrho}_2-\bm{\beta}_2^\top s(\bmX)}\}\kappa_{\control 2}}-\dfrac{g(\bmX)}{\{1+e^{-W^\tp-\bm{R}^\tp{\bmrho}_1-\bm{\beta}_1^\top s(\bmX)}\}\kappa_{\control 1}}\right\|_2\\
        &\lesssim\|g(\bmX)\|_2|\kappa_{\treat 2}-\kappa_{\treat 1}|+\|g(\bmX)\|_2|\kappa_{\control 2}-\kappa_{\control 1}|+\|g(\bmX)\|_2 \|\bm{\beta}_2-\bm{\beta}_1\|_2 \|s(\bmX)\|_2+\| g(\bmX)\|_2 \|\bm{R}^\tp\|_2\|{\bmrho}_2-{\bmrho}_1\|_2\\
        &\lesssim M_4(\bmX, \bm{R}^\tp)(|\kappa_{\treat 2}-\kappa_{\treat 1}|+|\kappa_{\control 2}-\kappa_{\control 1}|+\|\bm{\beta}_2-\bm{\beta}_1\|_2+\|{\bmrho}_2-{\bmrho}_1\|_2),
    \end{align*}
    where $M_4(\bmX, \bm{R}^\tp)=1+\|g(\bmX)\|_2+\|s(\bmX)\|_2\|g(\bmX)\|_2+\|\bm{R}^\tp \|_2\|g(\bmX)\|_2$.

    The bounds for the five estimating equations then imply that 
    \begin{align*}
        &\quad \ \|p_{{\bmtheta}_1}(\bmO, W^\tp+\bm{R}^\tp{\bmrho}_1)-p_{{\bmtheta}_2}(\bmO, W^\tp+\bm{R}^\tp{\bmrho}_2)\|_2\\
        &\le\sum_{i=1}^5\|p_{{\bmtheta}_1,i}(\bmO, W^\tp+\bm{R}^\tp{\bmrho}_1)-p_{{\bmtheta}_2, i}(\bmO, W^\tp+\bm{R}^\tp{\bmrho}_2)\|_2\\
        &\lesssim M(\bmX, \bm{R}^\tp)(|\tau_2-\tau_1|+|\kappa_{\treat 2}-\kappa_{\treat 1}|+|\kappa_{\control 2}-\kappa_{\control 1}|+\|\bm{\beta}_2-\bm{\beta}_1\|_2+\|{\bmrho}_2-{\bmrho}_1\|_2),
    \end{align*}
    where $M(\bmX, \bm{R}^\tp)=\|s(\bmX)\|_2^2+(1+\|s(\bmX)\|_2+\|\bm{R}^\tp\|_2)(1+|Y|+\|g(\bmX)\|_2)$.
    Note that $\bm{R}^\tp=g^\top(\bmX)$. From Condition \ref{cond:logistic_Z_est}(i),
    $\E\{M(\bmX,\bm{R}^\tp)^2\}<\infty$. This then implies that the class $\big\{p_{{\bmtheta}}(\bmO, W^\tp+\bm{R}^\tp{\bmrho}): ({\bmtheta}, {\bmrho})\in{\bm{\Theta}}\big\}$ is $\mathbb{P}$-Donsker. 
    Furthermore, by the above inequality, we can see that $\E[\|p_{{\bmtheta}_n}(\bmO, W^\tp+\bm{R}^\tp{\bmrho}_n)-p_{{\bmtheta}_0}(\bmO, W^\tp+\bm{R}^\tp{\bmrho}_0)\|_2^2]\rightarrow0$ when $\|({\bmtheta}_n, {\bmrho}_n)-({\bmtheta}_0, {\bmrho}_0)\|_2\rightarrow 0.$
    Therefore, Lemma \ref{lemma:p_donsker} holds. 
\end{proof}

\begin{lemma}\label{lemma:well_define}
    $\E\{p_{{\bmtheta}}(\bmO, W^\tp+\bm{R}^\tp{\bmrho})p^\top_{{\bmtheta}}(\bmO, W^\tp+\bm{R}^\tp{\bmrho})\}<\infty$ and $
    \E\big\{\nabla_{({\bmtheta}, {\bmrho})=({\bmtheta}_0, {\bmrho}_0)}p_{{\bmtheta}}(\bmO, W^\tp+\bm{R}^\tp{\bmrho})\big\}$ is invertible.
\end{lemma}

\begin{proof}[\bf Proof of Lemma \ref{lemma:well_define}]
    By Condition \ref{cond:logistic_Z_est}(i), we can verify that 
    $\E\{p_{{\bmtheta}}(\bmO, W^\tp+\bm{R}^\tp{\bmrho})p^\top_{{\bmtheta}}(\bmO, W^\tp+\bm{R}^\tp{\bmrho})\}<\infty$.
    Below we prove that $
    \E\big\{\nabla_{({\bmtheta}, {\bmrho})=({\bmtheta}_0, {\bmrho}_0)}p_{{\bmtheta}}(\bmO, W^\tp+\bm{R}^\tp{\bmrho})\big\}$ is invertible. Let   
    $g = g(\bmX)$, $s = s(\bmX)$ and $e_\tp=e_\tp(\bmX,\bmU)$ for notational convenience.

    Recall that $\bm{R}^\tp=g^\top$ and $({\bmtheta},{\bmrho})=(\tau, \bm{\beta}, \kappa_{\treat}, \kappa_{\control}, {\bmrho})$. We can verify that 
    \begin{align*}
        & \quad \ \nabla_{({\bmtheta}, {\bmrho})}p_{{\bmtheta}}(\bmO, W^\tp+\bm{R}^\tp{\bmrho})
        \\
        & = 
        \begin{pmatrix}
            0 & \dfrac{- e^{\bm{\beta}^\top s } s s^\top}{(1+e^{\bm{\beta}^{\top}s})^2} & 0 & 0 & 0\\
            0 & * & -1 & 0 & -Z \frac{e^{W^\tp+\bm{R}^\tp{\bmrho}+\bm{\beta}^\top s}}{(1+e^{W^\tp+\bm{R}^\tp{\bmrho}+\bm{\beta}^\top s})^2}  g^\top \\
            0 & * & 0 & -1 & (1-Z) \frac{  e^{-W^\tp-\bm{R}^\tp{\bmrho}-\bm{\beta}^\top s}}{(1+e^{-W^\tp-\bm{R}^\tp{\bmrho}-\bm{\beta}^\top s})^2}  g^\top \\
            -1 & * & * & * & * \\
            \bm{0} & * & \bm{a}_1 & \bm{a}_2 & \bm{a}_3 
        \end{pmatrix}
    \end{align*}
    where 
    \begin{align*}
        \bm{a}_1 & =
        - \dfrac{Zg}{\{1+e^{W^\tp+\bm{R}^\tp{\bmrho}+\bm{\beta}^\top s}\}\kappa_{\treat}^2}, 
        \quad \quad \quad 
        \bm{a}_2 = \dfrac{(1-Z)g}{\{1+e^{-W^\tp-\bm{R}^\tp{\bmrho}-\bm{\beta}^\top s}\}\kappa_{\control}^2}
        \\
        \bm{a}_3 & = - \frac{Z e^{W^\tp+\bm{R}^\tp{\bmrho}+\bm{\beta}^\top s}}{\kappa_\treat(1+e^{W^\tp+\bm{R}^\tp{\bmrho}+\bm{\beta}^\top s})^2} g g^\top -  \frac{(1-Z) e^{-W^\tp-\bm{R}^\tp{\bmrho}-\bm{\beta}^\top s}}{\kappa_c(1+e^{-W^\tp-\bm{R}^\tp{\bmrho}-\bm{\beta}^\top s})^2} gg^\top. 
    \end{align*}
    Note that 
    $W^\tp+\bm{R}^\tp{\bmrho}_0+\bm{\beta}_0^\top s = \logit(e_\tp)$,  
    which, by some algebra, further implies that 
    \begin{align*}
    \dfrac{1}{1+e^{W^\tp+\bm{R}^\tp{\bmrho}+\bm{\beta}^\top s}} & = 1 - e_\tp(\bmX,\bmU), \quad \quad \quad 
        \dfrac{1}{\{1+e^{-W^\tp-\bm{R}^\tp{\bmrho}-\bm{\beta}^\top s}\}} = e_\tp(\bmX,\bmU),
    \\
        \frac{e^{W^\tp+\bm{R}^\tp{\bmrho}+\bm{\beta}^\top s}}{(1+e^{W^\tp+\bm{R}^\tp{\bmrho}+\bm{\beta}^\top s})^2}
        & =
        \frac{e^{-W^\tp-\bm{R}^\tp{\bmrho}-\bm{\beta}^\top s}}{(1+e^{-W^\tp-\bm{R}^\tp{\bmrho}-\bm{\beta}^\top s})^2}= 
        e_\tp(\bmX,\bmU) \{1 - e_\tp(\bmX,\bmU)\}.
    \end{align*}
    By these facts, we can verify that $\kappa_{\treat 0} = \kappa_{\control 0} = \E[e_\tp  \{1 - e_\tp \}]$, 
    and 
    \begin{align*}
        & \quad \ \E \big\{ \nabla_{({\bmtheta}, {\bmrho}) = ({\bmtheta}_0,{\bmrho}_0)}p_{{\bmtheta}}(\bmO, W^\tp+\bm{R}^\tp{\bmrho}) \big\}
        \\
        & =  
        \begin{pmatrix}
            0 & 
            \bm{B}
            & 0 & 0 & 0\\
            0 & * & -1 & 0 & \E\{ -e_\tp^2(1-e_\tp) g^\top \} \\
            0 & * & 0 & -1 & \E\{e_\tp(1-e_\tp)^2 g^\top \}   \\
            -1 & * & * & * & * \\
            \bm{0} & * & -\kappa_{\treat 0}^{-2}\E\{ e_\tp(1-e_\tp) g\} & \kappa_{\control 0}^{-2}\E\{ e_\tp(1-e_\tp) g\} & -\kappa_{\treat 0}^{-1} \E\{ e_\tp(1-e_\tp) gg^\top\}
        \end{pmatrix}, 
    \end{align*}
    where $\bm{B} = - \E\{  e^{\bm{\beta}_0^\top s}/(1+e^{\bm{\beta}_0^{\top}s})^2 \cdot ss^\top \}$. This then implies that 
    \begin{align*}
        \left|\text{det}\left( \E \big\{ \nabla_{({\bmtheta}, {\bmrho}) = ({\bmtheta}_0,{\bmrho}_0)}p_{{\bmtheta}}(\bmO, W^\tp+\bm{R}^\tp{\bmrho}) \big\} \right) \right|
        & = \left|
        \text{det}(\bm{B})
        \right| 
        \cdot 
        \left|
        \text{det}(\bm{A})
        \right|, 
    \end{align*}
    where 
    \begin{align*}
        \bm{A} & = 
        \begin{pmatrix}
             -1 & 0 & \E\{ -e_\tp^2(1-e_\tp) g^\top \} \\
             0 & -1 & \E\{e_\tp(1-e_\tp)^2 g^\top \}   \\
             -\kappa_{\treat 0}^{-2}\E\{ e_\tp(1-e_\tp) g\} & \kappa_{\control 0}^{-2}\E\{ e_\tp(1-e_\tp) g\} & -\kappa_{\treat 0}^{-1} \E\{ e_\tp(1-e_\tp) gg^\top\}
        \end{pmatrix} \equiv 
        \begin{pmatrix}
            -\bm{I}_{2\times 2} & \bm{A}_{12}\\
            \bm{A}_{21} & \bm{A}_{22}
        \end{pmatrix},
    \end{align*}
    and $\bm{I}_{2\times 2}$ is an $2\times 2$ diagonal matrix. 
    By the property of block matrix, we can derive that 
    \begin{align*}
        \text{det}(\bm{A})
        & =
        \text{det} ( -\bm{I}_{2\times 2}
        )
                \cdot
        \text{det} 
        \left( 
        \bm{A}_{22} - \bm{A}_{21} \cdot (-\bm{I}_{2\times 2}) \cdot \bm{A}_{12}
        \right)        
    = \text{det} 
        \left( 
        \bm{A}_{22} + \bm{A}_{21} \cdot \bm{A}_{12}
        \right) 
    \\
    & = 
    \text{det} 
        \left( 
        -\kappa_{\treat 0}^{-1} \E\{ e_\tp(1-e_\tp) gg^\top\} + \kappa_{\treat 0}^{-2}\E\{ e_\tp(1-e_\tp) g\} \cdot 
        \E\{ e_\tp(1-e_\tp) g^\top \}
        \right)
    \\
    & = \text{det} 
        \left( 
        -\kappa_{\treat 0}^{-2} \bm{M}
        \right),
    \end{align*}
    where 
    $
        \bm{M} =  \E \{ gg^\top e_\tp (1-e_\tp) \} \E\{ e_\tp (1-e_\tp)\}
        - \E\{g e_\tp (1-e_\tp) \} \E\{ g^\top e_\tp (1-e_\tp) \}
    $
    is defined as in Lemma \ref{lemma:unqiue_root}. 
    From Condition \ref{cond:logistic_Z_est} and the proof of Lemma  \ref{lemma:unqiue_root}, we can know that both $-B$ and $M$ are positive definite. 
    Therefore, 
    $\E \big\{ \nabla_{({\bmtheta}, {\bmrho}) = ({\bmtheta}_0,{\bmrho}_0)}p_{{\bmtheta}}(\bmO, W^\tp+\bm{R}^\tp{\bmrho}) \big\}$ must have a nonzero determinant and is thus invertible. 
    
    From the above,  Lemma \ref{lemma:well_define} holds. 
    \end{proof}

\begin{proof}[\bf Proof of Theorem \ref{thm:valid_bootstrap_R_g}]
    By Lemma \ref{lemma:asy_ci}, it suffices to verify the following four conditions:
    \begin{enumerate}
        \item[(i)]  For any sequence $\{({\bmtheta}_n, {\bmrho}_n)\}\in\bm{\Theta}, \Phi({\bmtheta}_n, {\bmrho}_n)\rightarrow 0$ implies $\|({\bmtheta}_n, {\bmrho}_n)-({\bmtheta}_0, {\bmrho}_0)\|\rightarrow 0$;
        \item[(ii)] The class $\big\{p_{{\bmtheta}}(\bmO, W^\tp+\bm{R}^\tp{\bmrho}): ({\bmtheta}, {\bmrho})\in{\bm{\Theta}}\big\}$ is $\mathbb{P}$-Glivenko--Cantelli;
        \item[(iii)] The class $\big\{p_{{\bmtheta}}(\bmO, W^\tp+\bm{R}^\tp{\bmrho}): ({\bmtheta}, {\bmrho})\in{\bm{\Theta}}\big\}$ is $\mathbb{P}$-Donsker and $\E[\|p_{{\bmtheta}_n}(\bmO, W^\tp+\bm{R}^\tp{\bmrho}_n)-p_{{\bmtheta}_0}(\bmO, W^\tp+\bm{R}^\tp{\bmrho}_0)\|^2]\rightarrow0$, whenever $\|({\bmtheta}_n, {\bmrho}_n)-({\bmtheta}_0, {\bmrho}_0)\|\rightarrow 0$;
        \item[(iv)] $\E\{p_{{\bmtheta}}(\bmO, W^\tp+\bm{R}^\tp{\bmrho})p^\top_{{\bmtheta}}(\bmO, W^\tp+\bm{R}^\tp{\bmrho})\}<\infty$ and $\dot{\Phi}_0\equiv\E\big\{\nabla_{({\bmtheta}, {\bmrho})=({\bmtheta}_0, {\bmrho}_0)}p_{{\bmtheta}}(\bmO, W^\tp+\bm{R}^\tp{\bmrho})\big\}$ is invertible.
    \end{enumerate}
    These four conditions are verified by applying Lemmas \ref{lemma:str_posi}, \ref{lemma:p_gli_can}, \ref{lemma:p_donsker}, and \ref{lemma:well_define}, respectively.
    Therefore, we derive Theorem \ref{thm:valid_bootstrap_R_g}. 
\end{proof}

\subsection{Proof of Theorem \ref{thm:ci_tau_pred} and Corollaries \ref{cor:ci_tau} and \ref{cor:test_U}}

\begin{proof}[\bf Proof of Theorem \ref{thm:ci_tau_pred}]
    
Let $\varepsilon = \min\{\log(\Lambda_1'/\Lambda_1), \log(\Lambda_0'/\Lambda_0)\}$, and 
let $\mathcal{W}_n^\varepsilon(\vec{\bm{\Lambda}}, \vec{\bm{\delta}}_{n})$ be the  $\varepsilon$-neighborhood of $\mathcal{W}_n(\vec{\bm{\Lambda}}, \vec{\bm{\delta}}_{n})$ defined as in Section \ref{sec:pred_set_W}.
We can verify that 
$
\mathcal{W}_n^\varepsilon(\vec{\bm{\Lambda}}, \vec{\bm{\delta}}_{n})  = \mathcal{W}_n(\vec{\bm{\Lambda}}\cdot \exp(\varepsilon), \vec{\bm{\delta}}_{n}) \subset \mathcal{W}_n(\vec{\bm{\Lambda}}', \vec{\bm{\delta}}_{n}). 
$
These imply that
$
    \UU_{\alpha}(\vec{\bm{\Lambda}}', \vec{\bm{\delta}}_{n}) 
    \ge \UU_{\alpha}(\vec{\bm{\Lambda}}\cdot \exp(\varepsilon), \vec{\bm{\delta}}_{n}).
$
From Theorem \ref{thm:bootstrap_over_identify_pred_set}, we then have 
\begin{align*}
    & \quad \ \limsup_{n\rightarrow \infty} \Pr\{ \tau_\ow > \UU_{\alpha}(\vec{\bm{\Lambda}}', \vec{\bm{\delta}}_{n}), \  \bmW_{1:n}^\tp \in \mathcal{W}_n(\vec{\bm{\Lambda}}, \vec{\bm{\delta}}_{n}) \}\\
    & \le \limsup_{n\rightarrow \infty} \Pr\{ \tau_\ow > \UU_{\alpha}(\vec{\bm{\Lambda}}\cdot \exp(\varepsilon), \vec{\bm{\delta}}_{n}), \  \bmW_{1:n}^\tp \in \mathcal{W}_n(\vec{\bm{\Lambda}}, \vec{\bm{\delta}}_{n}) \}\\
    & \le \alpha.
\end{align*}
By the same logic, we can then derive that 
\begin{align*}
    \limsup_{n\rightarrow \infty} \Pr\{ \tau_\ow < \LL_{\alpha}(\vec{\bm{\Lambda}}', \vec{\bm{\delta}}_{n}), \  \bmW_{1:n}^\tp \in \mathcal{W}_n(\vec{\bm{\Lambda}}, \vec{\bm{\delta}}_{n}) \} & \le \alpha.
\end{align*}
Consequently, we also have 
\begin{align*}
    \limsup_{n\rightarrow \infty} \Pr\{ \tau_\ow \notin [\LL_{\alpha}(\vec{\bm{\Lambda}}', \vec{\bm{\delta}}_{n}), \UU_{\alpha}(\vec{\bm{\Lambda}}', \vec{\bm{\delta}}_{n})], \  \bmW_{1:n}^\tp \in \mathcal{W}_n(\vec{\bm{\Lambda}}, \vec{\bm{\delta}}_{n}) \}  \le 2 \alpha.
\end{align*}
Therefore, we derive Theorem \ref{thm:ci_tau_pred}. 
\end{proof}

\begin{proof}[\bf Proof of Corollary \ref{cor:ci_tau}]
    From Theorem \ref{thm:bootstrap_over_identify_pred_set}, it suffices to show that $\mathcal{W}_n(\vec{\bm{\Lambda}}, \vec{\bm{\delta}}_{n})$ is an asymptotic $1-\zeta$ prediction set for $W_{1:n}^\tp$.
    Let $\bm{Z}_{1:n} = (Z_1, Z_2, \ldots, Z_n)$. We have 
    \begin{align*}
        &\quad \ \Pr\left\{\bmW_{1:n}^\tp \in \mathcal{W}_n(\vec{\bm{\Lambda}}, \vec{\bm{\delta}}_{n}) \mid \bm{Z}_{1:n}\right\}\\
        &= \Pr\left[\sum_{i: Z_i=1} \I \{ |W_i^\tp| \leq \log(\Lambda_1) \} \geq n_1(1-\delta_{n,1}), \sum_{i: Z_i=0} \I \{ |W_i^\tp| \leq \log(\Lambda_0) \} \geq n_0(1-\delta_{n,0}) \mid \bm{Z}_{1:n} \right]\\
        &=\Pr\left[\sum_{i: Z_i=1} \I \{ |W_i^\tp| \leq \log(\Lambda_1) \} \geq n_1(1-\delta_{n,1}) \mid \bm{Z}_{1:n}\right] \\
        & \quad \ \cdot \Pr\left[\sum_{i: Z_i=0} \I \{ |W_i^\tp| \leq \log(\Lambda_0) \} \geq n_0(1-\delta_{n,0}) \mid \bm{Z}_{1:n}\right], 
    \end{align*}
    where the first equality follows from the definition of $\mathcal{W}_n(\vec{\bm{\Lambda}}, \vec{\bm{\delta}}_{n})$ in \eqref{eq:Lambda_delta_w_pred_set}, 
    and 
    the second equality follows from the fact that $W_i^\tp$s are conditionally independent given $\bm{Z}_{1:n}$.
    From Assumption \ref{asmp:robust_sen_spec}, conditional on $\bm{Z}_{1:n},$ 
    $\sum_{i: Z_i=1} \I \{ |W_i^\tp| \leq \log(\Lambda_1) \}$ is stochastically greater than or equal to $\Bin(n_1, 1-\delta_1)$,  and $\sum_{i: Z_i=0} \I \{ |W_i^\tp| \leq \log(\Lambda_0) \}$ is stochastically greater than or equal to $\Bin(n_0, 1-\delta_0)$ 
    Therefore, we further have 
    \begin{align*}
        &\quad \ \Pr\left\{\bmW_{1:n}^\tp \in \mathcal{W}_n(\vec{\bm{\Lambda}}, \vec{\bm{\delta}}_{n}) \mid \bm{Z}_{1:n}\right\}\\
        &\ge \Pr\left\{\Bin(n_1, 1-\delta_1) \geq n_1(1-\delta_{n,1})\right\}\cdot \Pr\left\{\Bin(n_0, 1-\delta_0) \geq n_0(1-\delta_{n,0})\right\}\\
        & = \Pr\left\{\Bin(n_1, \delta_1) \leq n_1 \delta_{n,1} \right\}\cdot \Pr\left\{\Bin(n_0, \delta_0) \leq n_0 \delta_{n,0} \right\} \\
        &\ge 1-\zeta,
    \end{align*}
    where the last inequality holds because $n_z\delta_{n,z}$ is the $\sqrt{1-\zeta}$th quantile of $\Bin(n_z, \delta_z)$ for $z=0,1$.
    By the law of total expectation, this immediately implies that 
    $
    \Pr\{\bmW_{1:n}^\tp \in \mathcal{W}_n(\vec{\bm{\Lambda}}, \vec{\bm{\delta}}_{n}) \} \ge 1-\zeta. 
    $    
    Therefore, we can derive Corollary \ref{cor:ci_tau}.   
\end{proof}

\begin{proof}[\bf Proof of Corollary \ref{cor:test_U}]
    If $H_{\vec{\bm{q}},\vec{\bm{\Lambda}}}^{\text{U}}$ holds, then we have 
    $\OR^\treat_{\m[q_1]} \le \Lambda_1$ and  $\OR^\control_{\m[q_0]} \le \Lambda_0$. 
    By definition, this implies that 
    $W_{1:n}^\tp \in \mathcal{W}(\vec{\bm{\Lambda}}, \vec{\bm{1}}-\vec{\bm{q}})$. 
    Consequently, under the assumption that $\tau_{\ow} = 0$, we have 
    \begin{align*}
        & \quad \ \Pr\big\{
        H_{\vec{\bm{q}},\vec{\bm{\Lambda}}}^{\text{U}} \text{ holds and } 0\notin [\LL_{\alpha/2}(\xi \vec{\bm{\Lambda}}, \vec{\bm{1}}-\vec{\bm{q}}), \UU_{\alpha/2}(\xi \vec{\bm{\Lambda}}, \vec{\bm{1}}-\vec{\bm{q}})]
        \big\}\\
        & = \Pr\big\{
         \tau_{\ow} \notin [\LL_{\alpha/2}(\xi \vec{\bm{\Lambda}}, \vec{\bm{1}}-\vec{\bm{q}}), \UU_{\alpha/2}(\xi \vec{\bm{\Lambda}}, \vec{\bm{1}}-\vec{\bm{q}})], \ 
        \bmW_{1:n}^\tp \in \mathcal{W}(\vec{\bm{\Lambda}}, \vec{\bm{1}}-\vec{\bm{q}}) 
        \big\}. 
    \end{align*}
    Corollary \ref{cor:test_U} then follows immediately from Theorem \ref{thm:ci_tau_pred}. 
\end{proof}

\subsection{Proof of Theorem \ref{thm:U_quant_pred_int}}
\begin{proof}[\bf Proof of Theorem \ref{thm:U_quant_pred_int}]
Adopt the notation from Section \ref{sec:formu_z_est}, including the estimating equation $p_{{\bmtheta}}(\cdot)$ in \eqref{eq:est_equ_ow} and all the involved notation there. 
For any $\vec{\bm{q}}\in[0,1]^2$,
define 
\begin{align*}
    \hat{\tau}_{\bstrap, \vec{\bm{q}}}^{\max} = \ & \sup \tau 
    \quad 
    \textsc{and}
    \quad 
    \hat{\tau}_{\bstrap, \vec{\bm{q}}}^{\min}  = \inf \tau
    \\
    \text{subject to} \quad & 
    \sum_{i=1}^n k_{\bstrap i} p_{\bmtheta}(\bmO_{i}, W_i) =\bm{0}, 
    \nonumber
    \\
    & (W_1, W_2, \ldots, W_n) \in \mathcal{W}_n\big(\xi(\OR_{\m[q_1]}, \OR_{\m[q_0]}), \vec{\bm{1}}-\vec{\bm{q}}\big),
    \nonumber
\end{align*}
and
\begin{align}\label{eq:quantile_proof_med}
    \check{\tau}_{\bstrap, \vec{\bm{q}}}^{\max} = \ & \sup \tau 
    \quad 
    \textsc{and}
    \quad 
    \check{\tau}_{\bstrap, \vec{\bm{q}}}^{\min}  = \inf \tau
    \\
    \text{subject to} \quad & 
    \sum_{i=1}^n k_{\bstrap i} p_{\bmtheta}(\bmO_{i}, W_i+\bm{R}^\tp\hat{{\bmrho}}_{\bstrap}) =\bm{0}, 
    \nonumber
    \\
    & (W_1, W_2, \ldots, W_n) \in \mathcal{W}_n\big((\OR_{\m[q_1]}, \OR_{\m[q_0]}), \vec{\bm{1}}-\vec{\bm{q}}\big),
    \nonumber
\end{align}
where $\mathcal{W}_n(\cdot, \cdot)$ is defined as in \eqref{eq:Lambda_delta_w_pred_set},  
$(\hat{{\bmtheta}}_\bstrap, \hat{{\bmrho}}_\bstrap)\equiv (\hat{\tau}_\bstrap, \hat{\bm{\beta}}_\bstrap, \hat{\kappa}_{\treat\bstrap}, \hat{\kappa}_{\control\bstrap}, \hat{{\bmrho}}_\bstrap)$ denotes the root to the estimating equation
$
    \sum_{i=1}^n k_{\bstrap i} p_{\bmtheta}(\bmO_{i}, W_i^\tp+\bm{R}^\tp{\bmrho}) =0,
$
and 
$(k_{\bstrap 1}, k_{\bstrap 2}, \ldots, k_{\bstrap n})$ follows the distribution of $\text{Multinomial} (n; n^{-1}, n^{-1}, \ldots, n^{-1})$.
By definition, we can verify that the $\log(\xi)$-neighborhood of $\mathcal{W}_n\big((\OR_{\m[q_1]}, \OR_{\m[q_0]}), \vec{\bm{1}}-\vec{\bm{q}}\big)$ is $\mathcal{W}_n\big(\xi(\OR_{\m[q_1]}, \OR_{\m[q_0]}), \vec{\bm{1}}-\vec{\bm{q}}\big)$.

Let $\Omega=B^{-1}\sum_{b=1}^B\I\{\max_{1\le i\le n}\|\bm{R}_i^\tp\hat{{\bmrho}}_b\|>\log(\xi)\}$. 
By the same logic as Lemma \ref{lemma:quantile_compare_gen}, we know that for all $\vec{\bm{q}}\in[0,1]^2$,
\begin{align*}
        \big[Q_{\alpha/2+\Omega}(\check{\tau}_{\bstrap, \vec{\bm{q}}}^{\min}), Q_{1-\alpha/2-\Omega}(\check{\tau}_{\bstrap, \vec{\bm{q}}}^{\max})\big]\subset\big[Q_{\alpha/2}(\hat{\tau}_{\bstrap, \vec{\bm{q}}}^{\min}), Q_{1-\alpha/2}(\hat{\tau}_{\bstrap, \vec{\bm{q}}}^{\max})\big].
    \end{align*}
Additionally, since $(\hat{{\bmtheta}}_\bstrap, \hat{{\bmrho}}_\bstrap)\equiv (\hat{\tau}_\bstrap, \hat{\bm{\beta}}_\bstrap, \hat{\kappa}_{\treat\bstrap}, \hat{\kappa}_{\control\bstrap}, \hat{{\bmrho}}_\bstrap)$ and $W_{1:n}^\tp$ are in the feasible region of \eqref{eq:quantile_proof_med} 
we further have,  for all $\vec{\bm{q}}\in[0,1]^2$, 
\begin{align*}
    \big[Q_{\alpha/2+\Omega}(\hat{\tau}_\bstrap), Q_{1-\alpha/2-\Omega}(\hat{\tau}_\bstrap)\big]\subset\big[Q_{\alpha/2+\Omega}(\check{\tau}_{\bstrap, \vec{\bm{q}}}^{\min}), Q_{1-\alpha/2-\Omega}(\check{\tau}_{\bstrap, \vec{\bm{q}}}^{\max})\big]
    \subset\big[Q_{\alpha/2}(\hat{\tau}_{\bstrap, \vec{\bm{q}}}^{\min}), Q_{1-\alpha/2}(\hat{\tau}_{\bstrap, \vec{\bm{q}}}^{\max})\big].
\end{align*}
Thus, 
\begin{align}\label{eq:sub_inter}
    \big[Q_{\alpha/2+\Omega}(\hat{\tau}_\bstrap), Q_{1-\alpha/2-\Omega}(\hat{\tau}_\bstrap)\big]\subset
    \bigcap_{\vec{\bm{q}}\in[0,1]^2}
    \big[Q_{\alpha/2}(\hat{\tau}_{\bstrap, \vec{\bm{q}}}^{\min}), Q_{1-\alpha/2}(\hat{\tau}_{\bstrap, \vec{\bm{q}}}^{\max})\big]. 
\end{align}

Suppose that 
$\big(\OR_{\m[q_1]}^\treat, \OR_{\m[q_0]}^\control\big) \notin \mathcal{I}_{\vec{\bm{q}}, \xi}^\alpha$ for some  $\vec{\bm{q}}\in[0,1]^2$. 
By definition, 
for some  $\vec{\bm{q}}\in[0,1]^2$, 
we have 
$\xi \big(\OR_{\m[q_1]}^\treat, \OR_{\m[q_0]}^\control\big) \notin \mathcal{I}_{\vec{\bm{q}}}^\alpha$ 
and consequently 
\begin{align*}
    0 & \notin 
    \left[ \LL_{\alpha/2}(\xi \big(\OR_{\m[q_1]}^\treat, \OR_{\m[q_0]}^\control\big), \vec{\bm{1}}-\vec{\bm{q}}), \ \ 
    \UU_{\alpha/2}(\xi \big(\OR_{\m[q_1]}^\treat, \OR_{\m[q_0]}^\control\big), \vec{\bm{1}}-\vec{\bm{q}})
    \right]\\
    & = \big[Q_{\alpha/2}(\hat{\tau}_{\bstrap, \vec{\bm{q}}}^{\min}), Q_{1-\alpha/2}(\hat{\tau}_{\bstrap, \vec{\bm{q}}}^{\max})\big], 
\end{align*}
From \eqref{eq:sub_inter}, this implies that 
$0\notin\big[Q_{\alpha/2+\Omega}(\hat{\tau}_\bstrap), Q_{1-\alpha/2-\Omega}(\hat{\tau}_\bstrap)\big]$. 

From the discussion before, we then have, for any $\eta > 0$, 
\begin{align*}
    &\quad \ \Pr\left\{
    \big(\OR_{\m[q_1]}^\treat, \OR_{\m[q_0]}^\control\big) \notin \mathcal{I}_{\vec{\bm{q}}, \xi}^\alpha 
    \text{ for some } \vec{\bm{q}}\in[0,1]^2
    \right\}\\
    &\le
    \Pr\Big\{
    0\notin\big[Q_{\alpha/2+\Omega}(\hat{\tau}_\bstrap), Q_{1-\alpha/2-\Omega}(\hat{\tau}_\bstrap)\big]
    \Big\}\\
    &\le
    \Pr\Big\{
    0\notin\big[Q_{\alpha/2+\Omega}(\hat{\tau}_\bstrap), Q_{1-\alpha/2-\Omega}(\hat{\tau}_\bstrap)\big], \Omega \le \eta
    \Big\} 
    + 
    \Pr(\Omega > \eta)
    \\
    & 
    \le 
    \Pr\Big\{
    0\notin\big[Q_{\alpha/2+\eta}(\hat{\tau}_\bstrap), Q_{1-\alpha/2-\eta}(\hat{\tau}_\bstrap)\big]
    \Big\} 
    + 
    \Pr(\Omega > \eta). 
\end{align*}
By the same logic as \eqref{eq:Omega_o1} in the proof of Theorem \ref{thm:bootstrap_over_identify} and the validity of the standard bootstrap, letting $n\rightarrow \infty$, we have 
\begin{align*}
    \limsup_{n\rightarrow\infty} \Pr\left\{
    \big(\OR_{\m[q_1]}^\treat, \OR_{\m[q_0]}^\control\big) \notin \mathcal{I}_{\vec{\bm{q}}, \xi}^\alpha 
    \text{ for some } \vec{\bm{q}}\in[0,1]^2
    \right\} \le \alpha + 2\eta. 
\end{align*}
Because $\eta$ can be arbitrarily small, we then have 
\begin{align*}
    \limsup_{n\rightarrow\infty} \Pr\left\{
    \big(\OR_{\m[q_1]}^\treat, \OR_{\m[q_0]}^\control\big) \notin \mathcal{I}_{\vec{\bm{q}}, \xi}^\alpha 
    \text{ for some } \vec{\bm{q}}\in[0,1]^2
    \right\}  \le\alpha.
\end{align*}
This immediately implies that 
\begin{align*}
    &\quad \ \liminf_{n \rightarrow \infty} \Pr\Big\{ \big(\OR^\treat_{\m[q_1]}, \OR^\control_{\m[q_0]}\big) \in \mathcal{I}_{\vec{\bm{q}}, \xi}^\alpha \text{ for all } \vec{\bm{q}} \in [0,1]^2\Big\}\\
    & =1-\limsup_{n\rightarrow\infty} \Pr\left\{
    \big(\OR_{\m[q_1]}^\treat, \OR_{\m[q_0]}^\control\big) \notin \mathcal{I}_{\vec{\bm{q}}, \xi}^\alpha 
    \text{ for some } \vec{\bm{q}}\in[0,1]^2
    \right\}\\
    & \ge 1-\alpha.
\end{align*}
Therefore, Theorem \ref{thm:U_quant_pred_int} holds. 
\end{proof}

\subsection{Proof of Corollary \ref{cor:U_quant_pred_int_max_tr_co}}
\begin{proof}[\bf Proof of Corollary \ref{cor:U_quant_pred_int_max_tr_co}]
From Theorem \ref{thm:U_quant_pred_int}, it suffices to prove that if $(\OR^\treat_{\m[q_1]}, \OR^\control_{\m[q_0]}) \in \mathcal{I}_{\vec{\bm{q}}, \xi}^\alpha \text{ for all } \vec{\bm{q}} \in [0,1]^2$, then $\overline{\OR}_{\m [q]}\in \mathcal{J}_{q, \xi}^\alpha \text{ for all } q \in [0,1]$.

Now suppose that $(\OR^\treat_{\m[q_1]}, \OR^\control_{\m[q_0]}) \in \mathcal{I}_{\vec{\bm{q}}, \xi}^\alpha \text{ for all } \vec{\bm{q}} \in [0,1]^2$. 
By definition, for any $\vec{\bm{q}} \in [0,1]^2$, we have 
    \begin{align*}
        0 \in \Big[\LL_{\alpha/2}\big(\xi(\OR^\treat_{\m[q_1]}, \OR^\control_{\m[q_0]}), \vec{\bm{1}}-\vec{\bm{q}}\big), \UU_{\alpha/2}\big(\xi(\OR^\treat_{\m[q_1]}, \OR^\control_{\m[q_0]}), \vec{\bm{1}}-\vec{\bm{q}}\big)\Big].
    \end{align*}
Consequently, for all $q \in [0,1]$,      
    \begin{align*}
        0 \in \Big[\LL_{\alpha/2}\big(\xi(\OR^\treat_{\m[q]}, \OR^\control_{\m[q]}), (1-q)\vec{\bm{1}}\big), \UU_{\alpha/2}\big(\xi(\OR^\treat_{\m[q]}, \OR^\control_{\m[q]}), (1-q)\vec{\bm{1}}\big)\Big]. 
    \end{align*}
    By definition, we can know that $\LL_{\alpha/2}( (\Lambda_1, \Lambda_0), (1-q)\vec{\bm{1}})$ is nonincreasing in $\Lambda_1$ and $\Lambda_0$, and $\UU_{\alpha/2}((\Lambda_1, \Lambda_0), (1-q)\vec{\bm{1}})$ is nondecreasing in $\Lambda_1$ and $\Lambda_0$. 
    By the definition of $\overline{\OR}_{\m [q]}$, for all $q \in [0,1]$, we must also have 
    \begin{align*}
        0 \in \Big[\LL_{\alpha/2}\big(\xi\overline{\OR}_{\m [q]}\vec{\bm{1}}, (1-q)\vec{\bm{1}}\big), \UU_{\alpha/2}\big(\xi\overline{\OR}_{\m [q]}\vec{\bm{1}}, (1-q)\vec{\bm{1}}\big)\Big], 
    \end{align*}
    which, by definition, is equivalent to
    $\overline{\OR}_{\m [q]}\in \mathcal{J}_{q, \xi}^\alpha$. 
    Therefore, in this case, we must have $\overline{\OR}_{\m [q]}\in \mathcal{J}_{q, \xi}^\alpha \text{ for all } q \in [0,1]$.

    From the above,  Corollary \ref{cor:U_quant_pred_int_max_tr_co} holds. 
\end{proof}

\subsection{Proof of Theorem \ref{thm:ci_tau_pred_whole} and Corollary \ref{cor:ci_tau_whole}}

\begin{proof}[\bf Proof of Theorem \ref{thm:ci_tau_pred_whole}]
Theorem \ref{thm:ci_tau_pred_whole} follows from Theorem \ref{thm:bootstrap_over_identify_pred_set}, by the same logic as the proof of Theorem \ref{thm:ci_tau_pred}. 
We omit the detailed proof for conciseness. 
\end{proof}

\begin{proof}[\bf Proof of Corollary \ref{cor:ci_tau_whole}]
    Similar to the proof of Corollary \ref{cor:ci_tau}, we can show that 
    \begin{align*}
        \Pr\left\{\bmW_{1:n}^\tp \in \mathcal{W}_n(\Lambda, \delta_{n}\right\}
        &= \Pr\Big[\sum_{i=1}^n  \I \{ |W_i^\tp| \leq \log(\Lambda) \} \geq n(1-\delta_{n}) \Big]
        \\
        & \ge 
        \Pr\left\{\Bin(n, 1-\delta) \geq n (1-\delta_{n})\right\}
        = \Pr\left\{\Bin(n, \delta) \leq n \delta_{n} \right\}\\
        & \ge 1-\zeta,
    \end{align*}
    where the first inequality follows from Assumption \ref{asmp:robust_sen_spec_whole}, and the last inequality follows from the definition of  $n \delta_{n}$. 
    Corollary \ref{cor:ci_tau_whole} follows then from Theorem \ref{thm:ci_tau_pred_whole}. 
\end{proof}

\subsection{Proof of Theorems \ref{thm:ci_tau_pred_whole_sharper} and \ref{thm:bootstrap_over_identify_pred_set_sharp}}

\begin{proof}[\bf Proof of Theorem \ref{thm:ci_tau_pred_whole_sharper}]
    Theorem \ref{thm:ci_tau_pred_whole_sharper} follows from Theorem \ref{thm:bootstrap_over_identify_pred_set_sharp}, and its proof is omitted for conciseness. 
\end{proof}

\begin{proof}[\bf Proof of Theorem \ref{thm:bootstrap_over_identify_pred_set_sharp}]

 Similar to \eqref{eq:bootstrap_R_rho}, we consider the following optimization:
    \begin{align}\label{eq:bootstrap_R_rho_pred_set_sharp}
    \breve{\theta}_{H, \bstrap} ^{(j)}= \ & \sup H({\bmtheta}) 
    \\
    \text{subject to} \quad & 
    \sum_{i=1}^n k_{\bstrap i} p_{\bmtheta}(\bmO_{i}, \bmW_i + \bm{R}^\tp_i \hat{{\bmrho}}_{\bstrap} ) =\bm{0},
    \nonumber
    \\
    & (\bmW_1, \bmW_2, \ldots, \bmW_n) \in \mathcal{W}_{n,j}.
    \nonumber
\end{align}
    Let $\Omega=B^{-1}\sum_{b=1}^B\I\{\max_{1\le i\le n}\|\bm{R}_i^\tp\hat{{\bmrho}}_b\|>\varepsilon\}$ where $B\equiv n^n$ is the total number of possible bootstrap resamples. Then we have 
    $\Pr(\Omega>\eta)\rightarrow 0$ as $n\rightarrow \infty$ for any $\eta>0$ as shown in the proof of Theorem \ref{thm:bootstrap_over_identify}.
    In addition, by the same logic as Lemma \ref{lemma:quantile_compare_gen}, we have $Q_{1-\alpha+\Omega}(\hat{\theta}_{H, \bstrap}^{(j)})\ge Q_{1-\alpha}(\breve{\theta}_{H, \bstrap}^{(j)})$ for all $j\in\mathcal{J}_n$.

    Below we prove that, if $H({\bmtheta}_0) > \max_{j\in\mathcal{J}_n}\{Q_{1-\alpha+\eta} ( \hat{\theta}_{H, \bstrap}^{(j)} )\}, \Omega\le\eta$ and $(\bmW_1^\tp, \bmW_2^\tp, \ldots, \bmW_n^\tp) \in \mathcal{W}_n$, then we must have $H({\bmtheta}_0)>Q_{1-\alpha}\{H(\hat{{\bmtheta}}_\bstrap)\}$.
    Now suppose that $\Omega\le\eta, H({\bmtheta}_0) > \max_{j\in\mathcal{J}_n}\{Q_{1-\alpha+\eta} ( \hat{\theta}_{H, \bstrap}^{(j)} )\},$ and $(\bmW_1^\tp, \bmW_2^\tp, \ldots, \bmW_n^\tp) \in \mathcal{W}_n$. 
    There must exist $k\in\mathcal{J}_n$ such that 
    $(\bmW_1^\tp, \bmW_2^\tp, \ldots, \bmW_n^\tp) \in \mathcal{W}_{n,k}$.
    Consequently, 
    \begin{align*}
        H({\bmtheta}_0) > Q_{1-\alpha+\eta} ( \hat{\theta}_{H, \bstrap}^{(k)} )\ge Q_{1-\alpha+\Omega} ( \hat{\theta}_{H, \bstrap}^{(k)} )\ge Q_{1-\alpha}(\breve{\theta}_{H, \bstrap}^{(k)})
        \ge 
        Q_{1-\alpha}\{ H(\hat{{\bmtheta}}_\bstrap) \}, 
    \end{align*}
    where the first inequality follows from our assumption, the second inequality follows from the fact that $\Omega \le \eta$, 
    the third inequality follows from the discussion before, 
    and the last inequality follows from the fact that $(\hat{{\bmtheta}}_{\bstrap}, \bmW_1^\tp, \ldots, \bmW_n^\tp)$ is in the feasible region of \eqref{eq:bootstrap_R_rho_pred_set_sharp} with $j=k$.

    From the discussion before, for any $\alpha'\in (0,1)$ and any $\eta \in (0, \alpha')$, 
    \begin{align*}
        &\limsup_{n\rightarrow \infty} \Pr[ H({\bmtheta}_0) > \max_{j\in\mathcal{J}_n}\{Q_{1-\alpha'+\eta} ( \hat{\theta}_{H, \bstrap}^{(j)} )\}, \  (\bmW_1^\tp, \bmW_2^\tp, \ldots, \bmW_n^\tp) \in \mathcal{W}_n ]\\
        \le&\limsup_{n\rightarrow \infty} \Pr[ H({\bmtheta}_0) > \max_{j\in\mathcal{J}_n}\{Q_{1-\alpha'+\eta} ( \hat{\theta}_{H, \bstrap}^{(j)} )\}, \  (\bmW_1^\tp, \bmW_2^\tp, \ldots, \bmW_n^\tp) \in \mathcal{W}_n , \ \Omega\le\eta]\\
        &+\limsup_{n\rightarrow \infty} \Pr[ H({\bmtheta}_0) > \max_{j\in\mathcal{J}_n}\{Q_{1-\alpha'+\eta} ( \hat{\theta}_{H, \bstrap}^{(j)} )\}, \  (\bmW_1^\tp, \bmW_2^\tp, \ldots, \bmW_n^\tp) \in \mathcal{W}_n , \ \Omega>\eta]\\
        \le&
        \limsup_{n\rightarrow \infty} \Pr[ H({\bmtheta}_0) > Q_{1-\alpha'}\{H(\hat{{\bmtheta}}_\bstrap)\} ]+\limsup_{n\rightarrow \infty} \Pr(\Omega>\eta)\\
        \le&\alpha'.
    \end{align*}
    By letting $\alpha'=\alpha+\eta$, we then have, for any $\alpha \in (0,1)$,  
    \begin{align*}
        \limsup_{n\rightarrow \infty} \Pr[ H({\bmtheta}_0) > \max_{j\in\mathcal{J}_n}\{Q_{1-\alpha} ( \hat{\theta}_{H, \bstrap}^{(j)} )\}, \  (\bmW_1^\tp, \bmW_2^\tp, \ldots, \bmW_n^\tp) \in \mathcal{W}_n ]\le&\alpha+\eta.
    \end{align*}
    Because $\eta$ can be arbitrarily small, this further implies that 
    \begin{align*}
        \limsup_{n\rightarrow \infty} \Pr[ H({\bmtheta}_0) > \max_{j\in\mathcal{J}_n}\{Q_{1-\alpha} ( \hat{\theta}_{H, \bstrap}^{(j)} )\}, \  (\bmW_1^\tp, \bmW_2^\tp, \ldots, \bmW_n^\tp) \in \mathcal{W}_n ]\le&\alpha.
    \end{align*}

    Finally, we prove \eqref{eq:narrow_decomp}. 
    For any $j \in \mathcal{J}_n$, we must have 
    $
    \hat{\theta}_{H, \bstrap}^{(j)} \le 
    \max_{k\in\mathcal{J}_n}\hat{\theta}_{H, \bstrap}^{(k)}, 
    $
    which implies that  
    $
    Q_{1-\alpha} ( \hat{\theta}_{H, \bstrap}^{(j)} ) \le 
    Q_{1-\alpha} \{ \max_{k\in\mathcal{J}_n}\hat{\theta}_{H, \bstrap}^{(k)}\}.
    $
    Thus, $
    \max_{j \in \mathcal{J}_n} \{ Q_{1-\alpha} ( \hat{\theta}_{H, \bstrap}^{(j)} ) \} \le 
    Q_{1-\alpha} \{ \max_{k\in\mathcal{J}_n}\hat{\theta}_{H, \bstrap}^{(k)}\}, 
    $
    i.e., the inequality in \eqref{eq:narrow_decomp} holds. 
    Because $\mathcal{W}_n = \cup_{j \in \mathcal{J}_n} \mathcal{W}_{n,k}$, 
    by definition, 
    we can verify that $\mathcal{W}_n^\varepsilon = \cup_{j \in \mathcal{J}_n} \mathcal{W}_{n,k}^\varepsilon$ 
    and consequently 
    $
    \hat{\theta}_{H, \bstrap} = \max_{j \in \mathcal{J}_n} \hat{\theta}_{H, \bstrap}^{(j)}.
    $
    This then implies the equality in \eqref{eq:narrow_decomp}.

    From the above, Theorem \ref{thm:bootstrap_over_identify_pred_set_sharp} holds. 
\end{proof}

\subsection{Proof of Corollary \ref{cor:test_U_whole} and Theorem \ref{thm:U_quant_pred_int_whole}}\label{sec:proof_simul_whole}

\begin{proof}[\bf Proof of Corollary \ref{cor:test_U_whole}]
    By the same logic as the proof of Corollary \ref{cor:test_U}, 
    Corollary \ref{cor:test_U_whole} follows from Theorem \ref{thm:ci_tau_pred_whole}. 
    In addition, Corollary \ref{cor:test_U_whole} also holds when we replace $[\LL_{\alpha/2}(\xi \Lambda, 1-q), \UU_{\alpha/2}(\xi \Lambda, 1-q)]$ by $[\tilde{\LL}_{\alpha/2}(\xi \Lambda, 1-q), \tilde{\UU}_{\alpha/2}(\xi \Lambda, 1-q)]$ defined as in Theorem \ref{thm:ci_tau_pred_whole_sharper}, and its proof follows from Theorem \ref{thm:ci_tau_pred_whole_sharper}. 
    For conciseness, we omit the detailed proof here. 
\end{proof}

To prove Theorem \ref{thm:U_quant_pred_int_whole}, we need the following lemma. 

\begin{lemma}\label{lemma:quantile_groups}
    Let $a_1\le a_2\le \ldots \le a_{n_1}$ and $b_1\le b_2\le \ldots b_{n_0}$ be two sequences of real numbers. 
    Let $c_1\le c_2 \le \ldots \le c_n$, where $n=n_1+n_0$,  denote the sorted sequence obtained by combining the $a_i$s and $b_i$s. 
    Define $a_0 = b_0 = c_0 = -\infty$ for convenience. 
    Then for any $0\le k \le n$, we have
    \begin{align*}
        c_k = \min_{(i,j): i+j = k, 0\le i \le n_1, 0\le j \le n_0}  \max\{a_i, b_j\} . 
    \end{align*}
\end{lemma}
\begin{proof}[\bf Proof of Lemma \ref{lemma:quantile_groups}]
Lemma \ref{lemma:quantile_groups} holds obviously when $k=0$. Below we consider only the case where $1\le k \le n$. 

First, we prove that 
$c_k \le \min_{(i,j): i+j = k, 0\le i \le n_1, 0\le j \le n_0}  \max\{a_i, b_j\}$ by contradiction. 
Suppose that
$c_k > \min_{(i,j): i+j = k, 0\le i \le n_1, 0\le j \le n_0}  \max\{a_i, b_j\}$.
Then there must exist $(i,j)$ such that $i+j = k, 0\le i \le n_1, 0\le j \le n_0$, and $c_k > \max\{a_i, b_j\}$. 
This implies that 
\begin{align*}
    \sum_{l=1}^n \I(c_l < c_k) 
    & = 
    \sum_{l=1}^{n_1} \I(a_l < c_k) 
    +
    \sum_{l=1}^{n_0} \I(b_l < c_k) 
    \ge 
    \sum_{l=1}^{n_1} \I(a_l \le a_i) 
    +
    \sum_{l=1}^{n_0} \I(b_l \le b_j) \\
    & \ge i + j = k. 
\end{align*}
However, $\sum_{l=1}^n \I(c_l < c_k)$ is at most $k-1$ by definition, leading to a contradiction. 
Thus, we must have $c_k \le \min_{(i,j): i+j = k, 0\le i \le n_1, 0\le j \le n_0}  \max\{a_i, b_j\}$.

Second, we prove that $c_k \ge \min_{(i,j): i+j = k, 0\le i \le n_1, 0\le j \le n_0}  \max\{a_i, b_j\}$. 
Without loss of generality, we assume that $c_k$ belongs to  $\{a_i\}_{i=1}^{n_1}$. 
Let $s$ be the largest integer in $[1, \min\{ n_1, k\}]$ such that $c_k = a_s$.
If $s = k$, then we have 
$c_k \ge \max\{a_s, b_0\} \ge \min_{(i,j): i+j = k, 0\le i \le n_1, 0\le j \le n_0}  \max\{a_i, b_j\}$. 
Otherwise, 
$s < k$. 
Then we have either $s < \min\{ n_1, k\}$ and $s = n_1$. In both cases, we can verify that $\sum_{l=1}^{n_1} \I(a_l \le a_s) = s$. This implies that   
\begin{align*}
    k & \le \sum_{l=1}^n \I(c_l \le c_k)
    =
    \sum_{l=1}^{n_1} \I(a_l \le a_s) 
    +
    \sum_{l=1}^{n_0} \I(b_l \le c_k)
    = s + \sum_{l=1}^{n_0} \I(b_l \le c_k), 
\end{align*}
which further impies 
$
    n_0 \ge \sum_{l=1}^{n_0} \I(b_l \le c_k) \ge k-s. 
$
Thus, we must have $b_{k-s} \le c_k$ and $k-s\le n_0$. 
Consequently, 
$c_k \ge \max\{a_s, b_{k-s}\} \ge \min_{(i,j): i+j = k, 0\le i \le n_1, 0\le j \le n_0}  \max\{a_i, b_j\}$.

From the above, we can then derive Lemma \ref{lemma:quantile_groups}. 
\end{proof}

\begin{proof}[\bf Proof of Theorem \ref{thm:U_quant_pred_int_whole}]
Define $\tilde{\mathcal{I}}_{q, \xi}^\alpha = \xi^{-1} \tilde{\mathcal{I}}_{q}^\alpha = \{\xi^{-1}\Lambda: \Lambda \in \tilde{\mathcal{I}}_{q}^\alpha\}$ for $\xi \ge 1$, where 
\begin{align*}
        \tilde{\mathcal{I}}_{q}^\alpha \equiv \big\{ \Lambda \in [1, \infty]: 
        0 \in [\tilde{\LL}_{\alpha/2}(\Lambda, 1-q), \tilde{\UU}_{\alpha/2}(\Lambda, 1-q)] \big\}.
\end{align*}
From Theorem \ref{thm:ci_tau_pred_whole_sharper}, 
we can know that $\tilde{\mathcal{I}}_{q}^\alpha \subset \mathcal{I}_{q}^\alpha$. 
Thus, to prove Theorem \ref{thm:U_quant_pred_int_whole}, it suffices to prove that it holds with  $\mathcal{I}_{q, \xi}^\alpha$ replaced by 
$\tilde{\mathcal{I}}_{q, \xi}^\alpha$. 
Below we prove Theorem \ref{thm:U_quant_pred_int_whole} with  $\mathcal{I}_{q, \xi}^\alpha$ there replaced by 
$\tilde{\mathcal{I}}_{q, \xi}^\alpha$.

Now we suppose that $(\OR^\treat_{\m[q_1]}, \OR^\control_{\m[q_0]}) \in \mathcal{I}_{\vec{\bm{q}}, \xi}^\alpha \text{ for all } \vec{\bm{q}} \in [0,1]^2$. 
By definition, for any $\vec{\bm{q}} \in [0,1]^2$, we have 
    \begin{align}\label{eq:interval_all_q_vec}
        0 \in \Big[\LL_{\alpha/2}\big(\xi(\OR^\treat_{\m[q_1]}, \OR^\control_{\m[q_0]}), \vec{\bm{1}}-\vec{\bm{q}}\big), \UU_{\alpha/2}\big(\xi(\OR^\treat_{\m[q_1]}, \OR^\control_{\m[q_0]}), \vec{\bm{1}}-\vec{\bm{q}}\big)\Big].
    \end{align}
Consider any $\delta \in [0,1]$. 
From Lemma \ref{lemma:quantile_groups}, there exist $\vec{\bm{\delta}} = (\delta_1, \delta_0)$ such that $\OR_{\m[\delta]}=\max\{\OR_{\m[\delta_1]}^\treat, \OR_{\m[\delta_0]}^\control\}$, 
$n_1\delta_1 + n_0 \delta_0 = \lceil n \delta
 \rceil$, 
and $n_z\delta_z$ is an integer between 
$0$ and $\min\{ n_z, \lceil n \delta \rceil \}$ for $z=0,1$. 
We can then verify that 
$n_1 (1-\delta_1)$ is an integer between 
$\max\{0, \lfloor n (1- \delta) \rfloor - n_0\}$ and 
$\min\{ n_1, \lfloor n (1-\delta) \rfloor \}$. 
By definition, we know that, with $i = n_1 (1-\delta_1)$, 
\begin{align*}
     \tilde{\LL}_{\alpha/2}(\xi \OR_{\m[\delta]}, 1-\delta) 
     & \le 
      \LL_{\alpha/2}\left(\xi \OR_{\m[\delta]} \cdot \vec{\bm{1}}, \left( \frac{i}{n_1}, \frac{\lfloor n(1-\delta)\rfloor-i}{n_0} \right)\right)
    = 
     \LL_{\alpha/2}\left(\xi \OR_{\m[\delta]}\cdot\vec{\bm{1}}, \vec{\bm{1}} - \vec{\bm{\delta}} \right). 
\end{align*}
Because $\OR_{\m[\delta]} = \max\{\OR_{\m[\delta_1]}^\treat, \OR_{\m[\delta_0]}^\control\}$ is greater than or equal to both $\OR_{\m[\delta_1]}^\treat$ and $\OR_{\m[\delta_0]}^\control$, we further have 
\begin{align*}
     \tilde{\LL}_{\alpha/2}(\xi \OR_{\m[\delta]}, 1-\delta) 
     & \le  
     \LL_{\alpha/2}\left(\xi \OR_{\m[\delta]}\cdot\vec{\bm{1}}, \vec{\bm{1}} - \vec{\bm{\delta}} \right)
     \le 
      \LL_{\alpha/2}\left(\xi(\OR^\treat_{\m[\delta_1]}, \OR^\control_{\m[\delta_0]}), \vec{\bm{1}} - \vec{\bm{\delta}} \right).
\end{align*}
Analogously, we can show that 
\begin{align*}
    \tilde{\UU}_{\alpha/2}(\xi \cdot\OR_{\m[\delta]}, 1-\delta) \ge \UU_{\alpha/2}\big(\xi(\OR^\treat_{\m[\delta_1]}, \OR^\control_{\m[\delta_0]}), \vec{\bm{1}}-\vec{\bm{\delta}}\big).
\end{align*}
From \eqref{eq:interval_all_q_vec}, we then have  
\begin{align*}
    0\in \left[\tilde{\LL}_{\alpha/2}(\xi \cdot\OR_{\m[\delta]}, 1-\delta), \tilde{\UU}_{\alpha/2}(\xi \cdot\OR_{\m[\delta]}, 1-\delta)\right],
\end{align*}
which, by definition, is equivalent to $\OR_{\m[\delta]}\in\tilde{\mathcal{I}}_{\delta, \xi}^\alpha$.

From the above, we then have 
\begin{align*}
    \Pr\{ \OR_{\m[q]} \in \tilde{\mathcal{I}}_{q, \xi}^\alpha \text{ for all } q \in [0,1]\} \ge \Pr\Big\{ \big(\OR^\treat_{\m[q_1]}, \OR^\control_{\m[q_0]}\big) \in \mathcal{I}_{\vec{\bm{q}}, \xi}^\alpha \text{ for all } \vec{\bm{q}} \in [0,1]^2\Big\}.
\end{align*}
From Theorem \ref{thm:U_quant_pred_int}, we can then derive Theorem \ref{thm:U_quant_pred_int_whole} with  $\mathcal{I}_{q, \xi}^\alpha$ there replaced by 
$\tilde{\mathcal{I}}_{q, \xi}^\alpha$. 
From the discussion at the beginning of this proof, we can also derive Theorem \ref{thm:U_quant_pred_int_whole}. 
\end{proof}

\end{document}